\documentclass[%
 reprint,%
 amssymb, amsmath,%
 aip,cha,%
]{revtex4-1}

\usepackage{bm}%
\usepackage[colorlinks=true,linkcolor=blue]{hyperref}%
\usepackage{amsmath,amssymb,amsthm}
\usepackage{amscd}
\usepackage[usenames,dvipsnames]{color}
\usepackage{graphicx}
\usepackage[all]{xy}

\usepackage[all]{xy}
\usepackage{graphicx}
\usepackage{soul}

\newcommand{\nablaF}{\nabla F}

\newcommand{\telque}{{\,;\,}}

\newcommand{\frakF}{\mathsf{F}}

\newcommand{\bbR}{\mathbb{R}}
\newcommand{\bbS}{\mathbb{S}}

\newcommand{\bbK}{\mathbb{K}}
\newcommand{\bbZ}{\mathbb{Z}}
\newcommand{\bbC}{\mathbb{C}}

\newcommand{\bbN}{\mathbb{N}}

\newcommand{\WF}{\mathrm{WF}}

\newcommand{\dotT}{\dot{T}}

\newcommand{\supp}{{\mathrm{supp}\,}}

\newcommand{\barF}{\overline{F}}

\newcommand{\calA}{{\mathcal A}}

\newcommand{\calL}{{\mathcal L}}

\newcommand{\Acal}{{\mathcal A}}

\newcommand{\calD}{{\mathcal D}}
\newcommand{\calE}{{\mathcal E}}

\newcommand{\calI}{{\mathcal I}}

\newcommand{\Ocal}{{\cal O}}

\newcommand{\bcam}{}
\newcommand{\ecam}{}

\newcommand{\C}{\mathbb{C}}

\newtheorem{prop}{Proposition}[section]
\newtheorem{thm}[prop]{Theorem}
\newtheorem{cor}[prop]{Corollary}

\newtheorem{lem}[prop]{Lemma}
\newtheorem{dfn}[prop]{Definition}
\expandafter\ifx\csname package@font\endcsname\relax\else
 \expandafter\expandafter
 \expandafter\usepackage
 \expandafter\expandafter
 \expandafter{\csname package@font\endcsname}%
\fi
\begin{document}

\title{Properties of field functionals
and characterization of local functionals}

\author{Christian Brouder}%
\email{christian.brouder@upmc.fr}
\affiliation{Sorbonne Universit\'es, UPMC Univ. Paris 06, UMR CNRS 7590,
Institut de Min\'eralogie, de Physique des Mat\'eriaux et de Cosmochimie,
Mus\'eum National d'Histoire Naturelle, IRD UMR 206,
4 place Jussieu, F-75005 Paris, France.}

\author{Nguyen Viet Dang}%
\affiliation{Institut Camille Jordan (UMR CNRS 5208)
  Universit\'e Claude Bernard Lyon 1,
  B\^at. Braconnier, 43 bd du 11 Novembre 1918,
  69622 Villeurbanne Cedex, France.}

\author{Camille Laurent-Gengoux}
\affiliation{Institut Elie Cartan de Lorraine (IECL)
  Universit\'e de Lorraine, B\^at. A, Ile du Saulcy 
  F-57045 Metz cedex 1,  France.}

\author{Kasia Rejzner}
\affiliation{Department of Mathematics,
  University of York, Heslington, York YO10~5DD, United Kingdom.}

\date{\today}%
\revised{\today}%

\begin{abstract}
Functionals (i.e. functions of functions) are widely used
in quantum field theory and solid-state physics.
In this paper, functionals are given
a rigorous mathematical framework
and their main properties
are described.
The choice of the proper space of test functions (smooth functions)
and of the relevant concept of differential
(Bastiani differential) are discussed.

The relation between the multiple derivatives
of a functional and the corresponding distributions
is described in detail.
It is proved that, in a neighborhood of every
test function, the support of a smooth functional
is uniformly compactly supported and the order of the corresponding
distribution is uniformly bounded.
Relying on a recent work by Yoann Dabrowski,
several spaces of functionals are furnished with
a complete and nuclear topology.
In view of physical applications, it is shown that most
formal manipulations can be given a rigorous meaning.

A new concept of local functionals is proposed and
two characterizations of them are given: the first one uses
the additivity (or Hammerstein) property, the
second one is a variant of Peetre's theorem.
Finally, the first step of a cohomological approach
to quantum field theory is carried out by proving a
global Poincar\'e lemma and defining multi-vector fields
and graded functionals within our framework.
\end{abstract}

\maketitle

\tableofcontents

\section{Motivation}
Functionals (i.e. functions of functions) 
are mathematical objects successfully applied
in many areas of physics.
Since Schwinger's ground-breaking
papers~\cite{Schwinger,SchwingerPNAS2},
Green functions of quantum field theory are obtained 
as functional derivatives of the generating 
functional $Z(j)$ with respect to the functions
$j$ (external sources).
In solid-state and molecular physics, the exchange and correlation
potential of density functional theory is computed from 
the functional derivative of the total energy
$E(\rho)$ with respect to the electron density 
$\rho$~\cite{Kohn-65,Engel-Dreizler}.
In perturbative algebraic quantum field theory (pAQFT), the
observables are functionals $F(\varphi)$ of the
classical field $\varphi$~\cite{Rejzner-book}. 
This formulation was possible due to the crucial result \cite{Dutsch06} 
that allowed to realize abstract quantum fields as concrete functionals
on the space of classical configurations. This viewpoint is not only
simplifying computations, but also allows to construct new perturbative
and exact models of QFT's~\cite{Brunetti-13-QG,Bahns-17}. 
It is, therefore, crucial to understand functional analytic properties of classical functionals to be able to use these in quantization and obtain even more models. The importance of this endeavour is justified by the fact that presently we do not know any exact interacting QFT models in 4 spacetime dimensions.

Functionals are also used in pure mathematics,
for example loop space cohomology~\cite{Lempert-04}
and infinite dimensional integrable systems:
the hierarchy of commuting Hamiltonians for the Korteweg de Vries
equation
is for instance all made of functionals \cite{dickey2003soliton}.

In all these fields, the concept of locality is crucial:
the Lagrangian of quantum field theory is local and
the counterterms of the renormalization process have
to be local, the approximations of $E(\rho)$ used
in practice are local and it is an open question
whether the true density functional $E(\rho)$ is
local or not.
Therefore, it is crucial to determine precisely
what is meant by a \emph{local functional}.
According to the standard 
definition~\cite{Lee-90,Wald-GR,Stasheff-98,Barnich-00},
if $\varphi$ is a classical field
(i.e. a smooth section of a vector bundle over $M$
and we momentarily consider $M=\bbR^d$ for
notational convenience),
then a functional $F(\varphi)$
is local if it is of the form
\begin{eqnarray}
F(\varphi) &=& \int_{\bbR^d} dx f\big(x,\varphi(x),\partial_\mu \varphi(x),
  \dots,\partial_{\mu_1\dots\mu_k} \varphi(x)\big).
\label{deflocal}
\end{eqnarray}
where $f$
is a smooth compactly supported function with a finite number of
arguments.

However, this definition of local functionals is not very handy
in practice because it is global and sometimes too restrictive.
For example, general relativity has no local gauge-invariant observables 
in the sense of Eq.~(\ref{deflocal}), whereas it has local gauge-invariant
observables when the concept of locality is 
slightly generalized, as discussed 
in~\cite{Brunetti-13-QG} (see also the parallel work \cite{Khavkine-15}).
Note that the concept of locality presented in 
the present paper gives a proper
topological framework for local functionals as understood 
by~\cite{Fredenhagen-11,Brunetti-09,Brunetti-09-Bar}. 

The present paper \bcam puts forth \ecam the following formulation of the  concept of locality:
\begin{dfn}
\label{defourlocal}
\bcam Let $M$ be a manifold\footnote{In this article, all manifolds shall be assumed to be paracompact.}. \ecam
Let $U$ be an open subset of $C^\infty(M)$.
A smooth functional $F:U\to \bbK$ is said to be
\emph{local}
if, for every $\varphi\in U$, there is a neighborhood 
$V$ of $\varphi$, an integer $k$,
an open subset $\mathcal{V}\subset J^kM$ and a smooth
function $f\in C^\infty(\mathcal{V})$ such that
$x\in M\mapsto f(j^k_x\psi)$ is supported in a compact subset $K\subset M$ and
\begin{eqnarray*}
F(\varphi+\psi) &=& F(\varphi)+\int_M f(j^k_x\psi) dx,
\end{eqnarray*}
whenever $\varphi+\psi\in V$
and where $j^k_x\psi$ denotes the $k$-jet of $\psi$ at $x$.
\end{dfn}
In other words, we require $F$ to be local in the
sense of Eq.~(\ref{deflocal}), but only around each
$\varphi\in U$ because the 
integer $k$ and the function $f$ can depend
on the neighborhood $V$. In short, our local functionals 
are local in the ``traditional sense'', but only locally in the configuration space
(i.e. in a neighborhood of each $\varphi$).
We do not need global locality to apply variational
methods and derive Euler-Lagrange equations.
We will show by
exhibiting an example that this concept
of locality is strictly more general than the traditional one.
Our first main result is a simple characterization
of local functionals 
in the sense of Def.~1.1:
\begin{thm}
Let $U$ be an open subset of $C^\infty(M)$. A
smooth functional $F:U\to \bbK$
(where $\bbK=\bbR$ or $\bbC$)
is local if and only if
\begin{enumerate}
\item $F$ is additive  
(i.e. it satisfies
   $F(\varphi_1+\varphi_2+\varphi_3) = F(\varphi_1+\varphi_2) 
  + F(\varphi_2+\varphi_3)
  - F(\varphi_2)$ whenever $\supp\,\varphi_1\cap \supp\,\varphi_3=\emptyset$)
\item For every $\varphi\in U$, the differential
   $DF_\varphi$ of $F$ at $\varphi$ is a distribution
   with empty wave front set. Thus, it can be represented
   by a function $\nablaF_\varphi\in \calD(M)$ (with $\calD(M)$
   the space of compactly supported smooth functions on $M$, 
   i.e. ``test functions").
\item The map $U\to \calD(M)$  defined by
  $\varphi\mapsto \nablaF_\varphi$ is smooth
(in the sense of Bastiani).
\end{enumerate}
\end{thm}
Our characterization of locality is inspired by
the microlocal functionals proposed by Brunetti, 
Fredenhagen and Ribeiro~\cite{Brunetti-12}. 
However, the proof of their Proposition~2.3.12
is not complete because the
application of the Fubini theorem and the second
use of the fundamental theorem of calculus are 
not justified. Our condition 3 solves that problem.
On the other hand, we do not need their assumption
that $F$ is compactly supported.

Let us stress that the notion of locality is quite subtle and depends strongly 
on the functional analytic setting. A functional characterization of a notion 
of local functionals on \textbf{measurable functions} might not be valid 
anymore when applied to \textbf{smooth functions} as is shown by the simple 
counterexample of Section~\ref{nonlocalpartiallyadditive}.
We also make a conjecture as to how to generalize our main result to multi-vector fields and 
graded functionals, which is crucial for 
a rigorous version of the Batalin-Vilkovisky
approach to gauge field theory and 
quantum gravity.
The second main result is a proof of the 
global Poincar\'e lemma (in
our context), which is crucial to set up the
BRST and variational complexes.
The last one is another characterization of
local (and multilocal) functionals in the
form of a Peetre's theorem.

Along the way to these results, we prove interesting properties of general functionals that we briefly  describe now.
In section~2, we explain why we choose test functions that are 
only smooth 
instead of smooth and compactly supported,
we describe the topology of the space of test functions and we present the concept of Bastiani differentiability
and its main properties. In section~3, we show that a smooth functional
is locally compactly supported (i.e. in a neighborhood
of every test function), we prove that the 
$k$th derivative of a functional defines a 
continuous family of distributions whose order is
locally bounded.
Section~4, which relies heavily on Dabrowski's 
work~\cite{Dabrowski-14,Dabrowski-14-2},
describes in detail a nuclear and complete
topology on several spaces of functionals used in
quantum field theory.
Section~5 discusses the concept of
additivity which characterizes local functionals.
Sections~6 and 7 prove the main results discussed above.
Note that the present paper has a somewhat foundational character,
in as much as the choice of test-functions, additivity property
and differential
are carefully justified from the physical and mathematical 
points of view. It contributes to the formulation 
of a mathematically rigorous basis on which the quantum field theory of
gauge fields and gravitation can be built.

Note also that this paper aims at both 
functional analysts and theoretical physicists.
Because of this dual readership, the proofs are
often more detailed than what would be required
for experts in functional analysis.

\section{Functionals and their derivatives}
To set up a mathematical definition of functionals,
we need to determine precisely which space of test functions
(i.e. classical fields and sources) we consider
and what we mean by a functional derivative.

\subsection{The space of classical fields}
Propagators and
Green functions of quantum fields in flat spacetimes
are tempered distributions~\cite{Epstein,Wightman}
and the corresponding test functions are rapidly
decreasing. 
Tempered distributions are computationaly convenient
because they have Fourier transforms.
However, tempered distributions cannot be canonically
extended to curved spacetimes (i.e. Lorentzian 
smooth manifolds) because the rapid decrease of test functions
at infinity is controlled by some Euclidian distance 
which is not canonically defined on
general spacetime manifolds~[\onlinecite[p.~339]{Martellini-82}].

The most natural spaces of test functions on a
general spacetime $M$ are the space
$C^\infty(M)$ of real valued smooth functions on $M$
and its subspace $\calD(M)$  of compactly supported
functions.  These two spaces are identical when $M$
is compact, but physically relevant spacetimes are not
compact because they are globally hyperbolic, and a choice
must be made.

In this paper, we choose $C^\infty(M)$ (or the set
$\Gamma(M,B)$ of smooth sections of a vector bundle $B$).
There is a strong physical reason for this~\cite{Dutsch}:
in the quantization process we must be able to deal with 
on-shell fields $\varphi$, that are smooth solutions to normally 
hyperbolic equations and as such cannot be compactly supported.
Therefore, the domain of the functionals can be $C^\infty(M)$
but not $\calD(M)$.
There are also good mathematical arguments for this choice:
In particular, $C^\infty(M)$ is a Fr\'echet space
and its pointwise multiplication is continuous~[\onlinecite[p.~119]{Schwartz-66}].
Moreover, the Fr\'echet property of $C^\infty(M)$ saves us the
trouble of distinguishing Bastiani 
from convenient differentiability
which is treated in Ref.~\onlinecite{KrieglMichor}.

The choice of $C^\infty(M)$ has, however, several drawbacks:
(i) Since smooth functions are generally not integrable over $M$, 
the Lagrangian density $\calL(\varphi)$ must be multiplied
by a smooth compactly supported function $g$ so that
$\calL(\varphi) g$ is integrable over $M$~\cite{Stueckelberg-51}.
As a result, long-range interactions are suppressed and 
infrared convergence is enforced.
This simplifies the problem but
makes it difficult to deal with the physics of infrared divergence.
(ii) The function $g$ breaks the diffeomorphism invariance
of the Einstein-Hilbert action. 
(iii) The effect
of a perturbation $\varphi+\epsilon\psi$  is easier to deal with when
$\psi$ is compactly supported because it avoids the presence of
boundary terms. This problem can be solved by considering
$C^\infty(M)$ as a manifold modeled on 
$\calD(M)$~\cite{Brunetti-12,Rejzner-book}, but this is an additional
complication.

\subsection{Locally convex spaces}
\label{LCS}

The spaces of test functions and functionals
considered in the paper are all locally convex.
The most pedagogical introduction to locally
convex spaces is probably 
Horvath's book~\cite{Horvath}, so we refer the reader to 
it for more details.

We describe now the topology of the spaces of test functions
that we use. For the space of smooth test functions
$C^\infty(\bbR^d)$, the topology  is 
defined by the seminorms
\begin{eqnarray}
\pi_{m,K}(f) &=& \sup_{x\in K} \sup_{|\alpha|\le m}
  |\partial^\alpha f(x)|,
\label{pimKdef}
\end{eqnarray}
where $f\in C^\infty(\bbR^d)$, $m$ is an integer,
$K$ is a compact subset
of $\bbR^d$, $\alpha=(\alpha_1,\dots,\alpha_d)$
is a $d$-tuple of nonnegative integers,
with $|\alpha|=\alpha_1+\dots+\alpha_d$
and $\partial^\alpha=\partial^{\alpha_1}_1\dots
\partial^{\alpha_d}_d$,
with $\partial_i=\partial/\partial x^i$ the
derivative with respect to the $i$-th coordinate of 
$x$~[\onlinecite[p.~88]{Schwartz-66}].

If $U$ is open in $\bbR^d$, we denote by $C^\infty(U)$
the space of all functions defined on $U$ which possess
continuous partial derivatives of all orders. We equip
$C^\infty(U)$ with the topology defined by the seminorms
$\pi_{m,K}$ where $K$ runs now over the compact 
subsets of $U$~[\onlinecite[p.~89]{Horvath}].
For every open set $U\subset \bbR^d$, the 
space $C^\infty(U)$ is Fr\'echet, 
reflexive, Montel, barrelled~[\onlinecite[p.~239]{Horvath}],
bornological~[\onlinecite[p.~222]{Horvath}] and
nuclear~[\onlinecite[p.~530]{Treves}].

We define now $C^\infty(M)$, where
$M$ is a 
$d$-dimensional manifold
(tacitly smooth, Hausdorff, paracompact and orientable)
described by charts $(U_\alpha,\psi_\alpha)$.
If for every $U_\alpha\subset M$ we are given a 
smooth function
$g_\alpha\in C^\infty(\psi_\alpha(U_\alpha))$ such
that $g_\beta=g_\alpha\circ\psi_\alpha\circ\psi_\beta^{-1}$
on $\psi_\beta(U_\alpha\cap U_\beta)$, we call the
system $g_\alpha$ a smooth function $g$ on
$M$. The space of smooth functions on $M$
is denoted by $C^\infty(M)$~[\onlinecite[p.~143]{HormanderI}].
This definition is simple but 
to describe the topological properties of 
$C^\infty(M)$ the following more conceptual
definition is useful.

Let $M$ be a manifold and $B\to M$ a
smooth vector bundle of rank $r$ over $M$ with projection $\pi$. Let 
$E=\Gamma(M,B)$ be the space of smooth sections of $B$ 
equipped with the following topology~\cite{KrieglMichor}:
\begin{dfn}
The topology on $\Gamma(M,B)$ is defined as
follows. Choose a chart $(U_\alpha,\psi_\alpha)_\alpha$
and a trivialization map
$\Phi_\alpha:\pi^{-1}(U_\alpha)\to \Omega\times\bbR^r$, where
$\Omega$ is a fixed open set in $\mathbb{R}^d$.
Then the map $\Phi_\alpha$ allows to identify 
$\Gamma(U_\alpha,B)$ with
$C^\infty(\Omega,\bbR^r)$ by
$\Phi_\alpha:\pi^{-1}{U_\alpha }\to \Omega\times\bbR^d$
such that
\begin{eqnarray*}
\Phi_\alpha(x,s(x))=
(\psi_\alpha(x),K_\alpha(s)(\psi_\alpha(x))),
\end{eqnarray*}
where 
\begin{eqnarray*}
K_\alpha:s\in\Gamma(U_\alpha,B)&\mapsto &K_\alpha(s)\in
C^\infty(\Omega,\mathbb{R}^r).
\end{eqnarray*}
The topology on $\Gamma(M,B)$ is the weakest
topology making all the maps $K_\alpha$ continuous.
\end{dfn}
This topology does not depend on the choice
of charts or trivialization maps~[\onlinecite[p.~294]{KrieglMichor}].
To interpret this topology, denote by 
$\rho_\alpha:s\in\Gamma(M,B)\mapsto s|_{U_\alpha}\in
\Gamma(U_\alpha,B)$ the restriction
map of sections on open sets of our open cover 
$(U_\alpha)_{\alpha\in I}$ of $M$.
The space $\Gamma(M,B)$ fits into the following 
complex of vector spaces
\begin{eqnarray}
0 &\to& \Gamma(M,B) \overset{\small{(\rho_\alpha)_{\alpha\in I}}}
{\longrightarrow} \prod_{\alpha\in I} 
\Gamma(U_\alpha,B)\simeq 
\prod_{\alpha\in I}C^\infty(\Omega,\mathbb{R}^r)
\nonumber\\&&
\overset{{\small(\rho_\alpha-\rho_\beta)_{\alpha\neq \beta}}}
{\longrightarrow}
\prod_{(\alpha\neq \beta)\in I^2}
\Gamma(U_\alpha\cap U_\beta,B). 
\end{eqnarray}

The topology on $\Gamma(U_\alpha,B)$ is given by the isomorphism 
$\Gamma(U_\alpha,B)\simeq C^\infty(\Omega,\mathbb{R}^r)$ hence 
it is nuclear Fr\'echet. The countable products 
$\prod_{\alpha\in I} \Gamma(U_\alpha,B)$ and 
$\prod_{(\alpha,\beta)\in I^2}\Gamma(U_\alpha\cap U_\beta,B)$ are 
therefore nuclear Fr\'echet. 
For every pair $(\alpha,\beta)$ of distinct elements of $I$, 
the difference of restriction maps
$\rho_\alpha-\rho_\beta$ is continuous and the topology on 
$\Gamma(M,B)$
is the weakest topology which makes the above complex
topological, which implies that it is nuclear Fr\'echet as the 
kernel of $\prod_{\alpha\not=\beta}\rho_\alpha-\rho_\beta$.

Locally convex spaces are very versatile and they
are the proper framework to define spaces of smooth functionals,
i.e. smooth functions on a space of functions (or sections of a bundle).
The first step towards this goal is to provide a rigorous
definition of functional derivatives.

\subsection{Functional derivatives}

To define the space of functionals, we consider
the main examples $Z(j)$ and $F(\varphi)$.
These two functionals send smooth classical fields
to $\bbK$, where $\bbK=\bbR$ or $\bbK=\bbC$.
Moreover, functional derivatives of $Z$ and $F$
of all orders are required to obtain the Green functions 
from $Z(j)$ and to quantize the product $F(\varphi)G(\varphi)$.
Therefore, we must define the derivative of a
function $f:E\to \bbK$, where $E$ is the space of
classical fields.

It will be useful to 
generalize the problem to functions $f$ between
arbitrary locally convex spaces $E$ and $F$.
To define such a derivative we start from
\begin{dfn}
Let $U$ be an open subset of a Hausdorff locally convex space $E$
and let $f$ be a map from $U$ to a Hausdorff locally convex space
$F$. Then $f$ is said to have a derivative at $x\in U$
in the direction of $v \in E$ if the following limit
exists\footnote{Since an open set is absorbing~[\onlinecite[p.~80]{Horvath}],
for every $x\in U$ and every $v\in E$ there is an $\epsilon>0$ such that
$x+tv\in U$ if $|t|<\epsilon$. Thus, $f(x+tv)$ is well defined
for every $t$ such that $|t|<\epsilon$.
}
\begin{eqnarray*}
Df_x(v) &:=&  \lim_{t\to 0}\frac{f(x+tv)-f(x)}{t}.
\end{eqnarray*}
\end{dfn}

One can also consider the same definition restricted to 
$t>0$~\cite{Shapiro-90}.
A function $f$ is said to have a 
\emph{G\^ateaux differential}~\cite{Ladas,Kurdila} 
(or a \emph{G\^ateaux variation}~\cite{Sagan})
at $x$ if $Df_x(v)$ exists for every $v\in E$.
However, this definition is far too weak for our
purpose because $Df_x(v)$ is generally neither
linear nor continuous in $v$ and it can be linear 
without being continuous and continuous without
being linear~[\onlinecite[p.~7]{Yamamuro}]. 
Therefore, we will use a stronger definition, namely 
\emph{Bastiani differentiability}~\cite{Bastiani-64}, which is the fundamental concept of
differentiability used throughout the paper:
\begin{dfn}
\label{defC1}
Let $U$ be an open subset of a Hausdorff locally convex space $E$
and let $f$ be a map from $U$ to a Hausdorff locally convex space
$F$. Then $f$ is \emph{Bastiani differentiable} on $U$
(denoted by $f\in C^1(U)$) if $f$ has a 
G\^ateaux differential at every $x\in U$
and the map $Df: U\times E\to F$
defined by $Df(x,v)=Df_x(v)$ is continuous on $U\times E$.
\end{dfn}
With this definition, most of the properties used in physics
textbooks (e.g. chain rule, Leibniz rule, linearity) are mathematically valid.

\subsubsection{Examples}
\label{Example-sect}

We shall consider several examples of functions from
$C^\infty(M) $ to $\bbR$ or $\bbC$, where $M=\bbR^d$:
\begin{eqnarray*}
F(\varphi) &=& \int_M f(x) \varphi^n(x) dx,\\
G(\varphi) &=& \int_{M^n} g(x_1,\dots,x_n) \varphi(x_1)\dots\varphi(x_n)
  dx_1\dots dx_n,\\
H(\varphi) &=& \sum_{\mu\nu} \int_M g^{\mu\nu}(x)
   h(x) \partial_\mu \varphi(x) 
   \partial_\nu \varphi(x) dx,\\
I(\varphi) &=& \int_M f(x) e^{\varphi(x)} dx,\\
J(\varphi) &=& e^{\int_M f(x) \varphi(x) dx},\\
K(\varphi) &=& \int_M f(x) \sin\big(\varphi(x)\big) dx,
\end{eqnarray*}
where $f$, $g$, and $h$ are smooth compactly supported
functions and where $g$ is a symmetric function of its arguments.
Further examples can be found 
in~\cite{Brunetti-09,Hamilton-82}.
It is immediate to check that
\begin{eqnarray*}
D F_\varphi (v) &=& n \int_M g(x) \varphi^{n-1}(x) v (x) dx,\\
DG_\varphi (v) &=& 
n \int_{M^n} g(x_1,\dots,x_n)
  \varphi(x_1)\dots \varphi(x_{n-1})
\\&&v(x_n) dx_1 \dots dx_n, \\
DH_\varphi (v) &=&
2\sum_{\mu\nu} \int_M g^{\mu\nu}(x)
   h(x) \partial_\mu \varphi(x) 
   \partial_\nu v(x) dx,\\
DI_\varphi(v) &=& \int_M f(x) e^{\varphi(x)} v(x)dx,\\
DJ_\varphi(v) &=& 
 e^{\int_M f(x) \varphi(x) dx} \, \int_M f(x) v(x) dx,\\
DK_\varphi(v) &=& \int_M f(x) \cos\big(\varphi(x)\big) v(x)dx.
\end{eqnarray*}

\subsubsection{Historical remarks}
Definition~\ref{defC1} is due to Bastiani~\cite{BastianiPhD,Bastiani-64}
and looks quite natural.  In fact, it is not so. 
For a long time, many different approaches were tried.
For any reasonable definition of differentiability, the map $Df_x: E\to F$ is
linear and continuous, so that $Df_x\in L(E,F)$.
If $E$ and $F$ are Banach spaces, then a map
$f:E\to F$ is defined to be continuously (Fr\'echet) differentiable
if the map $x\to Df_x$ is continuous
from $U$ to $L_c(E,F)$, where $L_c(E,F)$ is the
set of continuous maps from $E$ to $F$ equipped with the
operator norm topology.
But Fr\'echet differentiability is strictly stronger than
Bastiani's
differentiability specialized to Banach spaces~\cite{Neeb-01}.
This is why Bastiani's definition was often dismissed
in the literature~\cite{Penot-73} and, for locally
convex spaces that are not Banach, 
the map $Df$ was generally required to be continuous from $U$ to 
$L(E,F)$ equipped with some well-chosen topology.
However, when $E$ is not normable, 
no topology on $L(E,F)$ provides
the nice properties of Bastiani's 
definition~[\onlinecite[p.~6]{Keller-74}]
(Hamilton~[\onlinecite[p.~70]{Hamilton-82}] gives a simple
example of a map which is continuous $U\times E \to E$
but such that the corresponding map $U\to L(E,E)$ is not continuous).
Thus, $L(E,F)$ was equipped with various non-topological 
convergence structures~[\onlinecite[p.~23]{Keller-74}].
The result is an impressive zoology of differentiabilities.
Twenty-five of them were reviewed and classified
by Averbukh and Smolyanov~\cite{Averbukh-68}.
Still more can be found in the extensive
lists given by
G\"ahler~\cite{Gahler} and 
Ver Eecke~\cite{VerEecke} 
covering the period up to 1983
(see also \cite{Yamamuro,Keller-74,Nashed}).

Nowadays, essentially two concepts of differentiability
survive, Bastiani's and the so-called 
\emph{convenient approach} developed by Kriegl--Michor in 
the reference monograph~\cite{KrieglMichor}, which is weaker than
Bastiani's for general Hausdorff locally convex spaces.
In particular, on any locally convex space which is not
bornological, there is a conveniently smooth map 
which is not continuous~[\onlinecite[p.~19]{BastianiPhD}].
However, a nice feature of both approaches is that for a Fr\'echet 
space $E$, a function 
$f:E\to \bbK$ is smooth in the sense of Bastiani
iff it is smooth in the sense of the convenient
calculus~\cite{Abbati-99}.
Bastiani differentiability became widespread after it was used by
Michor~\cite{Michor-80},
Hamilton~\cite{Hamilton-82} (for Fr\'echet spaces)
and Milnor~\cite{Milnor-83} and it is now 
vigorously developed by Gl\"ockner and Neeb
(see also~\cite{Khesin-09}).

To complete this section, we would like
to mention that the Bastiani differential
is sometimes called the Michal-Bastiani
differential~\cite{Bertram-08,Szilasi-08,Dahmen-14} 
(or even Michel-Bastiani differential~\cite{Szilasi-08}).
This is not correct.
The confusion comes from the fact that
Bastiani defines her differentiability in 
several steps. She starts from the 
G\^ateaux derivability, then she says that a map
$f:U\to F$ is 
\emph{differentiable at $x$}
(see~[\onlinecite[p.~18]{BastianiPhD}] 
and~[\onlinecite[p.~18]{Bastiani-64}])
if: i) $Df_x$ is linear and continuous
from $E$ to $F$ and 
ii) the map $m_x: \bbR\times E\to F$ 
defined by
\begin{eqnarray*}
m_x(t,v) &=& \frac{f(x+tv)-f(x)}{t}-Df_x(v),
\end{eqnarray*}
for $x+t v\in U$,
is continuous  at $(0,v)$
for all $v\in E$.
This differentiability at $x$ is indeed
equivalent to the differentiability defined by Michal~\cite{Michal-38}
in 1938, as proved in Refs.~\onlinecite{Averbukh-68,Massam-74},
[\onlinecite[p.~72]{Keller-74}] and [\onlinecite[p.~202]{VerEecke}].
What we call Bastiani differentiability is called
\emph{differentiability on an open set}
by Bastiani (see~[\onlinecite[p.~25]{BastianiPhD}] and
[\onlinecite[p.~44]{Bastiani-64}]) and
is strictly stronger than Michal differentiability.

The same distinction between Michal-Bastiani
differentiability and Bastiani differentiability
is made by Keller~[\onlinecite[p.~72]{Keller-74}] 
in his thorough review. 
Bastiani's differentiability is denoted by
$C^1_c$ by Keller~\cite{Keller-74}, who
also attributes the definition equivalent to $C^1_c$ to 
Bastiani alone~[\onlinecite[p.~11]{Keller-74}].

In her PhD thesis, Andr\'ee Bastiani developed her concept of
differentiability to define distributions
on a locally convex space $E$ with values in 
a locally convex space $F$. She started
from Schwartz' remark that a distribution is, 
locally, the derivative of a continuous function.
She used her differential $D$ to define
$F$-valued distributions over $E$~\cite{Ehresmann-11}.
A drawback of Bastiani's framework with
respect to the convenient framework is that her
category is not Cartesian closed for locally convex
spaces that are not Fr\'echet. 

\subsection{Properties of the differential}
We review now some of the basic properties
of functional derivatives which will be used in
the sequel. We strongly recommend 
Hamilton's paper~\cite{Hamilton-82},
adapted to locally convex spaces by
Neeb~\cite{Neeb-01}.

\subsubsection{Continuity}
\label{conti-sect}
We characterize continuous (nonlinear) maps between two 
locally convex spaces.
\begin{lem}
\label{conti-lem}
Let $E$ and $F$ be locally convex spaces whose topology is 
defined by the families of seminorms
$(p_i)_{i\in I}$ and $(q_j)_{j\in J}$,
respectively. 
Then $f$ is continuous at 
$x$ iff, for every seminorm $q_j$ of $F$ and every $\epsilon >0$, 
there is a finite 
number $\{p_{i_1},\dots ,p_{i_k}\}$ of seminorms of $E$
and $k$ strictly positive numbers  $\eta_1,\dots, \eta_k$ such that 
$p_{i_1}(x-y)< \eta_1,\dots ,p_{i_k}(x-y)< \eta_k$
imply $q_j\big(f(y)-f(x)\big)<\epsilon$.
\end{lem}
\begin{proof}
This is just the translation in terms of seminorms 
of the fact that $f$ is continuous at $x$ if, for every open set
$V$ containing $f(x)$ , there is an open set $U$
containing $x$ such that $f(U)\subset V$~[\onlinecite[p.~86]{Kelley}].
\end{proof}
When the seminorms of $E$ are saturated~[\onlinecite[p.~96]{Horvath}],
as the seminorms  $\pi_{m,K}$ of $C^\infty(\bbR^d)$, the condition
becomes simpler: a map $f:C^\infty(\bbR^d)\to \bbK$
is continuous at $x$ if and only if, 
for every $\epsilon>0$, there is a seminorm
$\pi_{m,K}$ and an $\eta>0$ such that
$\pi_{m,K}(x-y) < \eta$ implies $|f(y)-f(x)|<\epsilon$.
Since Fr\'echet spaces are metrizable, we can also use the 
following characterization
of continuity~[\onlinecite[p.~154]{Bourbaki-Topo-II}]:
\begin{prop}
\label{sequen-cont}
Let $E$ be a metrizable 
topological space and $F$ a
topological space. Then, a 
map $f:E\to F$ is continuous at a point 
$x$ iff, whenever 
a sequence $(x_n)_{n\in\mathbb{N}}$
converges to $x$ in $E$, the sequence
$f(x_n)_{n\in\mathbb{N}}$ converges to $f(x)$ in $F$.
\end{prop}
Another useful theorem is~[\onlinecite[p.~III.30]{Bourbaki-TVS}]:
\begin{prop}
Let $E$ and $F$ be two Fr\'echet spaces and $G$ a locally
convex space. Every separately continuous bilinear mapping
from $E\times F$ to $G$ is continuous.
\end{prop}
This result extends to multilinear 
mappings from a product $E_1\times\dots\times E_n$ of Fr\'echet spaces
to a locally convex space~\cite{ReedSimonI,Neeb-11}.

\subsubsection{The fundamental theorem of calculus}
The fundamental theorem of calculus for functionals reads
\begin{thm}
Let $f$ be a Bastiani differentiable map between
two Hausdorff locally convex spaces $E$ and $F$. Let
$U$ be an open set in $E$, $x$ in $U$
and $v$ in $E$ such that 
$(x+tv)\in U$ for every  $t$ in an open
neighborhood $I$ of $[0,1]$, so that
$g:t\mapsto f(x+tv)$ is a map from $I$ to $F$. Then,
\begin{eqnarray}
f(x+v) &=& f(x) + \int_0^1 g'(t) dt
\nonumber\\&=&
  f(x) + \int_0^1 Df_{x+tv}(v) dt.
\label{fundthmD}
\end{eqnarray}
\end{thm}
To give a meaning to Eq.~(\ref{fundthmD}), we need to define an
integral of a function taking its values
in a locally convex space.
To cut a long story short~\cite{Bourbaki-Int-1-4,Bourbaki-Int-6}:
\begin{dfn}
Let $X$ be a locally compact space 
(for example $\bbR^n$ or some finite dimensional
manifold), $\mu$ a measure on $X$  and $F$ a Hausdorff locally convex space. 
Let $f$ be a compactly supported continuous function
from $X$ to $F$. Let $F'$ be the topological dual of
$F$ (i.e. the space of continuous linear maps
from $F$ to $\bbK$).
If there is an element $y\in F$ such that
\begin{eqnarray*}
\langle \alpha, y\rangle &=& 
\int_{X} \langle\alpha, f\rangle d\mu,
\end{eqnarray*}
for every $\alpha\in F'$, where $\langle\cdot,\cdot\rangle$ denotes
the duality pairing, then we
say that $f$ has a weak integral and we denote
$y$ by $\int_{X} f d\mu$.
\end{dfn}
The uniqueness of the weak integral follows from the
fact that $F$ is Hausdorff.
In general, the existence of a weak integral
requires some completeness property 
for $F$~[\onlinecite[p.~79]{Bourbaki-Int-1-4}].
However, this is not the case for the fundamental
theorem of calculus~[\onlinecite[p.~27]{BastianiPhD}].
This point was stressed by
Gl\"ockner~\cite{Glockner-02-complete}.

\subsubsection{Additional properties}
\label{lindifsect}
For maps between locally convex spaces, the linearity
of the differential is not completely 
trivial~\cite{Hamilton-82}.
\begin{prop}
Let $E$ and $F$ be locally convex spaces and $f$
be a Bastiani differentiable map from an open subset $U$ of
$E$ to $F$.  Then, for every $x\in U$, 
the differential $Df_x:E\to F$ is a linear map.
\end{prop}

The chain rule for Bastiani-differentiable functions
was first proved by Bastiani
herself~\cite{BastianiPhD}~\cite{Bastiani-64}
(see also \cite{Bertram-08}).
\begin{prop}
Assume that $E,F,G$ are locally convex spaces,
$U\subset E$ and $V\subset F$ are open subsets and 
$f:V\to G$ and $g:U\to V$ are 
two Bastiani-differentiable maps.
Then, the composite map $f\circ g:U\to G$ is Bastiani differentiable
and $D(f\circ g)_x=Df_{g(x)}\circ Dg_x$.
\end{prop}

By using these properties, the reader can prove that our
examples are all Bastiani differentiable.

\subsection{Smooth functionals}
To define smooth functionals we first define
multiple derivatives.
\begin{dfn}
\label{differ-def}
Let $U$ be an open subset of a locally convex space $E$
and $f$ a map from $U$ to a locally convex space
$F$. We say that $f$ is \emph{$k$-times Bastiani differentiable}
on $U$ if:
\begin{itemize}
\item The $k$th G\^ateaux differential
\begin{eqnarray*}
D^kf_x(v_1,\dots,v_k)
&=& \frac{\partial^k f(x+t_1 v_1+\dots + t_k v_k)}
  {\partial t_1\dots\partial t_k}\Big|_{\mathbf{t}=0},
\end{eqnarray*}
where $\mathbf{t}=(t_1,\dots,t_k)$,
exists for every $x\in U$ and every  $v_1, \dots,v_k \in E$.
\item The map
$D^k f:U\times E^k\to F$ is continuous.
\end{itemize}
\end{dfn}
Notice that for a function $f$ assumed to be $k$-times Bastiani differentiable,
the restriction to any finite dimensional affine subspace is not only 
$k$-times differentiable (in the usual sense) but indeed of class $C^k$.
The set of $k$-times Bastiani differentiable
functions on $U$ is denoted by $C^k(U)$,
or $C^k(U,F)$ when the target space $F$
has to be specified.
Bastiani gives an equivalent definition, called
$k$-times differentiability
on $U$~[\onlinecite[p.~40]{BastianiPhD}],
which is denoted by $C^k_c$ by Keller~\cite{Keller-74}.
\begin{dfn}
Let $U$ be an open subset of a locally convex space $E$
and $f$ a map from $U$ to a locally convex space
$F$. We say that $f$ is smooth on $U$ if
$f\in C^k(U,F)$ for every integer $k$.
\end{dfn}
We now list a number of useful properties of the
$k$-th Bastiani differential:
\begin{prop}
\label{Cku-prop}
Let $U$ be an open subset of a locally convex space $E$
and $f\in C^k(U,F)$, where $F$ is a locally convex space, then
\begin{enumerate}
\item $D^kf_x(v_1,\dots,v_k)$ is a $k$-linear 
  symmetric function of $v_1,\dots,v_k$~[\onlinecite[p.~84]{Hamilton-82}].
\item The function $f$ is of class $C^m$ for all 
    $0\le m \le k$~[\onlinecite[p.~40]{BastianiPhD}].
  In particular, $f$ is continuous.
\item The compositions of two functions in $C^k$ is in
  $C^k$ and the chain rule 
  holds~[\onlinecite[p.~51]{BastianiPhD}] and
  [\onlinecite[p.~84]{Hamilton-82}].
\item The map $D^mf$ is in
  $C^{k-m}(U,L(E^m,F))$~[\onlinecite[p.~40]{BastianiPhD}]
  where $L(E^m,F)$ is the space of jointly continuous
  $m$-linear maps from $E$ to $F$, equipped with the locally
  convex topology of uniform convergence on the compact sets
  of $E$: i.e. the topology generated by the seminorms
  \begin{eqnarray*}
  p_{C,j}(\alpha)=\sup_{(h_1,\dots,h_m)\in C} q_j\big(\alpha(h_1,\dots,h_m)\big),
  \end{eqnarray*}
where $C=C_1\times\dots \times C_m$,
  $C_i$ runs over the 
  compact sets of $E$ and
  $(q_j)_{j\in J}$ is a family of seminorms defining the
  topology of $F$.
\item If $E$ is metrizable, then $f\in C^k(U,F)$ iff
  $f$ belongs to $C^{k-1}(U,F)$ and 
  $D^{k-1}f:U\to L(E^{k-1},F)$ is Bastiani 
   differentiable~[\onlinecite[p.~43]{BastianiPhD}].
  Here the metrizability hypothesis is used to obtain a
  canonical injection from
  $C(U\times E,L(E^{k-1},F))$ to $C(U\times E^k,F)$.
\end{enumerate}
\end{prop}
We refer the reader to~\cite{Hamilton-82,Neeb-01}
and Bastiani's cited works for the proofs.
Other results on $C^k(U)$ functions can be found
in Keller's book~\cite{Keller-74}.
All the statements of Proposition~\ref{Cku-prop}
are valid for $k=\infty$, i.e. smooth functions.
Bastiani also defines jets of smooth 
functions between locally convex 
spaces~[\onlinecite[p.~52]{BastianiPhD}] and
[\onlinecite[p.~75]{Bastiani-64}].

Note that Neeb~\cite{Neeb}
and Gl\"ockner~\cite{Glockner-05} 
agree with Bastiani for the definition of the
first derivative but they use
an apparently simpler definition of higher
derivatives by saying that $f$ is $C^k$ iff
$df$ is $C^{k-1}$ iff $d^{k-1}f$ is $C^1$.
However, this definition is less natural because,
for example, $f\in C^2$ if 
$df:U\times E \to F$ is $C^1$. In the definition 
of the first derivative, $U$ is now replaced by
$U\times E$ and $E$ by $E\times E$. In other words,
$d^2$ is a continuous map from 
$U\times E^3$ to $F$. 
More generally $d^k$ is a 
continuous map from
$U\times E^{2^k-1}$ to $F$~[\onlinecite[p.~20]{Glockner-05}].
Moreover, according to Proposition~1.3.13~[\onlinecite[p.~23]{Glockner-05}],
a map $f$ belongs to $C^k$ if and only if 
it belongs to $C^k(U)$ is the sense of Bastiani,
and Bastiani's $D^kf$ is denoted by
$d^{(k)}f$ by Gl\"ockner~[\onlinecite[p.~23]{Glockner-05}] and 
called the $k$-th differential of $f$.
The $k$-th derivatives $d^kf$ and $d^{(k)}f=D^kf$
are not trivially related. 
For example~[\onlinecite[p.~24]{Glockner-05}]:
$d^2f(x,h_1,h_2,h_3)=D^2f(x,h_1,h_2)+Df(x,h_3)$.

The Taylor formula with remainder for
a function in $C^{n+1}(U)$ reads~[\onlinecite[p.~44]{BastianiPhD}]:
\begin{eqnarray}
f(x+th) &=& f(x) + \sum_{k=1}^{n} \frac{t^k}{k!}
  D^kf_x(h^k)
\nonumber\\&&
 + 
   \int_0^t \frac{(t-\tau)^n}{n!}  D^{n+1}f_{x+\tau h}(h^{n+1})
  d\tau,
\label{Taylor-eq}
\\&=&
f(x) + \sum_{k=1}^{n} \frac{t^k}{k!}
  D^kf_x(h^k)
\nonumber\\&+&
   \int_0^t   \frac{(t-\tau)^{n-1}}{(n-1)!}
  \big(D^{n}f_{x+\tau h}(h^{n}) -D^{n}f_{x}(h^{n})\big)
  d\tau.
\nonumber
\end{eqnarray}
Taylor's formula with remainder is a very important tool to deal with smooth functions
on locally convex spaces.


The reader can check that all our examples are smooth
functionals in the sense of Bastiani.
\begin{eqnarray*}
D^kF_\varphi(v_1,\dots,v_k) &=& 
\frac{n!}{(n-k)!} \int_M f(x) \varphi^{n-k}(x) 
\\&&
   v_1(x)\dots v_k(x) dx,
\end{eqnarray*}
for $k\le n$ and $D^kF_\varphi=0$ for $k>n$.
\begin{eqnarray*}
D^kG_\varphi(v_1,\dots,v_k) &=& 
\frac{n!}{(n-k)!} \int_{M^n} g(x_1,\dots,x_n)
\\&&
   v_1(x_1)\dots v_k(x_k) \varphi(x_{k+1})
  \dots \varphi(x_n) \\&& dx_1\dots dx_n,
\end{eqnarray*}
for $k\le n$ and $D^kG_\varphi=0$ for $k>n$.
Recall that $g$ is a symmetric, smooth compactly supported
function of its arguments.
The functional $H$ has only two non-zero derivatives
and
\begin{eqnarray*}
D^2H_\varphi(v_1,v_2) &=& 
2\sum_{\mu\nu} \int_M g^{\mu\nu}(x) 
  h(x) \partial_\mu v_1(x) \partial_\nu v_2(x) dx.
\end{eqnarray*}
The example $I$ has an infinite number of nonzero
derivatives:
\begin{eqnarray*}
D^kI_\varphi(v_1,\dots,v_k) &=& 
\int_M f(x) e^{\varphi(x)} v_1(x)\dots v_k(x) dx.
\end{eqnarray*}
Finally
\begin{eqnarray*}
D^kJ_\varphi(v_1,\dots,v_k) &=& 
e^{\int_M f(x) \varphi(x) dx} 
\int_{M^k} f(x_1)\dots f(x_k) 
\\&& v_1(x_1)\dots v_k(x_k)
  dx_1\dots dx_k.
\end{eqnarray*}

The functionals $F$, $G$ and $H$ are
polynomials in the sense of 
Bastiani~[\onlinecite[p.~53]{BastianiPhD}]:
\begin{dfn}
Let $E$ and $F$ be locally convex spaces. A polynomial of
degree $n$ on $E$
is a smooth function $f: E \to F$ such that 
$D^kf=0$ for all $k>n$.
\end{dfn}

Let $u$ be a distribution in ${\mathcal D}'(M^k) $, then the functional
$f:\calD(M)\to \bbK$ defined by
$f(\varphi)=u( \varphi^{\otimes k}) $
is polynomial in the sense of Bastiani and its $k$-derivative is:
$$ D^k f_\varphi(v_1, \dots,v_k) = 
 \sum_\sigma u(v_{\sigma(1)}\otimes \dots \otimes v_{\sigma(k)}), $$
where $\sigma$ runs over the permutations
of $\{1,\dots,k\}$ and the canonical inclusion 
$ \calD(M)^{\otimes k} \subset \calD(M^k)$
was used.

If $F$ and $G$ are smooth maps from
$E$ to $\bbK$, we can compose the smooth map
$\varphi\mapsto \big(F(\varphi),G(\varphi)\big)$
and the multiplication in $\bbK$ to show that:
\begin{prop}
Let  $E$ be a locally convex space, and $U$ an open
set in $E$. Then the space of 
smooth functionals from $U$ to $\bbK$ is a sub-algebra of the 
algebra of real valued functions.
\end{prop}

\section{Properties of functionals}
\label{funct-sect}
We now  prove important properties of smooth functionals.
We first investigate the support of a functional.
The fact that $DF_\varphi$ is continuous
from $C^\infty(M)$ to $\bbK$ exactly means that 
$DF_\varphi$ is a compactly supported distribution
for every $\varphi$. The support of $F$ is then
essentially the union over $\varphi$
of the supports of $DF_\varphi$. We prove that,
for any smooth functional and any $\varphi\in C^\infty(M)$,
there is a neighborhood $V$ of $\varphi$ such that
$F|_V$ is compactly supported.

The second property that we investigate is required
to establish a link with quantum field theory.
In this paper, we deal with functionals that are
smooth functions $F$ on an open subset $U$ of
$E=\Gamma(M,B)$, where $\Gamma(M,B)$
is the space of smooth sections of some finite
rank vector bundle $B$  on the manifold $M$.
There is a discrepancy between
$D^kF_\varphi$, which is a continuous \emph{multilinear} map
from $E^k$ to $\bbK$, and the quantum field
amplitudes (e.g. represented pictorially by Feynman diagrams) that are 
continuous \emph{linear} maps
from $E^{\hat{\otimes}_\pi k}=\Gamma(M^k,B^{\boxtimes k})$
to $\bbK$, i.e. elements 
of the space $\Gamma'(M^k,(B^*)^{\boxtimes k})$ of compactly
supported distributions with values in the
$k$-th external tensor power of the dual bundle
$B^*$.
It is easy to see that there is a canonical correspondence
between $D^kF_\varphi$ and its associated
distribution 
on $E^{\hat{\otimes}_\pi k}$, that we denote by $F^{(k)}_\varphi$.
However, the equivalence between the continuity
of $D^kF$ on $U\times E^k$ and the
continuity of $F^{(k)}$ on $U\times E^{\hat{\otimes}_\pi k}$
requires a proof.

Finally, we show that the order of $F^{(k)}$ is
locally bounded.

\subsection{Support of a functional}\label{SuppFunct}

Brunetti, D{\"u}tsch and Fredenhagen proposed to
define the support of a functional $F$
by the property that, if the support of the smooth function
$\psi$ does not meet the support of $F$,
then $F(\varphi+\psi)=F(\varphi)$ for all $\varphi$.
More precisely~\cite{Brunetti-09}:
\begin{dfn}\label{firstsupportdefi}
Let $F: U\to \bbK$ be a  Bastiani smooth function, 
with $U$ a subset of $C^\infty(M)$.
The support of 
$F$ is the set of points $x\in M$ such that,
for every open set ${\mathcal U}_x$ containing $x$, there is a
$\varphi \in U$ and a $\psi$ in $C^\infty(M)$ with $ \varphi + \psi \in U$
such that $\supp\psi\subset {\mathcal U}_x$ and $F(\varphi+\psi)\not=F(\varphi)$.
\end{dfn}
We want to relate this definition of the support of $F$ with
the support of $D_\varphi F$, which is compactly supported
as every distribution over $C^\infty(M)$~\cite{HormanderI}.
To do so, we need a technical lemma about
connected open subsets 
in locally convex spaces.
\begin{lem}\label{connexbrokenarcs}
Let $U$ be a connected open set in a locally convex space $E$
then any pair $(x,y)\in U^2$ can be connected by a 
piecewise affine path. 
\end{lem}
\begin{proof}
Define the equivalence relation
$\sim$ in $E$ as follows, two elements
$(x,y)$ are equivalent iff they are
connected by a 
piecewise affine path. 
Let us prove that this equivalence relation is both
open and closed hence any non empty equivalence class
for $\sim$ is both open and closed
in $U$ hence equal to $U$.

Let $x\in U$ then there exists a convex neighborhood $V$ of $x$ in $U$ which
means that every element in $V$ lies in the class of $x$, the relation is open.
Conversely let $y$ be in the closure of the equivalence class of $x$,
then any neighborhood $V$ of $y$ contains an element equivalent to $x$.
Choose some convex neighborhood $V$ then we find $z\in V$ s.t. $z\sim x$,
but $z\sim y$ hence $x\sim z\sim y$ and we just proved that the equivalence class of $x$
was closed.
\end{proof}

We can now prove an alternative formula, 
due to Brunetti, Fredenhagen and Ribeiro~\cite{Brunetti-12}.
\begin{lem}\label{supportsecond}
For every Bastiani smooth function $F: U\to \bbK$, 
with $U$ a connected open subset of $C^\infty(M)$:
\begin{eqnarray}
\supp (F) &=& \overline{\bigcup_{\varphi\in U}\supp DF_\varphi}.
\label{suppF}
\end{eqnarray}
\end{lem}
\begin{proof}
Using the result of Lemma
\ref{connexbrokenarcs}, we may reduce
to the case where $U$ is an open convex set.
We prove that both sets
$\overline{\bigcup_{\varphi}\supp DF_\varphi}$ and 
$\supp F$ as defined in Def.~\ref{firstsupportdefi}
have identical complements. Indeed, for every point $x\in M$, 
$x\notin\text{supp }F$ means by definition of the support that
there exists an open neighborhood $\Omega$ of $x$ such that 
$\forall (\psi,\varphi)
\in\mathcal{D}(\Omega)\times C^\infty(M)$,
$F(\varphi+\psi) = F(\varphi)$. It follows that
for all $\psi\in\mathcal{D}(\Omega)$,
there exists $\varepsilon>0$ such that  $\vert t\vert\leqslant\varepsilon\implies \varphi+t\psi\in U$ and 
$t\in[-\varepsilon,\varepsilon]\mapsto F(\varphi+t\psi)$ is a constant function of $t$
therefore $\frac{dF(\varphi+t\psi)}{dt}|_{t=0}=DF_\varphi(\psi)=0$.
This means that for all $\varphi\in U$, the support of $DF_\varphi\in\mathcal{E}^\prime(M)$
does not meet $\Omega$ hence $\Omega$ lies in the complement
of $\cup_{\varphi\in U}\text{supp }(DF_\varphi)$ and therefore
$x\in \Omega$ does not meet the closure $\overline{\cup_{\varphi\in U}\text{supp }(DF_\varphi)}$.
 
Conversely if $x$ does not meet the closure $\overline{\cup_{\varphi\in U}\text{supp }(DF_\varphi)}$, then
there is some neighborhood $\Omega$ of $x$ which does not meet
$\cup_{\varphi\in U}\text{supp }(DF_\varphi)$ therefore for all
$(\varphi,\psi)\in U\times\mathcal{D}(\Omega)$ s.t. $\varphi+\psi\in U$,
the whole straight path $[\varphi,\varphi+\psi]$
lies in $U$ (by convexity of $U$) hence
$$\forall t\in [0,1], DF_{\varphi+t\psi}(\psi)=0\implies 
\int_0^1 dt DF_{\varphi+t\psi}(\psi)=0, $$
and by the fundamental theorem of calculus
$ F(\varphi+\psi)=F(\varphi)$.
\end{proof}

Now we show that any smooth functional
is locally compactly supported:
\begin{prop}
\label{localcompactsupport}
Let $F:U \mapsto \bbK$
be a Bastiani smooth function, where $U$ is
an open connected
subset of $E=C^\infty(M)$. For every
$\varphi \in U $, there is a neighborhood
$V$ of $\varphi$ in $U$ and a compact
subset $K\subset M$ such that 
the support of $F$ restricted to $V$
is contained in $K$.
Moreover for all integers $n$ and all $\varphi\in U$,
the distributional support $\supp\left(D^n F_\varphi\right)$ is 
contained in $K^n\subset M^n$.
\end{prop}
\begin{proof}
By definition of the support of a functional, it is enough
to show that, for every $\varphi\in V$,
$\supp DF_\varphi \subset K$.
If $F$ is smooth, then $DF:U\times E\to \bbK$ is
continuous. Thus, it is continuous in the neighborhood
of $(\varphi_0,0)$ for every $\varphi_0\in U$. In other words,
for every $\epsilon >0$,
there is a neighborhood $V$ of $\varphi_0$, a seminorm
$\pi_{m,K}$ and an $\eta>0$ such that
$|DF_\varphi(\chi)| < \epsilon$ for every
$\varphi \in V$ and every $\chi\in E$ such that
$\pi_{m,K}(\chi)<\eta$. Now, for every $\psi\in E$
such that $\pi_{m,K}(\psi)\not=0$,
we see that $\chi=\psi\eta/(2\pi_{m,K}(\psi))$
satisfies $\pi_{m,K}(\chi)<\eta$.
Thus, $|DF_\varphi(\chi)| < \epsilon$ for every
$\varphi\in V$ and, by linearity,
$|DF_\varphi(\psi)| < (2\epsilon/\eta) \pi_{m,K}(\psi)$.
On the other hand, if $\pi_{m,K}(\psi)=0$, then for any $\mu>0$
$\psi_{m,K}(\mu\psi)=0 < \eta$, so that
$|DF_\varphi(\mu\psi)|<\epsilon$. By linearity,
$|DF_\varphi(\psi)|<\epsilon/\mu$ for any $\mu>0$
and we conclude that $|DF_\varphi(\psi)|=0$.
Thus, for every $\varphi\in V$ and every $\psi\in E$,
\begin{eqnarray*}
|DF_\varphi(\psi)| &\le & 2\frac{\epsilon}{\eta} \pi_{m,K}(\psi).
\end{eqnarray*}
Of course, this inequality implies that $ DF_\varphi (\psi) = 0 $ 
when $\psi$ is identically zero on the compact subset $K$.

Let us show that this implies that the support of $DF_\varphi$ is contained in $K$.
To avoid possible problems at the boundary, take any compact neighborhood $K'$
of $K$. Now, take an open set $\Omega$ in $M$ such that
$\Omega\cap K'=\emptyset$. Then, for every smooth
function $\psi$ supported in $\Omega$, we have
$\pi_{m,K}(\psi)=0$ because the seminorm $\pi_{m,K}$
takes only into account the points of $K$. As
a consequence, the restriction of $DF_\varphi$ to 
$\Omega$ is zero, which means that 
$\Omega$ is outside the support of $DF_\varphi$
for every $\varphi\in V$.
Thus, $\supp F|_V\subset K$ because,
for every $\varphi\in V$, the support of 
$DF_\varphi$ is included
in every compact neighborhood of $K$.
Finally we show that, if $F$ is compactly supported,
then all $D^nF_\varphi$ are compactly supported
with $\supp\, D^nF_\varphi \subset (\supp\, F)^{\times n}$.
This is easily seen by the following cutoff function argument:
if $F$ is compactly supported then for every cutoff function $\chi$ 
equal to $1$ on an arbitrary compact neighborhood of
$\supp F$, we have:
$F(\varphi)=F(\chi\varphi),\forall \varphi\in E$.
Then it is immediate by definition of $D^nF_\varphi$ that
\begin{eqnarray*}
D^nF_\varphi(\psi_1,\dots,\psi_n) &=&
\frac{d^nF(\chi(\varphi+t_1\psi_1+\dots+t_n\psi_n))}{dt_1\dots dt_n}|_{\mathbf{t}=0}
\\&=&
D^nF_{\chi\varphi}(\chi\psi_1,\dots,\chi\psi_n)
\end{eqnarray*}
thus
$\supp D^nF_\varphi \subset \supp\chi^{\times n} $
for all test function $\chi$ s.t. $\chi|_{\supp F}=1$ and therefore
$\supp\, D^nF_\varphi \subset (\supp\, F)^{\times n}$ since
$\underset{\chi|_{\supp F}=1}{\bigcap} \supp\chi=\supp F$.
\end{proof}

\subsection{A multilinear kernel theorem with parameters.}

We work with $M$ a smooth manifold and $B\to M$ a smooth vector bundle of 
finite rank over $M$.
Let $E=\Gamma(M,B)$ be its space of smooth sections and 
$U$ an open subset of $E$. 
We consider smooth
maps $F:E\mapsto \mathbb{K}$ where $\mathbb{K}$ is the field
$\mathbb{R}$ or $\mathbb{C}$.
In this section we relate the
Bastiani derivatives $D^kF$, which are $k$-linear
on $\Gamma(M,B)$ to  the distributions used in
quantum field theory, which are linear on
$\Gamma(M^k,B^{\boxtimes k})$. 
Since the $k$-th derivative $D^kF$ of a smooth map
is multilinear and continuous in the last $k$
variables, we can use the following result~[\onlinecite[p.~471]{Shult}]
and [\onlinecite[p~259]{Bourbaki-AlgebraI}]:
\begin{lem}
Let $E$ be a Hausdorff locally convex space. 
There is a canonical isomorphism between any
$k$-linear map $f: E^k\to \bbK$ and the map
$\overline{f}: E^{\otimes k}\to \bbK$, where $\otimes$ is the
algebraic tensor product,  which is linear
and defined as follows: if
$\chi=\sum_j \chi^j_1\otimes \dots \otimes \chi^j_k$
is a finite sum of tensor products, then
\begin{eqnarray}
\label{barf}
\bar{f}(\chi) &=& \sum_j 
  f(\chi^j_1, \dots ,\chi^j_k).
\label{barfdef}
\end{eqnarray}
\end{lem}
Let us give a topological version of this lemma, using the
projective topology.
We recall the definition of a family of seminorms
defining the projective topology on tensor powers of locally convex spaces 
following~[\onlinecite[p.~23]{Ryan-PhD}].
For arbitrary seminorms $p_1,\dots, p_k$ on $E$ there exists
a seminorm $p_1\otimes\dots\otimes p_k$ on $E^{\otimes k}$ defined
for every $\psi\in E^{\otimes k}$ by
$$p_1\otimes\dots\otimes p_k(\psi)=\inf \sum_n   
p_1(e_{1,n})\dots p_k(e_{k,n}),  $$ 
where the infimum is taken over the representations of $\psi$ as finite sums: $\psi=\sum_ne_{1,n}\otimes\dots\otimes e_{k,n}$. 
Following K{\"o}the~[\onlinecite[p.~178]{Kothe-II}], one can prove that
$p_1\otimes\dots\otimes p_k$ is the largest seminorm 
on $E^{\otimes k}$ such that
\begin{eqnarray}
p(x_1\otimes\dots\otimes x_k) &=& p_1(x_1) \dots p_k(x_k),
\label{tensorseminorm}
\end{eqnarray}
 for 
all $x_1,\dots,x_k$ in $E$. More precisely, if 
$p$ is a seminorm on $E^{\otimes k}$ satisfying
Eq.~\eqref{tensorseminorm}, then 
$p(X)\leqslant  (p_1\otimes\dots\otimes p_k)(X)$ for
every $X\in E^{\otimes k}$.

The projective topology on $E^{\otimes k}$ is defined
by the family of seminorms
$p_1\otimes\dots\otimes p_k$
where each $p_i$ runs over a family
of seminorms defining the topology of $E$~[\onlinecite[p.~24]{Ryan-PhD}].

When $E$ is Fr\'echet, its topology is defined
by a countable family of seminorms and it follows that 
the family of seminorms $p_1\otimes\dots\otimes p_k$ on 
$E^{\otimes k}$ is countable. Hence they can be used to construct 
a metric on $E^{\otimes k}$
which defines the same topology as the projective topology and
$E^{\hat{\otimes}_\pi k}$ is defined as the completion 
of $E^{\otimes k}$ relative to this metric
or equivalently with respect to the projective topology.
A fundamental property of the projective topology
is that $f:E^k\to \bbK$ is (jointly) continuous
iff $\bar{f}:E^{\otimes k}\to \bbK$, still defined by Eq.~(\ref{barfdef}) is continuous
with respect to the projective topology~[\onlinecite[p.~I-50]{Grothendieck-55}].
Then $\bar{f}$ extends uniquely to a continuous map (still denoted by
$\bar{f}$) on the completed tensor product 
$E^{\hat{\otimes}_\pi k}$~[\onlinecite[p.~III.15]{Bourbaki-TVS}].

If $E=C^\infty(M)$, then
$E^{\hat{\otimes}_\pi k}=C^\infty(M^k)$~[\onlinecite[p.~530]{Treves}],
and $\bar{f}$ becomes a compactly supported distribution
on $M^k$. More generally, 
if $E=\Gamma(M,B)$, then $E$ is 
Fr\'echet nuclear and 
$E^{\hat{\otimes}_\pi k}=
\Gamma(M^k,\mathcal{E}^{\boxtimes k})$~[\onlinecite[p.~72]{Tarkhanov}].
Thus, $\bar{f}$ becomes a compactly supported distributional section
on $M^k$. 
If $f=D^kF_\varphi$ we denote $\bar{f}$ by $F_\varphi^{(k)}$.

Recall that, if $U$ is an open subset of a Hausdorff locally 
convex space $E$, a map $F:U\to \bbK$ is smooth iff every
$D^k F$ is continuous from $U\times E^k$ to $\bbK$.
According to the previous discussion, continuity
on $E^k$ is equivalent to continuity on $E^{\hat{\otimes}_\pi k}$.
Therefore, it is natural to wonder when joint continuity 
on $U\times E^k$ is equivalent to joint continuity on
$U\times E^{\hat{\otimes}_\pi k}$.
This is the subject of the next paragraphs.

We now prove an equicontinuity lemma.

\begin{lem}
\label{equicontlem}
Let $E$ be a Fr\'echet space, $U$ open in $E$, and 
$F:U\times E^k\mapsto \bbK$ a continuous
map, multilinear in the last $k$ variables. 
Then for every $\varphi_0\in U$, there exist a neighborhood
$V$ of $\varphi_0$, a seminorm $q$ of $E^{\hat{\otimes}_\pi k}$ and 
a constant $C>0$ such that
$$\forall\varphi\in V, \forall\psi\in E^{\hat{\otimes}_\pi k},\,\ 
\vert \barF(\varphi,\psi)-\barF(\varphi_0,\psi) \vert\leqslant C\,q(\psi).$$
\end{lem}

\begin{proof}
By continuity of $F:U\times E^k\mapsto \bbK$,
for every $\epsilon >0$,
there exist a neighborhood $V$ of $\varphi_0$ and 
neighborhoods $U_1,\dots,U_k$ of zero such that
$\varphi\in V$ and $e_i\in U_i$ for $i=1,\dots, k$
imply 
$|F(\varphi,e_1,\dots,e_k)|\le \epsilon$.
Since $E$ is locally convex,
there are continuous seminorms $p_1,\dots, p_k$ on $E$
and strictly positive numbers $\eta_1,\dots,\eta_k$
such that $e_i\in U_i$ if $p_i(e_i) \le  \eta_i$.
Consider now arbitrary elements $e_1,\dots,e_k$
of $E$ such that $p_i(e_i)\not=0$. Then,
if $f_i= e_i\eta_i/p_i(e_i)$ we have
$p_i(f_i) \le \eta_i$.
Thus,
$|F(\varphi,f_1,\dots, f_k)| < \epsilon$ and,
by multilinearity,
$|F(\varphi,e_1,\dots, e_k)| < M p_1(e_1)\dots p_k(e_k)$
where
$M=\epsilon/ (\eta_1,\dots,\eta_k)$.
The argument in the proof of 
Proposition~\ref{localcompactsupport} shows that
$|F(\varphi,e_1,\dots, e_k)|=0$ if $p_i(e_i)=0$.
Therefore, for every $\varphi\in V$ and every
$e_1$,\dots, $e_k$ in $E$, 
$|F(\varphi,e_1,\dots, e_k)| \le M p_1(e_1)\dots p_k(e_k)$.
By defining $C=2M$ we obtain for every $ \varphi\in V$
and every $(e_1,\dots,e_k)\in E^k$
\begin{equation}
\vert F(\varphi,e_1,\dots,e_k)-F(\varphi_0,e_1,\dots,e_k)\vert\leqslant 
C p_1(e_1)\dots p_k(e_k).
\end{equation}
By definition of $\barF$, for all $(\varphi,\psi)\in V\times E^{\otimes k}$:  
\begin{eqnarray*}
\vert\barF(\varphi,\psi)-\barF(\varphi_0,\psi)\vert
&\leqslant & \sum_n\vert 
F(\varphi,e_{1,n},\dots,e_{k,n})
 \\&&-
F(\varphi_0,e_{1,n},\dots,e_{k,n})\vert 
\\
&\leqslant & C\sum_n  p_1(e_{1,n})\dots p_k(e_{k,n}),
\end{eqnarray*}
for all representations of $\psi$ as finite sum 
$\psi=\sum_ne_{1,n}\otimes\dots\otimes e_{k,n}$. 
Taking the infimum over such representations yields the estimate:
\begin{equation}
\forall (\varphi,\psi)\in V\times E^{\otimes k},\,\
\vert\barF(\varphi,\psi)-\barF(\varphi_0,\psi)\vert \leqslant 
C q(\psi),
\end{equation}
for the seminorm $q=p_1\otimes\dots\otimes p_k$ on $E^{\otimes_\pi k}$
and the above inequality extends to any $\psi$ of $E^{\hat{\otimes}_\pi k}$ 
since, in a Fr\'echet space, $\psi$ can be approximated by a convergent
sequence of elements in $E^{\otimes k}$ by density of $E^{\otimes k}$ 
in $E^{\hat{\otimes}_\pi k}$ and
by continuity of the seminorm $p_1\otimes\dots\otimes p_k$ for the 
topology
of $E^{\hat{\otimes}_\pi k}$.
\end{proof}

Another way to state the previous result is to say
that the family of linear maps
$\{F(\varphi,\cdot)\telque \varphi\in V\}$ is
equicontinuous~[\onlinecite[p.~II.6]{Bourbaki-TVS}].

\subsubsection{Proof of the main result.}
We are now ready to prove
\begin{prop}\label{multilinearkernelprop}
Let $E$ be a Fr\'echet space and $U\subset E$ an open subset. 
Then
$F:U\times E^k \mapsto \bbK$, multilinear in the 
last $k$ variables, is jointly continuous 
iff the corresponding map $\barF: U\times E^{\hat{\otimes}_\pi k}\to 
\bbK$ is jointly continuous.
\end{prop}

\begin{proof}
One direction of this theorem is straightforward and holds
if $E$ is any locally convex space.
Indeed, by definition of the projective tensor product,
the canonical multilinear mapping
$E^k\to E^{\hat\otimes k}$ is
continuous~[\onlinecite[p.~I-50]{Grothendieck-55}]. Therefore,
if $\barF$ is continuous on $U\times E^{\hat\otimes k}$ then, 
by composition with the canonical
multilinear mapping, $F$ is continuous on
$U\times E^k$.

Let us prove that the continuity of $F$ implies the continuity
of $\barF$. 
According to Lemma~\ref{conti-lem}, we have to show
that, for every $\varphi_0\in U$ and every
 $\varepsilon>0$, there exist a finite number of 
continuous seminorms $q_i$ on $E^{\hat{\otimes}_\pi k}$, 
a neighborhood $V$ of $\varphi_0$ and $\eta_i>0$ such that,
if $\varphi$ belongs to $V$ and 
$\psi\in E^{\hat{\otimes}_\pi k}$ satisfies
$q_i(\psi-\psi_0) \leqslant \eta_i$, then
$|\barF(\varphi,\psi)-\barF(\varphi_0,\psi_0)\vert\leqslant 
\varepsilon$.

In order to bound $\barF(\varphi,\psi)-\barF(\varphi_0,\psi_0)$,
we cut it into three parts
\begin{equation}
\label{eq:3term}
\begin{array}{rcl}
\barF(\varphi,\psi)-\barF(\varphi_0,\psi_0)&=& 
  \barF(\varphi,\psi_k)-\barF(\varphi_0,\psi_k) \\
&+&\barF(\varphi,\psi-\psi_k)-\barF(\varphi_0,\psi-\psi_k)\\
&+& \barF(\varphi_0,\psi)-\barF(\varphi_0,\psi_0),
\end{array}
\end{equation}
where $\psi_k$ is some element of the algebraic tensor product
$E^{\otimes k}$ close enough to $\psi_0$ that we choose now.
The equicontinuity Lemma~\ref{equicontlem} yields a neighborhood $V_2$
of $\varphi_0$, a constant $C>0$ and a continuous seminorm $q_2$ on 
$E^{\hat{\otimes}_\pi k}$ so that 
$$\forall\varphi\in V_2,\forall \psi\in E^{\hat{\otimes}_\pi k},\,\ 
\vert\barF(\varphi,\psi)-\barF(\varphi_0,\psi) 
\vert\leqslant  C q_2(\psi).$$
Now we use the fact that the algebraic tensor product $E^{\otimes
k}$ is everywhere dense in $E^{\hat{\otimes}_\pi k}$ 
hence there is some element $\psi_k$ in the algebraic tensor product
$E^{\otimes k}$  such that
$q_2(\psi_0-\psi_k)\leqslant \eta_2 $ 
with $\eta_2 := \frac{\varepsilon}{6C}$. 

Now that $\psi_k$ is chosen, we can bound the second term of the sum 
(\ref{eq:3term}), namely $\barF(\varphi,\psi-\psi_k)-
\barF(\varphi_0,\psi-\psi_k)$. From  the previous relation, 
for every $\varphi\in V_2$ and every
$\psi$ such that $q_2(\psi-\psi_0) \le \eta_2$, the triangle inequality 
for $q_2(\psi-\psi_k)$ gives us
\begin{eqnarray*}
\vert\barF(\varphi,\psi-\psi_k)-\barF(\varphi_0,\psi-\psi_k)
\vert &\leqslant& C(q_2(\psi-\psi_0)\\&& +q_2(\psi_0-\psi_k))\leqslant
\frac{\varepsilon}{3}.
\end{eqnarray*}

We continue by bounding the last term $\barF (\varphi_0,\psi) - \barF (\varphi_0,\psi_0)$ in the sum (\ref{eq:3term}).
Since $\varphi_0$ is fixed, the map $\psi\mapsto \barF(\varphi_0,\psi)$ is continuous in $\psi$ 
because, since $\barF(\varphi_0,\cdot)$ is continuous
on $E^{\otimes_\pi k}$, its extension to the completion
$E^{\hat\otimes_\pi k}$, also denoted by $\barF(\varphi_0,\cdot)$,
is continuous. It follows that 
there is some seminorm $q_1$ of $E^{\hat{\otimes}_\pi k}$
and a number  $\eta_1>0$ such that if $\psi \in U$ satisfies $q_1\left(\psi-\psi_0\right)\leqslant \eta_1$ then
$$\vert \barF(\varphi_0,\psi)-\barF(\varphi_0,\psi_0)\vert
\leqslant \frac{\varepsilon}{3}.$$

To bound the first term $ \barF(\varphi,\psi_k)-\barF(\varphi_0,\psi_k)$ in the sum (\ref{eq:3term}), we use the fact that
$\psi_k\in E^{\otimes k}$. Thus, 
$\psi_k=\sum_{j=1}^p (e_{1,j}\otimes\dots\otimes e_{k,j})$
for some $(e_{1,j},\dots,e_{k,j})\in E^{k}$.
By definition of $\barF$, 
\begin{eqnarray*}
\barF(\varphi,\psi_k)-\barF(\varphi_0,\psi_k) &=& \sum_{j=1}^p
F(\varphi,e_{1,j},\dots,e_{k,j})\\&&
-F(\varphi_0,e_{1,j},\dots,e_{k,j}).
\end{eqnarray*} 
By continuity of $F$ in the first factor, the finite sum
$\sum_{j=1}^p F(\varphi,e_{1,j},\dots,e_{k,j})$ is continuous in
$\varphi$ and there is some neighborhood $V_3$ of $\varphi_0$
such that for all $\varphi\in V_3$ the following bound
$$\vert\sum_{j=1}^p  F(\varphi,e_{1,j},\dots,e_{k,j})-
F(\varphi_0,e_{1,j},\dots,e_{k,j})\vert\leqslant
\frac{\varepsilon}{3},$$
holds true.
\ecam

Finally we found some neighborhood $V=V_2\cap V_3$ of $\varphi_0$, 
two seminorms $q_1$ and $q_2$ of $E^{\hat{\otimes}_\pi k}$, 
and two numbers $\eta_1>0$ and $\eta_2=\epsilon/6C$ such
that $q_1(\psi-\psi_0)< \eta_1$ and $q_2(\psi-\psi_0) < \eta_2$
imply
$$|\barF(\varphi,\psi)-\barF(\varphi_0,\psi_0)|\leqslant \varepsilon.$$
The proposition is proved.
\end{proof}

Now we can specialize our result to the space of
smooth sections of vector bundles.
We recall a fundamental result on the projective
tensor product of sections~[\onlinecite[p.~72]{Tarkhanov}]:
\begin{prop}
Let $\Gamma(M,B)$ be the  space
of smooth sections of some smooth finite rank vector bundle $B\to M$
on a manifold $M$.
Then $\Gamma(M,B)^{\hat{\otimes}_\pi k}=
\Gamma(M^k,B^{\boxtimes k})$.
\end{prop}
Note that we could remove the index $\pi$ in
$\hat{\otimes}_\pi$ because we saw that $\Gamma(M,B)$
is nuclear.
If we specialize
Proposition~\ref{multilinearkernelprop}
to sections of vector bundles (which is a Fr\'echet space)
we obtain
\begin{thm}\label{multilinearkernelbundles}
Let $E=\Gamma(M,B)$ be the space
of smooth sections of some smooth finite rank vector bundle 
$B\to M$.
Then 
$F:U\times E^k \to \bbK$ multilinear in the last $k$
variables is jointly continuous
iff the corresponding map $\barF: U\times \Gamma(M^k,B^{\boxtimes k})
\to \bbK$ is jointly continuous.
\end{thm}
\begin{proof}
The proof is an immediate consequence of the
fact that $E^{\hat{\otimes}_\pi k}=\Gamma(M^k,B^{\boxtimes k})$
and Proposition \ref{multilinearkernelprop}.
\end{proof}
The definition of a Bastiani smooth functional implies the
following corollary:
\begin{thm}\label{joint:cont}
Let $E=\Gamma(M,B)$ be the space
of smooth sections of some smooth finite rank vector bundle 
$B\to M$.
A map $F:U\to\bbK$, where $U$ is open in $E$,
is Bastiani smooth iff
the maps $F^{(k)}: U\times \Gamma(M^k,B^{\boxtimes k}) \to \bbK$
are (jointly) continuous for every $k\ge 1$.
\end{thm}
To interpret Theorem~\ref{multilinearkernelprop} in terms of 
distributional kernels,
let $B\to M$ denote a smooth vector bundle of finite rank over
a manifold $M$ equipped with a fixed density $\vert dx\vert$ and
$B^*\to M$ the corresponding dual bundle.
Recall that $\Gamma(M,B)^\prime\simeq
\Gamma(M,B^*)\otimes_{C^\infty(M)}
\mathcal{E}^\prime(M)$~\cite{Grosser-08},
where $\Gamma(M,B^*)\otimes_{C^\infty(M)}\mathcal{E}^\prime(M)$
denotes
the compactly supported distributional sections of the dual bundle
$B^*$. 
In global analysis, to every continuous linear map
$L:\Gamma(M,B)\to \Gamma(M,B)^\prime$,
we associate the
continuous bilinear map $B:(\varphi_1,\varphi_2)\in
\Gamma(M,B)^2\mapsto  \langle
\varphi_1 ,L\varphi_2 \rangle \in\bbK$ where the pairing is understood
as a pairing between a smooth function
and a distribution once the smooth density on $M$ is fixed.

The usual kernel theorem of the theory of distributions
states that a bilinear map can be represented
by a distribution:
$K_L\in \mathcal{E}^\prime(M\times M)\otimes_{C^\infty(M^2)}
\Gamma(M^2,B^*\boxtimes B^*)$
living on configuration space $M^2$ such that,
for every $(\varphi_1,\varphi_2)\in \Gamma(M,B)^2$
$$\langle K_L,\varphi_1\boxtimes \varphi_2
\rangle_{\Gamma_2^\prime,\Gamma_2}=\langle
\varphi_1 ,L\varphi_2
\rangle_{\Gamma(M,B),\Gamma(M,B)^\prime}, $$
where $\Gamma_2=\Gamma(M^2,B\boxtimes B)$.
Theorem~\ref{multilinearkernelprop} generalizes the kernel
theorem by using multilinear maps instead of
bilinear ones and by considering that these
multilinear maps depend continuously and non-linearly
on a parameter $\varphi$.

\subsection{Order of distributions}
If $F$ is a Bastiani smooth map from an open subset $U$
of $E=C^\infty(M)$ to $\bbK$, then,
for every $\varphi\in U$,
$D^kF_\varphi$ is a compactly supported distribution.
Therefore, the order of $F^{(k)}_\varphi$ is 
finite~[\onlinecite[p.~88]{Schwartz-66}].
For some applications, for example to local functionals,
it is important to require the order of
$F^{(k)}_\varphi$ to be  \emph{locally bounded}:
\begin{prop}
\label{bounded-order}
Let $E=C^\infty(M)$ and $F:E\to \bbK$ be a smooth 
functional on an open subset $U$ of $E$. Then, for
every $\varphi_0\in U$ and every integer $k$,
there is a neighborhood $V$
of $\varphi_0$, an integer $m$ and a compact $K\subset M^k$
such that, for every $\varphi\in V$, 
the order of $F^{(k)}_\varphi$ is smaller than $m$
and $F^{(k)}_\varphi$ is supported in $K$.
\end{prop}
\begin{proof}
According to Lemma~\ref{equicontlem}, for every 
$\varphi_0$ in $U$, there is a neighborhood $V$ of $\varphi_0$,
a constant $C$ and a seminorm $\pi_{n,K}$ of $C^\infty(M)$ such that 
\begin{eqnarray*}
\vert F^{(k)}_\varphi(\psi)-F^{(k)}_{\varphi_0}(\psi) \vert &\leqslant &
C\,\pi_{n,K}(\psi).
\end{eqnarray*}
This means that the order of 
$F^{(k)}_\varphi-F^{(k)}_{\varphi_0}$ is bounded
by $n$~[\onlinecite[p.~64]{Schwartz-66}], and the order
of $F^{(k)}_\varphi$ is bounded by $n$ plus the order
of $F^{(k)}_{\varphi_0}$.
Moreover, if $\supp\psi\cap K=\emptyset$, then $\pi_{n,K}(\psi)=0$
and $F^{(k)}_\varphi(\psi)-F^{(k)}_{\varphi_0}(\psi)=0$.
This means that the support of
$F^{(k)}_\varphi-F^{(k)}_{\varphi_0}$ is contained in $K$
and the support of $F^{(k)}_\varphi(\psi)$ is contained
in the compact $K\cup \supp F^{(k)}_{\varphi_0}$.
\end{proof}

Note also that, in general, the order of the
distributions is not bounded on $U$:
\begin{lem}
\label{counterorderlem}
Let $g\in \calD(\bbR)$ and $(\chi_n)_{n\in\bbZ}$
a sequence of functions such that
$\chi_n\in \calD([n-1,n+1])$
and $\sum_{n\in \bbZ}\chi_n=1$. Then, the functional
\begin{eqnarray*}
F(\varphi) &=& \sum_{n=-\infty}^\infty 
\int_\bbR \chi_n(\varphi(x)) 
\frac{d^{|n|} \varphi}{dx^{|n|}}(x)
g(x)dx,
\end{eqnarray*}
is Bastiani smooth on $C^\infty(\bbR)$ but the order of 
$F^{(k)}$ is not bounded on $C^\infty(\bbR)$.
\end{lem}
\begin{proof}
The functional  $F$ is smooth because, for every 
$\varphi_0\in C^\infty(\bbR)$, we can define a neighborhood of 
$\varphi_0$ by 
$V=\{\varphi\telque \pi_{0,K}(\varphi-\varphi_0) < \epsilon\}$,
where $K$ is a compact neighborhood of the support of $g$.
Let $N$ be the smallest integer strictly greater
than $\pi_{0,K}(\varphi_0)+\epsilon$. Then,
$-N <  \varphi(x) <  N$ for every $\varphi\in V$ and
every $x \in K$ and 
\begin{eqnarray*}
F(\varphi) &=& \sum_{n=-N-1}^{N+1} 
\int \chi_n(\varphi(x)) \varphi^{(|n|)}(x) g(x)dx,
\end{eqnarray*}
is a finite sum of smooth functionals. 

However, the order of
\begin{eqnarray*}
F^{(1)}_\varphi(\psi) &=& \sum_{n=-\infty}^\infty 
\int \Big(\chi_n(\varphi(x)) \psi^{(|n|)}(x)
\\&& + \chi^\prime_n(\varphi(x)) 
\psi^{(1)}(x) \varphi^{(|n|)}(x) \Big)g(x)dx,
\end{eqnarray*}
is not bounded on $C^\infty(\bbR)$. Indeed, for any
positive integer $n$, we can find a smooth function $\varphi$ such
that
$\chi_n\big(\varphi(x)\big) g(x) \not=0$ for some $x\in \supp g$.
Since $F^{(1)}_\varphi(\psi) $ contains a
factor $\psi^{(n)}(x)$ it is at least of order $n$.
\end{proof}

\subsection{Derivatives as smooth functionals}
In the next section we equip several spaces of
functionals with a topology. As a warm-up
exercise, we show here that the maps
$F^{(k)}$ are smooth functionals from
$C^\infty(M)$ to $\calE'(M^k)$.

We adapt to the case of functionals
the general result given in item~4 of Prop.~\ref{Cku-prop} 
stating that, if $F$ is a smooth functional on $U$, then 
$D^kF$ is  a Bastiani smooth map
from $U$ to $L(E^k,\bbK)$.
We need to identify the topology of
$L(E^k,\bbK)$ used by Bastiani. Let us start with $L(E,\bbK)$. Bastiani
furnishes $E$ with 
the topology of convergence on all compact sets of 
$E$. In other words, the seminorms that define
the topology of $L(E,\bbK)$ are
$p_C(u)=\sup_{f\in C} |\langle u,f\rangle|$,
where $C$ runs over the compact subsets of 
$C^\infty(M)$. Since $C^\infty(M)$ is a
Montel space~[\onlinecite[p.~239]{Horvath}], 
the topology of uniform convergence on compact
sets is the same as the strong 
topology~[\onlinecite[p.~235]{Horvath}]. 
This means that $L(E,\bbR)$ is the
space $\calE'(M)$ of compactly supported
distributions with its usual topology.
Similarly, $L(E^k,\bbR)$ can be identified to a subset of
$\calE'(M^k)$ with its usual topology.
We just obtained the following result:
\begin{prop}
\label{Dksmooth}
Let $U$ be an open subset of $C^\infty(M)$ and
$F:U\to \bbK$ a Bastiani smooth functional.
Then, for every integer $k$, the map
$F^{(k)}:U\to \calE'(M^k)$ is smooth in the
sense of Bastiani.
\end{prop}

\section{Topologies on spaces of functionals}
We need to define a topology on the various
spaces of functionals used in quantum field theory.
The generally idea is to define seminorms on
$F$ and its derivatives $F^{(k)}$.
The topology proposed by Brunetti, D{\"u}tsch
and Fredenhagen~\cite{Brunetti-09}
is the initial topology of all the
maps $F\to F^{(k)}_\varphi$, where 
each $F^{(k)}_\varphi$ belongs to a nuclear
space determined by a wavefront set condition.
This topology is nuclear, but the absence of
a control of the dependence on $\varphi$ 
makes it generally not complete.
We then describe Bastiani's topology, which
is complete but has two drawbacks: it does not
take wavefront set conditions into account and
it is generally not nuclear.
Finally we shall describe the family of topologies
proposed by Dabrowski~\cite{Dabrowski-14} which are both
nuclear and complete.

\subsection{Bastiani's topology}

Bastiani defines several topologies on the
space of Bastiani smooth maps between two locally
convex spaces~[\onlinecite[p.~65]{BastianiPhD}].
For the case of functionals, we
consider the topology defined by
the following seminorms:
\begin{eqnarray*}
p_{C_0}(F)&=& \sup_{\varphi\in C_0} |F(\varphi)|,\\
p_{C_0,C}(F) &=&
\sup_{(\varphi,h_1,\dots,h_k)\in C_0\times C} |D^kF_\varphi(h_1,\dots,h_k)|,
\end{eqnarray*}
where $C=C_1\times\dots \times C_k$ and
$C_i$ runs over the compact sets of $\Gamma(M,B)$.
By using Bastiani's results~[\onlinecite[pp.~66]{BastianiPhD}]
we obtain
\begin{prop}
Let $B\overset{\pi}{\rightarrow} M$ be a finite
rank vector bundle over the manifold $M$
and $\Gamma(M,B)$ be the space of smooth
sections of $B$.
Then, with the seminorms defined above,
the space of smooth functionals on $\Gamma(M,B)$
is a complete locally convex space.
\end{prop}
A similar topology was used
by Gl\"ockner~[\onlinecite[p.~367]{Glockner-02}]
and Wockel~[\onlinecite[p.~29]{Wockel-03}] and [\onlinecite[p.~12]{Wockel-06}].

\subsection{Nuclear and complete topologies}
\label{Dabtoposect}
Quantum field theory uses different spaces of
functionals defined by conditions on the
wave front set of $F^{(k)}_\varphi$.
Recall that the wave front set describes 
the points and the directions of singularity
of a distribution~\cite{BDH-13}.
Yoann Dabrowski~\cite{Dabrowski-14} recently described
nuclear and complete topologies for 
spaces of functionals with wave front set conditions.
We present some of his topologies for several common spaces
of functionals.

Dabrowski's definition differs from Bastiani in two respects.
To describe the first difference, recall that,
according to Proposition~\ref{Dksmooth}, if 
$F:U\to\bbR$ is a smooth functional, then the derivatives
$F^{(k)}:U \to \calE'(M^k)$ are smooth functionals.
To add the wave front set conditions, Dabrowski requires
$F^{(k)}$ to be smooth from $U$ to 
$\calE'_{\Gamma_k}(M^k)$, which is the space of 
compactly supported distributions whose wave front sets
are included in $\Gamma_k$, a cone in $T^*M^k$.
In fact, Dabrowski supplements this definition with a more
refined wave front set (the dual wave front set) which
enables him to equip $\calE'_{\Gamma_k}(M^k)$
with a Montel, complete, ultrabornological and nuclear
topology. He also considers support conditions
which are different from compact.

To describe the second difference, recall that 
Bastiani's topology gives a locally
convex space which is complete. However it is generally not
nuclear. This is due to a theorem by Colombeau and 
Meise~\cite{Colombeau-81} which says, 
broadly speaking, that a function space over a
Fr\'echet space cannot be nuclear for the topology
of convergence over some
balanced, convex, compact sets
of infinite dimension.
To avoid that problem, the variable
$\varphi$ is made to run over \emph{finite dimensional}
compact sets. More precisely, Dabrowski considers
compact sets in $\bbR^m$ for any finite value of $m$
and smooth maps $f$ from $\bbR^m$ to an open
subset of $C^\infty(M)$. He defines two
families of seminorms:
\begin{eqnarray}
p_{f,K}(F) &=& \sup_{\varphi\in f(K)} |F(\varphi)|,
  \label{normfK}\\
p_{n,f,K,C}(F) &=& \sup_{\varphi\in f(K)} \sup_{v\in C}
    |\langle F^{(n)}_\varphi,v\rangle|,
  \label{normnfKC}
\end{eqnarray}
where $K$ is a compact subset of $\bbR^m$ for some $m$
and $C$ is an equicontinuous subset of the dual
of the space of distributions to which $F^{(n)}_\varphi$
belong.
Dabrowski proved that, with this family of seminorms,
the space of functionals $\frakF$ is a complete locally convex
nuclear space~\cite{Dabrowski-14-2}.

We describe now several types of functionals
that have been used in the literature and we
specify more precisely their topologies.

\subsection{The regular functionals}
A polynomial functional of the form
\begin{eqnarray*}
F(\varphi) &=& \sum_n \int_{M^n} dx_1\dots dx_n 
  f_n(x_1,\dots,x_n) \varphi(x_1)\dots\varphi(x_n),
\end{eqnarray*}
where the sum over $n$ is finite and
$f_n\in \calD(M^n)$, is called a
\emph{regular functional}~\cite{Fredenhagen-13},
because
all its derivatives are smooth functions~\cite{Fredenhagen-Houches},
i. e. the wave front set of $F^{(k)}_\varphi$ is empty.
More generally,
we define the space $\frakF_{\mathrm{reg}}(M)$ of
regular functionals to be the set of Bastiani smooth functionals $F$
such that
$\WF(F^{(n)}) =\emptyset$ for every $n>0$.
Thus, $F^{(n)}\in \calE'_\emptyset(M^n)=\calD(M^n)$
and the sets $C$ in 
Equation~(\ref{normnfKC}) are the equicontinuous
sets of $\calD'(M^n)$. By a general theorem~[\onlinecite[p.~200]{Horvath}], 
the topology of uniform
convergence on the equicontinuous sets of
$\calD'(M^n)$ is equivalent to the
topology given by the seminorms
of its dual $\calD(M^n)$. In other words, the 
topology of $\frakF_{\mathrm{reg}}(M)$ is defined by
the seminorms~\cite{Dabrowski-14-2}:
\begin{eqnarray}
p_{f,K}(F) &=& \sup_{\varphi\in f(K)} |F(\varphi)|,
  \\
p_{n,f,K,\alpha}(F) &=& \sup_{\varphi\in f(K)}
  p_{\alpha,n} \big( F^{(n)}_\varphi\big),\label{pnfKalpha}
\end{eqnarray}
where $p_{\alpha,n}$ runs over a defining
family of seminorms of 
$\calD(M^n)$~\cite{Baer-14}.
With this topology, $\frakF_{\mathrm{reg}}(M)$ is
nuclear and complete.

Note that the tensor
product of elements of 
$\calD(M^m)$ with elements
of $\calD(M^n)$ is not continuous
in $\calD(M^{m+n})$~\cite{Hirai-01}.
Thus, the product in 
$\frakF_{\mathrm{reg}}(M)$ is hypocontinuous
but not continuous.

\subsection{The microcausal functionals}
It is possible to describe quantum field theory 
(up to renormalization) as the
deformation quantization of classical field 
theory~\cite{Dutsch}.
For the deformation quantization of the product
$F G$ of two functionals to first order in $\hbar$, 
we need to evaluate
$\langle DF_\varphi \otimes DG_\varphi,\Delta_+\rangle$,
where $\Delta_+$ is a singular distribution
(the Wightman propagator) and $\langle \cdot,\cdot\rangle$
is an extension of the duality pairing
between distributions and test functions~\cite{Dabrouder-13}. 
For this pairing to be meaningful to all orders
in $\hbar$, the wave front set of $\Delta_+$
imposes that the wave front set of $D^{(k)}_\varphi$
must not meet the cone $\Gamma_k$ defined as follows~\cite{Dutsch}.

Let $M$ be a Lorentzian manifold with pseudo-metric
$g$. Let $V_{x}^+$ (resp. $V_{x}^-$) be the set of 
$(x;\xi)\in T^*_xM$ such that
$g^{\mu\nu}(x) \xi_\mu \xi_\nu \ge 0$
and $\xi_0\ge0$ (resp. $\xi_0\le0$), where
we assume that $g^{00}>0$.
We define the closed cone
\begin{eqnarray*}
\Gamma_k &=& \{(x_1,\dots,x_k;\xi_1,\dots,\xi_k)
\in \dotT^*M^n\telque
(\xi_1,\dots,\xi_n)\in \\&&
(V_{x_1}^+\times\dots\times  V_{x_n}^+)
\cup 
(V_{x_1}^-\times\dots\times  V_{x_n}^-)
\},
\end{eqnarray*}
where $\dotT^*M^n$ is the cotangent bundle
$T^*M^n$ without its zero section.
The space  $\frakF_\mathrm{mc}$ of microcausal functionals
was originally defined as the
set of Bastiani smooth functionals such that
$F^{(n)}_\varphi\in \calE'_{\Xi_n}(M^n)$ for
every $\varphi$, where
$\Xi_n=\dotT^*M^n\backslash \Gamma_n$ is an open 
cone~\cite{Dutsch,Hollands,Brunetti-09,%
Rejzner-PhD,Fredenhagen-11,Brunetti-12,Fredenhagen-13}.

However, the space $\calE'_{\Xi_n}(M^n)$ being not even
sequentially complete~\cite{Dabrouder-13}, it
is not
suitable to define a complete space of functionals.
Therefore, Dabrowski defines the space $\frakF_\mathrm{mc}$
of microcausal functionals to
be the set of Bastiani smooth functionals such that
$F^{(n)}_\varphi\in \calE'_{\Xi_n,\overline{\Xi}_n}(M^n)$,
which is the completion of $\calE'_{\Xi_n}(M^n)$.
Dabrowski proved that $\calE'_{\Xi_n,\overline{\Xi}_n}(M^n)$
is the set of compactly supported distributions
$u\in \calE'(M^n)$ such that 
the \emph{dual wavefront set} of $u$ is in $\Xi_n$
and the wavefront set of $u$ is in its closure
$\overline{\Xi}_n$ (see~\cite{Dabrowski-14} for a precise
definition of these concepts and of the topology).
This completion is not only complete, but even
Montel and nuclear~\cite{Dabrowski-14}.
According to the general results of Ref.~\onlinecite{Dabrowski-14},
the sets $C$ are now equicontinuous sets of the
bornologification of the normal topology
of $\calD'_{\Gamma_n}$.  However, it was shown~\cite{Dabrowski-14}
that these equicontinuous sets are the same as
the bounded sets of $\calD'_{\Gamma_n}$ with its
normal topology. Therefore, the sets $C$ are the 
well-known bounded sets of 
$\calD'_{\Gamma_n}$\cite{Dabrouder-13}.

With this topology, the space $\frakF_\mathrm{mc}$
is a complete nuclear algebra with
hypocontinuous product.

\subsection{Local functionals}
\label{Loctopsect}
As discussed in the introduction, local
functionals are the basic building block
(Lagrangian) of quantum field theory.  
We shall see that local functionals
are a closed subset of the set of
smooth functionals such that 
$F^{(1)}_\varphi$ can be identified with
an element of $\calD(M)$ that we denote by $\nabla F_\varphi$ 
 and the wave front set of
$F^{(k)}_\varphi$ is included in the conormal $C_k$
of $D_k=\{(x_1,\dots,x_k)\in M^k \telque
 x_1=\dots=x_k\}$. Recall that the conormal
of $D_k$ is the set
of $(x_1,\dots,x_k;\xi_1,\dots,\xi_k)\in T^*M^k$ such that
$x_1=\dots=x_k$ and $\xi_1+\dots+\xi_k=0$.

Since the additivity property (defined in
Section~\ref{additivity-sect}) of local functionals
complicates the matter, we follow Dabrowski~\cite{Dabrowski-14} and, 
for any open set $\Omega\subset M$, we first define
$\frakF_C(\Omega)$ to be the set of smooth maps such that 
$\varphi  \mapsto  \nabla F_\varphi$ is Bastiani smooth
from $C^\infty(\Omega)$ to $\calD(M)$ and, for every
integer $k$,  $\varphi \mapsto F^{(k)}_\varphi$ is Bastiani 
smooth from $C^\infty(\Omega)$ to $\calE'_{C_k}(M^k)$
(we do not need to index $\calE'(M^k)$ with two cones because
$C_k$ is closed~\cite{Dabrowski-14}).
The set $\frakF_\mathrm{loc}(\Omega)$ of local functionals is then
the subset of $\frakF_C(\Omega)$ satisfying the addivity condition.

The topology of $\frakF_C$ is induced by the family of
seminorms given by Eq.~~(\ref{pnfKalpha}) that depend on the 
equicontinuous sets of
the dual 
$\calD'_{\Lambda_k,\overline{\Lambda}_k}(M^k)$ of $\calD'_{C_k}(M^k)$, where 
$\Lambda_k=\dotT^*M^k\backslash C_k$.
They were determined by Dabrowski~[\onlinecite[Lemma~28]{Dabrowski-14}]:
\begin{prop}
A subset $B$ of $\calD'_{\Lambda_k}(M^k)$ is equicontinuous
if and only if there is a closed cone $\Gamma\subset \Lambda_k$
such that $\WF(u)\subset \Gamma$ for every $u\in B$ and
$B$ is bounded in $\calD'_\Gamma(M^k)$.
\end{prop}
The bounded sets of $\calD'_\Gamma(M^k)$ are characterized
in detail in Ref.~\onlinecite{Dabrouder-13}.
The topology of $\calD'_{\Lambda_k}(M^k)$, where $\Lambda_k$ is
open,  can be described as
a non-countable inductive limit as follows.
Write the complement $\Lambda_k^c=\cup \Gamma_n$, where
each $\Gamma_n$ is a compactly supported closed set.
We write the open set $\Gamma_n^c$ as a countable union
of closed sets $\Gamma_n^c=\cup_m \Lambda_{n,m}$, so that
$\Gamma_n = \cap_m \Lambda_{n,m}^c$
and $\Lambda_k^c  = \cup_n \cap_{m} \Lambda_{n,m}^c$ is
a countable union of countable intersections of open sets.
We obtain
$\Lambda_k=\cap_n \cup_m \Lambda_{n,m}$.
We define for a sequence $\alpha$ the closed set
$\Pi_\alpha=\cap_n \Lambda_{n,\alpha(n)}$, such that
$\alpha \le \beta$ implies $\Pi_\alpha\subset \Pi_\beta$.
Then $\Lambda=\cup_\alpha \Pi_\alpha$ is a non-countable
inductive limit of closed cones from which we can
define the topology of $\calD'_{\Lambda_k}(M^k)$
as a non-countable inductive limit of
$\calD'_{\Pi_\alpha}(M^k)$.

The space $\frakF_C$ furnished with the topology induced
by the seminorms defined by Eqs.~(\ref{normfK}) and (\ref{normnfKC})
is complete and nuclear.
The space $\frakF_\mathrm{loc}$ of local functionals
is the closed subset of $\frakF_C$ defined by the additivity
condition defined in the next section.
As a closed subspace of a nuclear complete space,
the space of local functionals is nuclear and complete.

Further examples of spaces of functionals are given
by Dabrowski~\cite{Dabrowski-14-2}.

\section{Additivity}
The characterization of local functionals is a long-standing
mathematical problem. 
According to Rao~\cite{Rao-80}, the first criterium was
proposed by Pinsker in 1938 and called
\emph{partial additivity}~\cite{Pinsker-38}. 
This criterium is also used in physics, but we shall
see that it is not what we need by exhibiting a
partially additive functional which is not 
local.  Then, we shall discuss a more stringent
criterium which is exactly what we need.

\subsection{Partial additivity}

When looking for an equation to characterize functionals
having the form of Eq.~(\ref{deflocal}), one can make the 
following observation. 
Let $\varphi_1$ and $\varphi_2$ be two smooth functions with
disjoint support $K_1$ and $K_2$ and assume that
$f(x,\varphi(x),\dots)=0$ if $\varphi=0$ on a neighborhood of
$x$~\cite{Pinsker-38}, so that $F(0)=0$. Then, 
since the support of $\varphi_1+\varphi_2$ is included in
$K_1\cup K_2$,
\begin{eqnarray*}
F(\varphi_1+\varphi_2) &=& 
 \int_{K_1}  dx f(x,\varphi_1(x)+\varphi_2(x),\dots)
\\&&
+ \int_{K_2}  dx f(x,\varphi_1(x)+\varphi_2(x),\dots)
\\&=&
 \int_{K_1}  dx f(x,\varphi_1(x),\dots)
\\&&
+ \int_{K_2}  dx f(x,\varphi_2(x),\dots) =F(\varphi_1)+F(\varphi_2).
\end{eqnarray*}
Therefore, it is tempting to use the condition of locality:
\begin{equation}
F(\varphi_1+\varphi_2)=F(\varphi_1)+F(\varphi_2),
\label{perturbfunctionaleq}
\end{equation}
for $\varphi_1$ and $\varphi_2$ with disjoint support and functionals
$F$ such that $F(0)=0$.
And indeed, many authors since 1938,
including Gelfand and Vilenkin~[\onlinecite[p.~275]{Gelfand-IV}],
used condition~(\ref{perturbfunctionaleq}), but with
disjoint support 
replaced by $\varphi_1\varphi_2=0$ and smooth functions
by measurable functions
(see \cite{Rao-80} for a review).
In perturbative quantum field theory, partial additivity 
in our sense is also 
used when the function $f$ in Eq.~(\ref{deflocal}) is
polynomial~\cite{Dutsch04,Brunetti-09,Keller-09}
because, in that case, partial additivity is equivalent
to locality in the sense of Eq.~(\ref{deflocal})~\cite{Brunetti-09}.

However, this definition of locality does not suit our
purpose, essentially because the set of functions $\varphi$ that
can be written as $\varphi=\varphi_1+\varphi_2$ (with
$\supp\varphi\cap\supp\varphi_2=\emptyset$) is not dense 
in the space of smooth functions.
We show this now and we construct a partially additive functional
which is not local.

\subsection{A non-local partially additive functional}
\label{nonlocalpartiallyadditive}
We work in the space $C^\infty(\bbS^1)$ of
smooth functions on the unit circle.
We denote by $\mathcal{I}$ the subset of functions 
$f=\varphi_1+\varphi_2$ which are sums
of two elements of $C^\infty(\bbS^1)$
whose supports are disjoint. 
It is not a vector subspace of $C^\infty(\bbS^1)$.

The idea of the construction is the following. 
In the metric space $C^\infty(\mathbb{S}^1)$, we will show that the subset
$\mathcal{I}$ is bounded away from the constant function $f=1$. This means 
that the functional equation~(\ref{perturbfunctionaleq}) only concerns 
the restriction $F|_{\mathcal{I}}$ to a subset which is
bounded away from $1$. Therefore there is some open neighborhood of 
$f=1$ which does not meet $\mathcal{I}$. 
Then we use Sobolev norms to build some cut--off function $\chi$ 
to glue a local functional
near $\mathcal{I}$ with a nonlocal functional near $f=1$.

\begin{lem}
\label{triviallemma1}
The constant function $f=1$ is bounded away from
$\mathcal{I}$ in $C^\infty(\bbS^1)$:
if $f\in \calI$, then
$||f-1||_{C^0}=\sup_{x\in\mathbb{S}^1} \vert f(x)-1\vert\ge 1$.
\end{lem}
\begin{proof}
Let us denote by $\Vert .\Vert_{C^0}$ the norm
$\Vert f\Vert_{C^0}=\sup_{x\in\mathbb{S}^1} \vert f(x)\vert$.
It is a continuous norm for the Fr\'echet topology
of $C^\infty(\bbS^1)$ because
$||f||_{C^0}=\pi_{0,\bbS^1}(f)$. 
Then, if $\supp\varphi_1\cap\supp\varphi_2=\emptyset$ we have
$ \Vert \varphi_1+\varphi_2-1\Vert_{C^0} \geqslant 1.$
Indeed, the supports of $\varphi_1$ and $\varphi_2$ being
compact, the fact that they do not meet
implies that they are at a finite distance.
Thus, there is a point $x\in \bbS^1$ such that
$\varphi_1(x)=\varphi_2(x)=0$.
Hence, $\vert \varphi_1(x)+\varphi_2(x)-1\vert =1$ and 
$\sup_{x\in\mathbb{S}^1}
\vert\varphi_1(x)+\varphi_2(x)-1\vert\ge 1$.
\end{proof}

The second step is to build a smooth function $\chi$
such that $\chi(1)=1$ and $\chi|_\calI=0$.
\begin{lem}
\label{lesstriviallemma2}
There is a smooth function
$\chi:C^\infty(\bbS^1)\to \bbR$ such that
$\chi=1$ on a neighborhood of $f=1$ and 
$\chi(f)=0$ if $||f-1||_{C^0}\ge 1$.
In particular, $\chi|_\calI=0$.
\end{lem}
\begin{proof}
First recall that the Sobolev norm $H^{2k}$ on 
$\mathbb{S}^1$ is defined as
\begin{eqnarray}
\Vert f\Vert_{H^{2k}} &=& \sqrt{\int_{\mathbb{S}^1}\left(
(1-\Delta)^kf(x)\right)^2dx}
\nonumber\\&=&
2\pi \left(\sum_{n\in\mathbb{Z}}
(1+n^2)^{2k}\vert\widehat{f}(n)\vert^2\right)^{\frac{1}{2}},
\end{eqnarray}
where the last representation uses the Fourier series
$f(x)=\sum_n \hat{f}(n) e^{in x}$.
By the Sobolev injections, $H^2(\mathbb{S}^1) $ injects continuously in 
$C^0(\mathbb{S}^1)$.
In other words, there is a constant $C>0$ such that
$\Vert f\Vert_{C^0} \leqslant C \Vert f\Vert_{H^2}$
for every $f\in C^\infty(\bbS^1)$.

Now we take a function $g\in C^\infty(\bbR)$ such
that $g(t)=1$ when $t\le 1/3C^2$ and 
$g(t)=0$ when $t\ge 1/2 C^2$ and
we define $\chi:\C^\infty(\bbS^1)\to \bbR$ by composing $g$ with 
the square of the Sobolev norm.
\begin{eqnarray*}
\chi(f) &=& g\big(||1-f||^2_{H^2}\big).
\end{eqnarray*}
If $\Vert 1-f\Vert_{C^0}\geqslant 1$ 
(in particular, if $f\in \calI$ by 
Lemma~\ref{triviallemma1})
the Sobolev injection leads to:
$$1\leqslant   \Vert 1-f\Vert_{C^0}\leqslant C \Vert 1-f\Vert_{H^2}\implies
\Vert 1-f\Vert_{H^2}^2\geqslant \frac{1}{C^2} $$
hence $g\left(\Vert 1-f\Vert_{H^2}^2\right)=0$ 
by definition of $g$.

On the other hand $\Vert 1-f\Vert_{H^2)}\leqslant \frac{1}{\sqrt{3}C} $ 
means that $f$ belongs to the
 neighborhood of the constant function $f=1$
defined by $V=\{ f \telque \Vert 1-f\Vert_{H^2}\leqslant 1/\sqrt{3}C\}
$.
On this neighborhood,
$g\left(\Vert 1-f\Vert_{H^2}^2\right)=1$.
The smoothness of $\chi$ is an immediate consequence
of the chain rule, the smoothness of $g$ and of the squared Sobolev norm
$\Vert .\Vert_{H^2(\mathbb{S}^1)}^2$.
\end{proof}

We are now ready to define our counterexample:
\begin{thm}
The functional $F_\mathrm{nl}$ on $C^\infty(\mathbb{S}^1)$ defined
for any integer $N>1$ by 
\begin{equation}
F_\mathrm{nl}(f) = \big(1-\chi(f)\big) \int_{\mathbb{S}^1}f(x)dx+\chi(f)
  \left(\int_{\mathbb{S}^1}f(x)dx\right)^{N},
\end{equation}
is partially additive but not local.
\end{thm}
\begin{proof}
For every
$(\varphi_1,\varphi_2)\in C^\infty(\mathbb{S}^1)^2$
whose supports are disjoint, $f=\varphi_1+\varphi_2\in\mathcal{I}$ 
hence $\chi(f)=0$ by Lemma
\ref{lesstriviallemma2}.
Moreover, we saw that, if $\supp\varphi_1\cap\supp\varphi_2=\emptyset$,
then there is a point $x\in \bbS^1$ such that
$\varphi_1(x)=\varphi_2(x)=0$. Thus,
$||1-\varphi_1||_{C^0}\ge 1$ and 
$||1-\varphi_2||_{C^0}\ge 1$. As a consequence,
$\chi(\varphi_1)=\chi(\varphi_2)=0$ by Lemma \ref{lesstriviallemma2} and 
$F_\mathrm{nl}(\varphi_1+\varphi_2)=
\int_{\mathbb{S}^1}(\varphi_1(x)+\varphi_2(x))dx
=F_\mathrm{nl}(\varphi_1)+F_\mathrm{nl}(\varphi_2)$.

On the other hand, in the neighborhood $V$
of $f=1$ given by Lemma \ref{lesstriviallemma2}, $\chi(f)=1$
hence $F_\mathrm{nl}(f)=\left(\int_{\mathbb{S}^1}f(x)dx\right)^N$ 
which is not local. It is even a typical example
of a multilocal functional~\cite{Fredenhagen-11}.
\end{proof}

Since partial additivity is equivalent to locality for
polynomial functions, the non-locality of $F_\mathrm{nl}$ 
can be considered to be non perturbative.
Moreover,  the fact that the
derivatives $D^nF_\mathrm{nl}$ calculated at 
$f=0$ are
supported in the thin diagonal of $(\mathbb{S}^{1})^n$,
although $F_\mathrm{nl}$ is not local, means that
locality cannot be controlled by the support of 
differentials taken at a single function $f$.
We come now to the property that is relevant for
quantum field theory.

\subsection{Additive functionals}
\label{additivity-sect}
In 1965, Chacon and Friedman~\cite{Chacon-65} introduced a more stringent 
concept of additivity which meets our needs:
\begin{dfn}
We say that a Bastiani smooth map $F:C^\infty(M)\to \bbK$ is
\emph{additive} if,
for every triple $(\varphi_1,\varphi_2,\varphi_2)$ of 
smooth functions on $M$, the property  
$\supp\,\varphi_1\cap \supp\,\varphi_3=\emptyset$  implies the property 
\begin{eqnarray}
F(\varphi_1+\varphi_2+\varphi_3) &=& F(\varphi_1+\varphi_2) 
  + F(\varphi_2+\varphi_3)
  - F(\varphi_2).
\nonumber\\&&
\label{Hammerstein}
\end{eqnarray}
\end{dfn}
In the literature, the additivity equation (\ref{Hammerstein})
is also called the \emph{Hammerstein
property}~\cite{Batt-73,Fesmire-74,Korvin-81,Lermer-98,Ercan-99,%
Millington-07}.
The additivity property is equivalent
to the fact that the functional derivatives
are supported on the thin 
diagonal $D_n=\{(x_1,\dots,x_n)\in M^n\telque
x_1=\dots=x_n\}$~\cite{Brunetti-09,Keller-09}.
\begin{prop}\label{localsecondderivativediagonal}
A smooth functional $F$ on $C^\infty(M)$ is additive iff 
$\supp\,F^{(2)}_\varphi\subset D_2$ for every 
$\varphi\in C^\infty(M)$, where 
$D_2=\{(x,y)\in M^2\telque x=y\}$.
If $F$ is an additive functional, then
$\supp\,F^{(n)}_\varphi\subset D_n$ for every 
$\varphi\in C^\infty(M)$, where 
$D_n=\{(x_1,\dots,x_n)\in M^n\telque x_1=\dots=x_n\}$.
\end{prop}
\begin{proof}
We first prove that the second derivative of an additive functional
is localized on the diagonal~\cite{Brunetti-09}.
If we use the additivity property with
$\varphi_1=\lambda\psi$, $\varphi_3=\mu\chi$
and $\supp\,\psi\cap\supp\,\chi=\emptyset$,
then 
\begin{eqnarray*}
F(\lambda\psi+\varphi_2+\mu\chi) &=& 
  F(\lambda\psi+\varphi_2) 
  + F(\varphi_2+\mu\chi)
  - F(\varphi_2).
\end{eqnarray*}
Since no term on the right hand side of this equation depends
on both $\lambda$ and $\mu$, we have
\begin{eqnarray*}
\frac{\partial^2F(\lambda\psi+\varphi_2+\mu\chi)}{\partial\lambda
\partial\mu}  &=& 
  D^2F_{\lambda\psi+\varphi_2+\mu\chi}(\psi,\chi)
\\&=&
  F^{(2)}_{\lambda\psi+\varphi_2+\mu\chi}(\psi\otimes \chi)=0.
\end{eqnarray*}
This equation, being true for every $\varphi_2$, can be
written
 $F^{(2)}_{\varphi}(\psi\otimes \chi)=0$ for every $\varphi$
and every pair $(\psi,\chi)$ with disjoint supports.
Now for every point $(x,y)\in M^2$ such that 
$x\not=y$, there are two open sets $U_x$ containing
$x$ and $U_y$ containing $y$ such that 
$U_x\cap U_y=\emptyset$. Then, any pair of functions
$\psi$ and $\chi$ supported in $U_x$ and $U_y$
satisfies $F^{(2)}_\varphi(\psi\otimes \chi)=0$.
Since the functions $\psi\otimes \chi$ are dense
in $\calD(M^2)$, this implies that every test functions
$f\in \calD(M^2)$
supported in $U_x\times U_y$ 
satisfies $F^{(2)}_\varphi(f)=0$. Thus
$(x,y)\notin \supp\, F^{(2)}_\varphi$
and 
$\supp F^{(2)}_\varphi\subset D_2$.
To determine the support of $F^{(n)}_\varphi$, 
consider a point $(x_1,\dots,x_n)$ which is not in $D_n$.
Then, there are two indices $i$ and $j$ such that
$x_i\not= x_j$. Denote by $U_x$ an open neighborhood of
$x_i$ and by $U_y$ an open neighborhood of $x_j$ and
repeat the previous proof to obtain 
$F^{(2)}_\varphi(\psi\otimes\chi)=0$ for every $\varphi$
and every pair $(\psi,\chi)$ with supports in
$U_x$ and $U_y$.
Now, rewrite $\varphi=\varphi_0+\sum \lambda_k \psi_k$,
where $\psi_k(x_k)\not=0$ and the sum is over all integers from 1 to $n$
except $i$ and $j$. Then, the derivatives with respect
to $\lambda_k$ are all zero and we find again with the
same argument
that $(x_1,\dots,x_n)$ is not in the support of 
$F^{(n)}_\varphi$ for every $\varphi$.

Conversely~\cite{Keller-09,Brunetti-12}, assume
that $\supp\,F^{(2)}_\varphi\subset D_2$
for every $\varphi$. As we have seen in the
first part of the proof, this means that, if
$\psi$ and $\chi$ have disjoint support, then
$D^2F_\varphi(\psi,\chi)=F^{(2)}_\varphi(\psi\otimes\chi)=0$.
By the fundamental theorem of calculus,
\begin{eqnarray*}
F(\varphi+\psi+\chi) &=& F(\varphi+\psi) 
  + \int_0^1 d\mu \frac{d}{d\mu} F(\varphi+\psi+ \mu \chi),\\
F(\varphi+\psi+\mu\chi) &=& F(\varphi+\mu\chi) 
  + \int_0^1 d\lambda \frac{d}{d\lambda} 
     F(\varphi+\lambda\psi+ \mu \chi).
\end{eqnarray*}
Thus,
\begin{eqnarray*}
F(\varphi+\psi+\chi) &=& F(\varphi+\psi) 
  + \int_0^1 d\mu \frac{d}{d\mu} F(\varphi+ \mu \chi)
\\&&+
   \int_0^1 d\lambda \int_0^1 d\mu \frac{\partial^2}
   {\partial\lambda\partial\mu} 
     F(\varphi+\lambda\psi+ \mu \chi).
\end{eqnarray*}
The last term is zero because 
$D^2F_\varphi(\psi,\chi)=0$ and the second
term is $F(\varphi+\chi)-F(\varphi)$. 
We recover the additivity condition.
\end{proof}
Finally, additivity is stronger than partial additivity
because the latter corresponds to the case $\varphi_2=0$
and $F(0)=0$. 
It is strictly stronger because $F_\mathrm{nl}$ is not additive.

\section{Characterization of smooth local functionals}\label{Smooth:loc:funct} 
In this section, we give a characterization of local functionals
inspired by the topology
described in Section~\ref{Loctopsect}.
In the sequel, we shall deal with compactly supported distributions $u$
with empty wavefront sets. 
We repeat the definition of local functionals
in terms of jets:
\begin{dfn}
Let $U$ be an open subset of $C^\infty(M)$.
A Bastiani smooth functional $F:U\to \bbK$ is said to be
\emph{local}
if, for every $\varphi\in U$, there is a neighborhood 
$V$ of $\varphi$, an integer $k$,
an open subset $\mathcal{V}\subset J^kM$ and a smooth
function $f\in C^\infty(\mathcal{V})$ such that
$x\in M\mapsto f(j^k_x\psi)$ is supported in a compact subset $K\subset M$ and
\begin{eqnarray*}
F(\varphi+\psi) &=& F(\varphi)+\int_M f(j^k_x\psi) dx,
\end{eqnarray*}
whenever $\varphi+\psi\in V$
and where $j^k_x\psi$ denotes the $k$-jet of $\psi$ at $x$.
\end{dfn}
We invite the reader not familiarized with jet bundles to have a look at
Section~\ref{sec:jets}, where these objects are carefully defined.
Note that the representation of $F$ by $f$ is not
unique: adding the total derivative of a function
does not change the result. 
We shall see that $f$ belongs to a unique cohomology class for some
specific cohomology theory on the space of local functionals.

 Before we state the main Theorem
of this section, let us start by some useful definition-lemma~:
\begin{lem}\label{nablaFunique}
Let $U$ be an open subset of $C^\infty(M)$ and 
$F:U\to \bbK$ be Bastiani smooth.
For every $\varphi$ such that the distribution $DF_\varphi\in \mathcal{E}^\prime(M)$ 
has empty wave front set, there exists a \textbf{unique 
function} $\nabla F_\varphi\in \mathcal{D}(M)$ such that
\begin{equation}
DF_\varphi[h]=\int_M\nabla F_\varphi(x)h(x)dx.
\end{equation} 
\end{lem}
\begin{proof}
Once a density $dx$ is fixed on $M$,
functions in $L^1_{loc}(M)$ (in particular in $C^\infty(M)$) can be identified with distributions
by the map~:
$$f\in L^1_{loc}(M)\mapsto\left(\phi\mapsto \int_M f\phi dx \right)  $$
and~[\onlinecite[Theorem 1.2.4]{HormanderI}] shows that 
the distribution is uniquely defined when $f$ is continuous hence when $f$ is smooth.

Since $WF(DF_\varphi)=\emptyset$, there exists a unique
$C^\infty$ function $\nabla F_\varphi$ which represents
the distribution $DF_\varphi\in \mathcal{E}^\prime(M)$ by integration on $M$ against $dx$.
\end{proof}

The main theorem of this section is
\begin{thm}
\label{TheoPrincipal}
Let $U$ be an open subset of $C^\infty(M)$ and 
$F:U\to \bbK$ be Bastiani smooth.
Then, $F$ is local if and only if
the following two conditions are satisfied:
\begin{enumerate}
\item $F$ is additive,
\item for every $\varphi\in U$, the differential 
   $DF_\varphi = F^{(1)}_\varphi $ of $F$ at $\varphi$ has an
   empty  wave front set and the map 
   $\varphi\mapsto \nabla F_\varphi$ is Bastiani smooth
  from $U$ to $\calD(M)$.
\end{enumerate}
\end{thm}
Note that our definition of locality is strictly
more general than the usual one because the counterexample 
described in Lemma~\ref{counterorderlem} is local in our sense
but not in the sense of Eq.~(\ref{deflocal}) since
its order is infinite.

The proof is delayed to Section \ref{sec:proofPrincipal}.
Since this theorem deals with jets, we start with a short presentation of the jet bundle.
Our point of view on jets is based on the concept of 
infinitesimal neighborhoods due to Grothendieck and is closely related
to several expositions in the literature~\cite{Navarro, moosa2004jet, kantor1977formes}.

\subsection{\bcam The manifold of jets of functions on a manifold \ecam}
\label{sec:jets}

Let $M$ be a manifold. For every smooth real-valued function $\varphi $ on $M$, 
we call \emph{$k$-jet of $\varphi$
at a point $x \in M$} the class $j^k_x (\varphi)$ of $\varphi$ in the 
quotient $ C^\infty(M)/I_x^{k+1}$, with the understanding that $I_x^{k+1}$
stands for the $(k+1)$-th power of the ideal $I_x$ of smooth functions on 
$M$ vanishing at $x \in M$. Recall that $I_x^{k+1}$ coincides
with the ideal of smooth functions on $M$ whose $k+1$ first 
derivatives vanish at the point $x$.

For all $x \in M$, the space $J^k_x(M)$ of all $k$-jets of functions 
on $M$ at $x$ coincides with $ C^\infty(M)/I_x^{k+1}$ 
and is called the 
\emph{space of $k$-jets at $x$}. It is clearly a vector space.
The disjoint union $J^k(M) := \coprod_{x \in M}  J^k_x(M) $ 
is a smooth vector bundle over $M$ called
the bundle of $k$-th jets. 
Consider the map:
$$ \begin{array}{rrcl} {\mathfrak J}_\Delta: &
 C^\infty (M \times M) & \to & \Gamma( J^k(M) )      \\ &
 \psi & \mapsto & x \mapsto j^k_x({\mathfrak i}_x^*\psi),
    \end{array}
 $$
where $ {\mathfrak i}_x : M \to M \times M$ is the map $y \mapsto (x,y)$.
It is known that ${\mathfrak J}_\Delta$ is surjective onto the space of 
smooth sections of $J^k(M)$ and its kernel is 
the $(k+1)$-th power of the ideal ${\mathcal I}_\Delta $ of 
functions on $M \times M$ vanishing on the diagonal.  

Last, the projection $p_1: M \times M \to M$ onto the first component dualizes in an algebra morphism $\varphi\mapsto p_1^*\varphi$ from $C^\infty (M)$
to $C^\infty (M \times M)$ which endows  $C^\infty(M \times M)$ with a $C^\infty(M)$-module structure.
The space of sections of $J^k(M)$ is also a $C^\infty(M)$-module,
and it is routine to check that $ {\mathfrak J}_\Delta$ is 
a morphism of $C^\infty(M)$-modules.
Therefore, the space of sections of $J^k(M)$ is, as a 
$C^\infty(M)$-module, isomorphic
to the quotient $ C^\infty (M \times M)/{\mathcal I}_\Delta^{k+1} $ and 
$\Gamma(J^k(M))$ fits into the following exact sequence of 
$C^\infty(M)$-modules:
\begin{eqnarray*}
0 &\to & {\mathcal I}_\Delta^{k+1}\to C^\infty(M\times M)
\to \Gamma( J^k(M) )\simeq 
\\&& C^\infty(M\times M)/{\mathcal I}_\Delta^{k+1}
\to 0.
\end{eqnarray*}
And the map that to $f$ associates
its $k$-jet reads:
\begin{eqnarray*}
f\in C^\infty(M)\mapsto [(p_1^*f)]\in  C^\infty(M\times M)/{\mathcal I}_\Delta^{k+1}.
\end{eqnarray*}

The purpose of the rest of this section is to prove the technical Proposition \ref{prop:NeDependQueDesJets}, the statement of which we now explain. For all integer $k \in {\mathbb N}$, there is a natural vector bundle morphism ${\mathfrak j}^k$ from the trivial bundle over $M$ with typical fiber $E=C^\infty(M)$ to the bundle $J^k(M) \to M$ of $k$-jets. This morphism simply consists in assigning 
to a pair $(f,x)$ in $ E \times M \to M$  the $k$-jet of $f$ at $x$. 
In equation:
 \begin{equation}\label{eq:assignJet} \begin{array}{rrcl} {\mathfrak j}^k :& E \times M & \mapsto & j^k(M) \\  & (f,x) &\to & j_x^k (f) .\end{array} \end{equation}
The result goes as follows.

\begin{prop} \label{prop:NeDependQueDesJets}
Let $E=C^\infty(M)$ and
$ V \subset E$ be an open subset and $k \in {\mathbb N}$ an integer.
\begin{enumerate}
\item The subset ${\mathfrak j}^{k} ( V \times  M )$ is an open subset of $j^k(M)$.
\item Let $c$ be a smooth $\bbK$-valued function on 
$ V \times M$, with $ V \subset E$ an open subset. Assume that
$c(f,x)$ depends only on the $k$-jet of the function $f$ at the point $x$. 
Then there exists a unique smooth $\bbK$-valued function 
$ {\tilde c}$ on the open subset  ${\mathfrak j}^k(V \times M)\subset j^k(M)$ 
that makes the following diagram commutative:
\begin{equation} 
\label{eq:diagramme}  
\xymatrix{ V \times M \ar[d]^{{\mathfrak j}^k}  \ar[r]^{c}& 
{\bbK}.
\\ {\mathfrak j}^k(V \times M) \ar@{.>}[ur]^{\tilde{c}}  & } 
\end{equation}
i.e, such that the relation  $c(f,x) = \tilde{c}(j^k_x(f))$ holds 
for all $f \in V$ and $x \in M$. 
\end{enumerate}
\end{prop}
 
When $V=E$, Proposition \ref{prop:NeDependQueDesJets} specializes to the following easier statement:

\begin{cor} \label{coro:NeDependQueDesJets}
Let $c$ be a smooth function from $ E \times M$ to 
$\bbK$.
Assume that there exists an integer $k$ such that
$c(f,x)$ depends only on the $k$-jet of the function $f$ at the point $x$. 
Then there exists a unique smooth $\bbK$-valued function 
$ {\tilde c}$ on  $ j^k(M)$ such that the following diagram commutes:
\begin{equation}
\label{eq:diagrammeCoro}  
\xymatrix{ E \times M \ar[d]^{{\mathfrak j}^k}  \ar[r]^{c}& 
{\bbK}.
\\ j^k(M) \ar@{.>}[ur]^{\tilde{c}}  & } 
\end{equation}
i.e. such that the relation  $c(f,x) = \tilde{c}(j^k_x(f))$ 
holds for all $f \in V$ and $x \in M$. 
\end{cor}

Before establishing these results, we shall need several lemmas.

\begin{lem}\label{lem:existsSection}
The vector bundle morphism $ {\mathfrak j}^k$ described in (\ref{eq:assignJet}) is surjective and admits a smooth section~${\mathfrak s}^k $.\end{lem}
\begin{proof}
The section ${\mathfrak s}^k $, when it exists, being by construction a right inverse of $ {\mathfrak j}^k$, the latter is surjective. It suffices therefore to prove the existence of ${\mathfrak s}^k $.

We first prove that the lemma holds true for $ M$ an open subset $V $ of $ {\mathbb R}^n$. In that case, 
 the bundle of $k$-jets $j^k_x (V)$ is isomorphic to the trivial bundle over $V$ with typical fiber the space of polynomials of degree less or equal to $k$.  There is an obvious candidate for the section of  $ {\mathfrak j}^k$: it consists in mapping $\alpha_x \in j^k_x (V)$ to the unique polynomial of degree $k$ whose $k$-jet at $x \in V$ is $\alpha$. The henceforth obtained assignment, that we denote by ${\mathfrak s}^k_{V} $, is a smooth vector bundle morphism from $j^k (V) \to V$ to the trivial bundle $C^\infty(V) \times V \to V$.  It is by construction a section of $ {\mathfrak j}^k$.

We now go back to the general case of an arbitrary manifold $M$. For every point $x \in M$, choose $V_x$ a coordinate neighbourhood  and let $\chi$ be a smooth function with compact support on $V_x$ which is identically equal to $1$ in a neighbourhood $V_x' \subset V_x$ of $x$. Since $V_x$ is a coordinate neighbourhood, it can be identified with an open subset of ${\mathbb R}^n $, which allows to consider
 $${\mathfrak s}^k_{V_x} : j^k(V_x) \mapsto C^\infty(V_x) \times V_x $$
as in the previous paragraph. We can then consider the composition of vector bundle morphisms over $V_x$:
 $$ \xymatrix{ j^k(M) |_{V_x} \simeq j^k(V_x)  \ar[r]^{{\mathfrak s}^k_{V_x}  \, \, \,} & C^\infty(V_x) \times V_x  
 \ar[r]^{  m_\chi \times id} &  E \times V_x  }$$
where $ \left. j^k(M) \right|_{V_x} \simeq j^k(V_x) $ is the obvious identification of the $k$-jet bundle of $V_x$
to the restriction to $V_x$ of the $k$-jet bundle on $M$, 
and where $ m_\chi $ is the smooth linear map from $C^\infty(V_x)$ to $E=C^\infty(M)$ defined by $m_\chi (f) = f \chi$.
Since $\chi=1$ identically equal to $1$ on $V_x'$, the restriction to $V_x' $ of this vector bundle morphism is by construction a section of the restriction of  $ {\mathfrak j}^k $ to $V_x'$.

Since the manifold $M$ is paracompact\footnote{
Since all 
manifolds in this article are assumed to be paracompact, 
for every property $T$ on the set of all open subsets on $M$, 
provided that \emph{(i)} open subsets of subsets satisfying 
$T$ satisfy $T$, and \emph{(ii)} every point of $M$ admits a 
neighbourhood that satisfies property $T$, then there exists an open 
cover $(U_i)_{i \in I}$ made of subsets satisfying the property $T$ 
that admits a partition of unity $(\chi_i)_{i \in I} $ relative to it.}, 
the latter point implies that the manifold $M$ can be covered by open 
subsets $(U_i)_{i \in I}$ such that the restriction of $ {\mathfrak j}^k $ 
to $U_i$ admits a section ${\mathfrak s}^k_i$. Without any loss of 
generality, we can assume the existence of a smooth partition of unity 
$(\chi_i)_{i \in I} $ relative to this open cover. A global smooth section 
of $ {\mathfrak j}^k$ is then given by the explicit formula 
$ {\mathfrak s}^k = \sum_{i \in I} \chi_i \, {\mathfrak s}^k_i $, as 
follows from the obvious computation:
 $$ {\mathfrak j}^k \circ {\mathfrak s}^k = \sum_{i \in I} \chi_i \,  
  {\mathfrak j}^k \circ {\mathfrak s}^k_i  
= \sum_{i \in I} \chi_i \, {\rm id}_{j^k(M)} = {\rm id}_{j^k(M)}$$
where we used the fact that $\mathfrak{j}^k$ commutes with multiplications
by $\chi_i$ since $\mathfrak{j}^k:E\times M\mapsto J^kM$
is a vector bundle morphism.
 This completes the proof.
\end{proof}

Since $c:V\times M\to \bbK$ is only defined on the
open subset $V\times M$ of $E\times M$, we need
the following refinement of Lemma~\ref{lem:existsSection}
where the local sections ${\mathfrak t}^k_x$
of $ {\mathfrak j}^k$ are valued in $V\times M$:
\begin{lem}\label{lem:existsSection2}
For every $ (f,x) \in E \times M$, the vector bundle morphism $ {\mathfrak j}^k$ described in (\ref{eq:assignJet}) admits a smooth section ${\mathfrak t}^k $ through\footnote{i.e. such that
 ${\mathfrak t}^k_x \circ {\mathfrak j}^k_x (f)  = f $.} $(f,x)$. 
\end{lem}
\begin{proof}
Notice that Lemma~\ref{lem:existsSection2} can be derived from 
Lemma \ref{lem:existsSection}
for any vector bundle morphism over the identity of $M$. 
A careful check shows that the arguments below are absolutely general and indeed show that for any two vector bundles $E_1,E_2$ over $M$, 
any vector bundle morphism $E_1\mapsto E_2$ over $M$ that admits a section 
is a submersion, and admits a section through every point of $E_1$.

We prefer to do it, however, in our particular setting -- since one 
of the bundles is infinite dimensional and requires careful attention.

Let $ {\mathfrak s}^k$ be a section of $ {\mathfrak j}^k $ as in 
Lemma \ref{lem:existsSection}.
Consider the smooth map defined
at all point  $ y \in M$ by
$$ \begin{array}{rrcl} {\mathfrak t}^k_y : 
  & j^k_y(M) &  \to &  (E \times  M)_y \simeq E, \\
 & \beta & \mapsto & {\mathfrak s}^k_y (\beta)+ 
( f - {\mathfrak s}^k_y \circ j^k_y(f) )
\end{array} $$
This map is smooth by construction. 
It is again a section ${\mathfrak j}^k$, as follows from 
the following computation, valid for all $y \in M, \beta \in j^k_y(M)$:
\begin{eqnarray*}  {\mathfrak j}^k_y \circ {\mathfrak t}^k_y (\beta) 
& = & {\mathfrak j}^k_y \big({\mathfrak s}^k_y (\beta)+   
( f - {\mathfrak s}^k_y \circ j^k_y(f) ) \big) \\
& = & {\mathfrak j}^k_y  \circ {\mathfrak s}^k_y (\beta)+  
{\mathfrak j}^k_y(f) - {\mathfrak j}^k_y \circ
{\mathfrak s}^k_y \circ j^k_y(f)
 \\
&=& \beta +   {\mathfrak j}^k_y (f) -  {\mathfrak j}^k_y (f) = \beta.
\end{eqnarray*}
Then the section ${\mathfrak t}^k $ above satisfies by construction:
$$ {\mathfrak t}^k_x \circ  {\mathfrak j}^k_x(f) = 
{\mathfrak s}^k_x \circ 
j^k_x(f)  + f - {\mathfrak s}^k_x \circ j^k_x(f) = f.$$
 This completes the proof.
\end{proof}

Lemma~\ref{lem:existsSection2} has the following immediate consequence.

\begin{lem}\label{lem:JIsSubmersion}
The vector bundle morphism $ {\mathfrak j}^k$ described in (\ref{eq:assignJet}) is  a submersion.
\end{lem}
\begin{proof}
For every $ (f,x)$ in $E \times M$, let us choose $ {\mathfrak t}^k $ to be a section through $(f,x)$ as in  
Lemma \ref{lem:existsSection2}. By construction, the differential of  $ {\mathfrak j}^k $
at $(f,x)$ admits the differential of $ {\mathfrak t}^k $ at $ {\mathfrak j}^k (f,x)$ as right inverse, so it is surjective. 
\end{proof}

We can now prove Proposition \ref{prop:NeDependQueDesJets}. 
\begin{proof}
Since the vector bundle morphism $ {\mathfrak j}^k$ described in (\ref{eq:assignJet})
is a submersion by Lemma \ref{lem:JIsSubmersion} and since $V \times M$ is open in $E \times M$, the subset ${\mathfrak j}^k(V \times M) $  is an open subset of $j^k(M)$. This proves the first item in Proposition \ref{prop:NeDependQueDesJets}. 

Let us now prove the second item. Assume that we are given a function 
$c:V \times M \mapsto \bbK$ such that the value $c(f,x)$ 
at an arbitrary $f \in E$ and $x \in M$ depends only the $k$-jet of $f$ at $x$.
The existence of an unique function $\tilde{c}$ from $ j^k(M)$ to 
$\bbK$ making the diagram (\ref{eq:diagramme}) commute
is simply a set-theoretic property: the difficulty is to show that this 
function $\tilde{c} $ is smooth.

When $V=E$ (i.e. under the assumptions of Corollary \ref{coro:NeDependQueDesJets}), the smoothness of $\tilde{c} $ 
follows directly from Lemma \ref{lem:existsSection}, which implies that the commutative diagram (\ref{eq:diagrammeCoro}) can be completed to 
$$  \xymatrix{ E \times M \ar@/^/[d]^{{\mathfrak j}^k}  \ar[r]^{c}& 
{\bbK}.\\ 
j^k(M) \ar@/^/[u]^{{\mathfrak s}^k} \ar@{.>}[ur]^{\tilde{c}}  & } $$
which amounts to say that the following relation holds
 \begin{equation}\label{eq:RetrouverTilde_c} \tilde{c}  = c \circ {\mathfrak s}^k.\end{equation}
The latter formula and the smoothness of ${\mathfrak s}^k $ implies that, when $c$ is assumed to be a smooth function, so is the function $ \tilde{c} $ by composition. This proves Corollary \ref{coro:NeDependQueDesJets}.

For the general case, we have to choose, for all 
$\alpha \in  {\mathfrak j}^k(V \times M) $ a section 
$ {\mathfrak t}^k$ of $ {\mathfrak j}^k$ such that 
$ {\mathfrak t}^k(\alpha) \in V \times M$. Such a section 
$ {\mathfrak t}^k$ always exists by Lemma \ref{lem:existsSection2}. 
Since ${\mathfrak t}^k$ is smooth, there exists a neighbourhood $W_\alpha$ of $\alpha $ in $ {\mathfrak j}^k(V \times M) $ 
on which ${\mathfrak t}^k$ takes values in the domain of definition 
$ V \times  M $ of $c$,
which implies that the commutative diagram (\ref{eq:diagramme}) can be completed to 
$$  \xymatrix{ 
(V\times M)\cap ({\mathfrak j}^k)^{-1}(W_\alpha) \ar@/^/[d]^{{\mathfrak
j}^k}  \ar[r]^{c}& {\bbK}.\\ 
W_\alpha \ar@/^/[u]^{{\mathfrak t}^k} \ar@{.>}[ur]^{\tilde{c}}  & } $$
 In turn, the  commutativity of this diagram gives the explicit description of $ {\tilde{c}} $ through the following formula, valid on $W_\alpha$:  
 \begin{equation}
\label{eq:RetrouverTilde_c2} 
\tilde{c} = c \circ {\mathfrak t}^k.
\end{equation}
Formula (\ref{eq:RetrouverTilde_c2}) and the smoothness of 
${\mathfrak t}^k $ imply that, when $c$ is assumed to be a smooth 
function, so is,  by composition, the restriction to $W_\alpha$ of 
the function $ \tilde{c} $. Since every 
$\alpha \in {\mathfrak j}^k(V \times M)$ admits a neighbourhood on which 
the restriction of $\tilde{c}$ is smooth, the function $\tilde{c}$ 
is a smooth function. This completes the proof. 
\end{proof}

\subsection{Properties of $F^{(2)}_\varphi$}

We first show that the two assumptions of our theorem
are equivalent to some strong assumptions on the second derivative of $F$:
\begin{lem}\label{hardLemma}
Let $U$ be an open subset of $C^\infty(M)$ and 
$F:U\to \bbK$ be Bastiani smooth. Assume that
for every $\varphi\in U$, the differential 
$DF_\varphi = F^{(1)}_\varphi $ of $F$ at $\varphi$ has no wave front set, i.e.
$\WF(F^{(1)}_\varphi)=\emptyset$. Then the two following properties are equivalent:
\begin{enumerate}
\item $F$ is additive and the map $\varphi\mapsto \nabla F_\varphi$ 
  is Bastiani smooth from $U$ to $\calD(M)$.
\item For every $\varphi_0\in U$, 
   there is a neighborhood $V$ of $\varphi_0$, a compact $K\subset M$ 
   and a finite family of Bastiani smooth maps 
   $f_\alpha:V\to \calD(K)$ with 
   $\vert\alpha\vert\leqslant k$, such that in any system of 
   local coordinates $(x,y)$ on $M^2$:
   \begin{eqnarray}
     F^{(2)}_\varphi(x,y) &=& \sum_{\vert\alpha\vert\leqslant k} 
     f_\alpha(\varphi)(x)
        \partial^\alpha_y \delta(x-y),
     \label{repF2}
   \end{eqnarray}
   for every $\varphi\in V$.
\end{enumerate}
\end{lem}

In particular, both conditions imply that $D^2F_\varphi$ is represented by a 
distribution $F^{(2)}_\varphi$ whose wave front set is the conormal bundle of 
the diagonal in $M^2$~[\onlinecite[p.~32]{Viet-16}].

In the sequel, we shall often use the following simple lemma~:
\begin{lem}
Let $E$, $F$ and $G$ be locally convex spaces.
If $f:E\to F$ is Bastiani smooth
and $\ell:F\to G$ is linear and continuous, then
$\ell\circ f:E\to G$ is Bastiani smooth
and $D^k(\ell\circ f) = \ell\circ D^k f$.
\end{lem}
\begin{proof}
This is a consequence of three facts:
the map $\ell$ is Bastiani smooth because it is linear and
continuous, $\ell\circ f$ is Bastiani smooth because it is
the composition of two Bastiani smooth maps and
the chain rule. 
\end{proof}

We also need the following lemma in the proof of Lemma \ref{hardLemma}:
\begin{lem}\label{trivialLemm}
Let $U$ be a convex open subset of $E=C^\infty(M)$
containing the origin and $F:U\to E$ a
Bastiani smooth map.
Then, $G:U\to E$ defined by
$G(\varphi)=\int_0^1 F(s\varphi)ds$ is Bastiani smooth.
\end{lem}
\begin{proof}
The first step is to
define a candidate for the Bastiani differential $D^kG$ by determining
$D^kG(x)$ pointwise in $x\in M$.  
For every $(t_1,\dots,t_k,x)\in [0,1]^k\times M$ and
$(\varphi,\psi_1,\dots,\psi_k)\in U \times E^k$, the function
$(t_1,\dots,t_k,x)\mapsto\int_0^1 ds 
F(s(\varphi+t_1\psi_1+\dots+t_k\psi_k
))(x)$ is smooth in $(t_1,\dots,t_k,x)$
by dominated convergence theorem since $x$ can always be restricted to some
compact subset $K\subset M$ to obtain uniform bounds. 
We can differentiate
in $(t_1,\dots,t_k)$ outside and inside the integral and both differentials
coincide. Therefore, for every $x\in M$, the Bastiani $k$th--differential
$D^kG(x)$ of $G(x)=\int_0^1 ds F(s\varphi(x))$ exists and satisfies the
relation
$D^kG_\varphi(\psi_1,\dots,\psi_k)(x)=\int_0^1 ds s^k
D^kF_{s\varphi}(\psi_1,\dots,\psi_k)(x)$.
Let us show that $D^kG:U\times E^k\mapsto E$ is jointly
continuous in $(\varphi,\psi_1,\dots,\psi_k)$.

We know that the map $\chi:(s,\varphi,\psi_1,\dots,\psi_k)\in [0,1]\times U\times
E^k\mapsto s^k D^kF_{s\varphi}(\psi_1,\dots,\psi_k)\in E$ 
is continuous by joint continuity of $D^kF: U\times E^k\mapsto E$ 
and composition of the continuous maps
$$(s,\varphi,\psi)
\mapsto (s\varphi, \psi)
 \mapsto s^k D^kF_{s\varphi}(\psi),$$
where $\psi=(\psi_1,\dots,\psi_k)$.
Then by [\onlinecite[Thm 2.1.5 p.~72]{Hamilton-82}] applied to the function
$\chi$, the
integrated map $(\varphi,\psi_1,\dots,\psi_k)
\mapsto \int_0^1 ds s^k D^kF_{s\varphi}(\psi_1,\dots,\psi_k)$
is continuous and the proof is complete because continuity
holds true for every $k$.
\end{proof}

Let us now prove Lemma \ref{hardLemma}.

\begin{proof}
First of all, by Proposition~\ref{localsecondderivativediagonal}, $F$ is additive 
iff its second derivative is represented by a distribution supported in the 
diagonal.
We start by proving the direct sense
assuming that $\varphi\in U\mapsto \nabla F_\varphi\in C^\infty(M)$ 
is Bastiani smooth.

We first show that item~1 implies item~2 in
Lemma~\ref{hardLemma}.
Since $F$ is Bastiani smooth for any $\varphi_0\in U$, 
we already know by
Proposition \ref{localcompactsupport}
that there is some neighborhood 
$V$ of $\varphi_0$ on which $F|_V$ has fixed compact support that we denote 
by $K$. Therefore, 
$\nabla F_\varphi$ belongs to $\calD(K)$ for every
$\varphi \in V$ and 
$F^{(2)}_\varphi$ is supported in $K\times K$.
Since $F^{(2)}_\varphi$ is also supported in 
the diagonal of $M^2$ 
by Proposition~\ref{localsecondderivativediagonal},
the support of $F^{(2)}_\varphi$ is contained in the
diagonal of $K\times K$ which can be identified with $K$ itself.

Since $DF_\varphi$ has an empty wavefront set by assumption, 
its singular support is empty and
it can be represented by a \emph{unique} smooth compactly supported
function $\nabla F_\varphi$~[\onlinecite[p.~37]{HormanderI}]
such that
\begin{equation}\label{gradientformula}
\frac{d}{dt}F(\varphi+th)|_{t=0}= DF_\varphi(h)=\int_M \nabla F_\varphi(x) h(x)dx .
\end{equation}
The main step is to represent $F^{(2)}_\varphi$ as the Bastiani 
differential of $\nabla F_\varphi$
by calculating the second derivatives in two different ways.
The Bastiani smoothness of $F$ yields:
\begin{eqnarray*}
D^2F_\varphi(g,h)&=&\frac{d^2}{dt_1dt_2}F(\varphi+t_1h+t_2g)|_{t_1=t_2=0}\\
&=&\frac{d}{dt_2}\left( \frac{d}{dt_1}
F(\varphi+t_1h+t_2g)|_{t_1=0}\right)|_{t_2=0}\\
&=&\frac{d}{dt_2}\left( \int_M \nabla F_{\varphi+t_2g}(x)
h(x)dx\right)|_{t_2=0},
\end{eqnarray*}
where we used the Schwarz lemma and Equation~(\ref{gradientformula}).
To justify switching $\frac{d}{dt_2}$ and integration over $M$,
observe that the map
$\varphi\in U\mapsto \nabla F_{\varphi}\in \calD(M)$
is Bastiani smooth hence $C^1$. It follows by the chain rule that 
$t\mapsto \frac{d}{dt}\nabla F_{\varphi+tg} $ is a $C^0$ map 
valued in $\calD(M)$.
Since $\nabla F_\varphi$ is actually in 
$\calD(K)$ for every $\varphi\in V$ and the topology
induced by $\calD(M)$ on $\calD(K)$ is the usual
Fr\'echet topology of $\calD(K)$, the map
$\nabla F$ is smooth from $V$ to
the Fr\'echet space $\calD(K)$.

Since $\calD(K)$ injects continuously in $(C^0(K),\pi_{0,K})$, 
this implies that $(t,x)\in [-1,1]\times K \mapsto \frac{d}{dt}
\nabla F_{\varphi+tg}(x)\in C^0([-1,1]\times K)$.  
Hence the integrand $\frac{d}{dt_2}\nabla F_{\varphi+t_2g}(x) h(x)$
is in $C^0([-1,1] \times K)$ and 
is bounded on the integration domain. A continuous map $u:t\in [-1,1]\mapsto u(t,.)\in(C^{0}(K),\pi_{0,K})$ corresponds to a map
also denoted by $u\in C^0([-1,1]\times K)$. Indeed for every
convergent sequence $(t_n,x_n)\underset{n\rightarrow \infty}{\rightarrow}(t,x)$ in $[-1,1]\times K$, the simple estimate 
$\vert u(t,x)-u(t_n,x_n) \vert\leqslant  \vert u(t,x)-u(t,x_n) 
\vert+\vert u(t,x_n)-u(t_n,x_n) \vert\leqslant 
\vert u(t,x)-u(t,x_n) \vert+\pi_{0,K}(u(t_n,.)-u(t,.))$ 
shows that $u(t_n,x_n)\underset{n\rightarrow \infty}{\rightarrow} u(t,x)$.

By the dominated convergence theorem, we can differentiate under the 
integral sign:  
\begin{eqnarray*}
D^2F_\varphi(g,h)&=&\frac{d}{dt_2}\left( \int_M \nabla 
F_{\varphi+t_2g}(x) h(x)dx\right)|_{t_2=0}\\
&=&\int_M \left( \frac{d}{dt_2}\nabla F_{\varphi+t_2g}(x)|_{t_2=0}\right)
 h(x)dx.
\end{eqnarray*}
By definition  $\frac{d}{dt_2}\nabla F_{\varphi+t_2g}|_{t_2=0}$ is only
the Bastiani derivative 
\begin{eqnarray}
D\nabla F:(\varphi,g)\in V\times C^\infty(M) \mapsto D\nabla 
F_\varphi[g]\in \mathcal{D}(K),
\label{DnablaF}
\end{eqnarray}
where $D\nabla F$ is a Bastiani smooth map since 
$\nabla F$ is Bastiani smooth.

Note also that by Theorem~\ref{multilinearkernelbundles}, the second derivative
$D^2F_\varphi(g,h)$ can be represented by a map
$\varphi\in U \mapsto F^{(2)}_\varphi\in \mathcal{E}^\prime(M^2) $ 
such that
$\forall (\varphi,g,h)\in U\times C^\infty(M)^2,\,\ D^2F_\varphi(g,h)=
\left\langle F^{(2)}_\varphi,g\otimes h\right\rangle$
which means that $F^{(2)}_\varphi$ is the distributional kernel of the
second derivative $D^2F_\varphi$.
We now arrive at the following equality which identifies two different
representations of the second derivative.
For every $(\varphi,g,h)\in U\times C^\infty(M)^2$
\begin{eqnarray}\label{identitysecondderivatives}
\left\langle F^{(2)}_\varphi,g\otimes h\right\rangle=
\int_M D\nabla F_{\varphi}[g](y) h(y)dy.
\end{eqnarray}
By the same Theorem~\ref{multilinearkernelbundles} and the chain rule, 
$D\nabla F_{\varphi}[g](y)=ev_y D\nabla F_{\varphi}[g]$
is linear continuous in $g\in C^\infty(M)$,
hence there is a distribution, denoted by $D\nabla F_{\varphi}(x,y)$,
such that $\int_MD\nabla F_{\varphi}(x,y)g(x)dx=
D\nabla F_{\varphi}[g](y)$
and $\int_MD\nabla F_{\varphi}(x,y)g(x)dx$ is
in $\calD(K)$ by Eq.~(\ref{DnablaF}).
Since the above identity holds for all $(g,h)\in C^\infty(M)^2$, we have
in the sense of distributions that
$F^{(2)}_\varphi(x,y)=D\nabla F_{\varphi}(x,y)$
where the map
\begin{eqnarray}
(\varphi,g)\in V\times C^\infty(M)\mapsto \int_M 
F^{(2)}_\varphi(x,\cdot)g(x)dx\in \calD(K),
\nonumber\\\label{smoothpushforward}
\end{eqnarray}
is Bastiani smooth. 

It suffices to do the last part of the proof, which is local in nature,
on   $M=\bbR^d$.
We now represent $F^{(2)}_\varphi(x,y)$ as a $C^\infty(M)$-linear 
combination of derivatives of 
Dirac distributions concentrated on the diagonal.
By Proposition \ref{localsecondderivativediagonal}, 
the additive property satisfied by $F$ implies that 
the distribution $F^{(2)}$ associated to the 
second derivative $D^2F$ is supported in the diagonal $D_2\subset M\times M$.
By Proposition \ref{bounded-order}, the kernel $F^{(2)}_\varphi(x,y) \in 
\calE^\prime(M\times M)$ 
has bounded distributional order uniformly in $\varphi\in V$.
Schwartz' theorem on distributions supported on a 
submanifold~[\onlinecite[p.~101]{Schwartz-66}] states 
that in local coordinates, there exists a finite sequence of
distributions $\left(\varphi\in V\mapsto f_\alpha(\varphi,.)\in 
\calD(K)\right)_{\vert\alpha\vert\leqslant k}$
such that $F^{(2)}_\varphi=\sum_{\vert\alpha\vert\leqslant k} 
     f_\alpha(\varphi,x) \partial^\alpha_y \delta(x-y)$.
We denote the distributions $f_\alpha(\varphi)$ by
$f_\alpha(\varphi,x)$ because we shall show that
$\varphi\mapsto f_\alpha(\varphi)$ is Bastiani smooth
from $V$ to $\calD(K)$.

By Equation~(\ref{smoothpushforward}), we know 
 that for every $(\varphi,g)\in V\times C^\infty(M)$
the map from $V\times C^\infty(M)$ to $\mathcal{D}(K)$
defined by
$$(\varphi,g)\mapsto
\int_M F^{(2)}_\varphi(x,\cdot)g(x)dx 
=\sum_{\vert\alpha\vert\leqslant k}(-1)^{\vert \alpha\vert}  
  f_\alpha(\varphi,.)\partial^{\alpha}g(.)$$
is smooth.
Choosing $g$ to be equal to the Fourier oscillatory function 
$e^{-i\langle \xi.x\rangle}$, we obtain
by the chain rule 
the maps that sends $(\varphi,\xi)$ to
\begin{eqnarray*}
\int_M F^{(2)}_\varphi(x,y)
   e^{-i\langle \xi.x\rangle}dy &=&
\int_M\sum_{\vert\alpha\vert\leqslant k}(-1)^{|\alpha|} 
  f_\alpha(\varphi,y)
\\&&
\delta(x-y)\partial^\alpha_y e^{-i\langle \xi.x\rangle}dy\\
&=&\sum_{\vert\alpha\vert\leqslant k}(-1)^{\vert \alpha\vert}  
f_\alpha(\varphi,x)(-i\xi)^{\alpha},
\end{eqnarray*}
is Bastiani smooth. 
Moreover, since the image of the map 
in Eq.~\eqref{smoothpushforward} is in $\calD(K)$ for
every smooth $g$, we obtain that
$\sum_{\vert\alpha\vert\leqslant k}(-1)^{\vert \alpha\vert}  
f_\alpha(\varphi,\cdot)(-i\xi)^{\alpha}$ is in
$\calD(K)$ for every $\xi$. 
This is only possible if $f_\alpha(\varphi)\in \calD(K)$
for every $|\alpha|\leqslant k$.
Therefore $\varphi\mapsto 
f_\alpha(\varphi)=(i\frac{d}{d\xi})^\alpha\int_M F^{(2)}_\varphi(.,y)
e^{-i\langle \xi.y\rangle}dy|_{\xi=0}$ is Bastiani smooth 
from $V$ to $\calD(K)$ and the proof of the direct sense is complete.

Conversely, we want to prove that
if there is a neighborhood $V$ of $\varphi_0$, a compact $K\subset M$ 
and a finite family of smooth maps 
$\varphi\mapsto f_\alpha(\varphi)\in \mathcal{D}(K),
\vert\alpha\vert\leqslant k$ 
such that in any system of local coordinates $(x,y)$ on $M^2$:
$$F^{(2)}_\varphi(x,y)=\sum_{\vert\alpha\vert\leqslant k} 
f_\alpha(\varphi)(x)\partial^\alpha_y\delta(x-y),$$
then
$\varphi\mapsto \nabla F_\varphi\in\mathcal{D}(M)$
is Bastiani smooth.
Without loss of generality, we assume that $V$ is convex.
By the Taylor formula with remainder for Bastiani smooth functions, 
for every $(\varphi,\psi_1,\psi_2)\in V\times C^\infty(M)^2$:
\begin{eqnarray*}
D^2F_{\varphi+s_1\psi_1+s_2\psi_2}(\psi_1,\psi_2) &=&
\partial_{s_1}\partial_{s_2}F(\varphi+s_1\psi_1+s_2\psi_2)
\\&=&
\partial_{s_2}DF_{\varphi+s_1\psi_1+s_2\psi_2}(\psi_1),
\end{eqnarray*}
for $s_1$ and $s_2$ small enough.
It follows by the fundamental theorem
of calculus and by evaluating at $s_1=0$
the previous relation
that
\begin{eqnarray*}
DF_{\varphi+t\psi_2}(\psi_1) &=&
DF_\varphi(\psi_1)+\int_0^t \partial_{s_2} DF_{\varphi+s_2\psi_2}(\psi_1)ds_2 
\\&=&
 DF_\varphi(\psi_1)+\int_0^tD^2F_{\varphi+s\psi_2}(\psi_1,\psi_2)ds,
\end{eqnarray*}
where by assumption 
$DF_\varphi(\psi_1)$ is represented by integration against a smooth function
\begin{eqnarray*}
DF_\varphi(\psi_1) &=& \int_M \nabla F_\varphi(x)\psi_1(x)dx,\\
D^2F_{\varphi+s\psi_2}(\psi_1,\psi_2) &=& \int_{M\times M} 
F^{(2)}_{\varphi+s\psi_2}(x,y)\psi_1(x)\psi_2(y)dxdy,
\end{eqnarray*} 
and 
$F^{(2)}_{\varphi+s\psi}$ is
supported on a subset of the diagonal $D_2\subset M\times M$
that can be identified with $K$.
Hence, for $\psi\in C^\infty(M)$ such that
$\varphi+\psi \in V$:
\begin{eqnarray*}
\nabla F_{\varphi+\psi}(x) &=&
\nabla F_\varphi(x)+\int_0^1\left(\int_{M} 
  F^{(2)}_{\varphi+s\psi}(x,y) \psi(y)dy\right)ds\\
&=& \nabla F_\varphi(x)+
\sum_{\vert\alpha\vert\leqslant k} (-1)^{|\alpha|}
\partial^\alpha\psi(x)
\\&&
\int_0^1
f_{\alpha}(\varphi+s\psi)(x)ds .
\end{eqnarray*}
To show that the map $\chi:V\to \calD(K)$ defined by
$\chi(\psi)=\nabla F_{\varphi+\psi}$ is smooth,
we notice that, according to the last equation,
$\nabla F_{\varphi+\psi}$
is the sum of the constant $\nabla F_\varphi$
and a finite linear combination of products of
$\psi\mapsto \partial^\alpha\psi$ 
by an integral over $s$.
The integrand $f_{\alpha}(\varphi+s\psi)$ is
smooth by assumption.
Therefore, the map $$\psi\in (V-\varphi)\mapsto \int_0^1
f_{\alpha}(\varphi+s\psi)(x)ds\in \mathcal{D}(K) $$  is smooth
by Lemma~\ref{trivialLemm}
and the fact that the topology induced on
$\calD(K)$ by the topology of $C^\infty(M)$
is the standard topology of $\calD(K)$.
The map $\psi\mapsto \partial^\alpha\psi$ is 
smooth because it is linear and continuous.
Finally, the product of the integral by
$\partial^\alpha\psi$ is smooth by a trivial
extension of Lemma~\ref{lemme-produit}.
This completes the proof of Lemma \ref{hardLemma}.
\end{proof}

We are now ready to prove Theorem~\ref{TheoPrincipal}
characterizing local functionals.

\subsection{Proof of Theorem~\ref{TheoPrincipal}}
 \label{sec:proofPrincipal}

Let us start by proving the
converse sense where we assume that $F$ is the integral 
of some local function on jet space. Let $\varphi\in U$ and $V$
some neighborhood of $\varphi$ such that
$F(\varphi+\psi)=\int_Mf(x,j^k_x\psi)dx$ for every $\psi\in V$ where 
$j_x^k\psi$ is the $k$-jet of $\psi$ at $x$ and where $f$ is smooth
and compactly 
supported in the variable $x$ in some fixed compact $K\subset M$. 
Without loss of generality, 
we can restrict the support $K$ of $f$ by a smooth partition of unity and 
assuming that $K$ is contained in some open chart of $M$,
we may reduce to the same problem for $f\in C^\infty(\Omega)$ where 
$\Omega$ is some open set in $\bbR^d$ and $K\subset \Omega$.

We choose a smooth compactly supported function $\chi\in \mathcal{D}(\Omega)$ 
such that $\chi=1$ on a compact neighborhood of $K$
with $\supp\chi\subset \Omega$  and we observe that
$$\Psi:\psi\in C^\infty(\Omega)\longmapsto 
\left(\partial^\alpha\psi\right)_{\vert\alpha\vert\leqslant k}\chi\in 
\mathcal{D}(\mathrm{supp}(\chi))^{\frac{(d+k)!}{d!}},$$
is linear continuous hence Bastiani smooth. 
We need a simple
\begin{lem}
Let $\Omega$ be an open set in $\bbR^d$ 
then the map
$$\Phi:\varphi\in C^\infty(\Omega,\bbR^r)\mapsto 
\{ x\mapsto(x,\varphi(x)) \}\in C^\infty(\Omega,\bbR^d\times 
\bbR^r),$$
is Bastiani smooth.
\end{lem}
\begin{proof}
The first Bastiani differential $D\Phi_\varphi(h)$
can be identified with the smooth function
$x\mapsto (0,h(x))$, which 
is linear continuous in $h$ and does not depend on $\varphi$.
It is thus smooth and so is $\Phi$. 
\end{proof}

Therefore, the composition
$$\Phi\circ\Psi: C^\infty(\Omega)\to C^\infty(\Omega,\bbR^d \times 
  \bbR^{\frac{(d+k)!}{d!}})$$ defined by
$$\Phi\circ\Psi(\psi):x\mapsto
\big(x,\chi\partial^\alpha\psi(x)_{\vert\alpha\vert\leqslant k}\big),$$ 
is Bastiani smooth
and finally
$$ \psi\in C^\infty(M) \mapsto f(.,j_x^k\psi(.))\in\mathcal{D}(K)
\mapsto \int_{\Omega}f(x,j^k_x\psi)dx,$$
is Bastiani smooth by the chain rule and since the last integration map 
is linear continuous thus Bastiani smooth.

Now let us prove the direct sense of Thm.~\ref{TheoPrincipal},
where we start from a functional 
characterization of $F$ and end up with a representation 
as a function $F(\varphi+\psi)=\int_Mf(x,j^k_x\psi)dx$ on jet space,
for $\varphi+\psi$ in a neighborhood $V$ of $\varphi$,
that we assume convex.
We start by deriving a candidate for the function $f$.
According to the fundamental
theorem of calculus,
\begin{eqnarray} \label{eq:JusteTheoFOnd}
F(\varphi+\psi) &
= & F(\varphi)+
\int_0^1 dt 
DF_{\varphi+t\psi} (\psi).
\label{Fphi}
 \end{eqnarray}
As discussed at the beginning of this section,
since we assume that 
$\WF(F^{(1)}_\varphi)=\emptyset$
for every $\varphi\in U$, 
there exists
a unique smooth compactly supported function
$ x\mapsto \nabla F_\varphi (x) $ such that:
\begin{eqnarray}
 F^{(1)}_\varphi (\psi)  &=& \int_M dx \nabla F_\varphi (x) \psi(x).
\label{F1nabla}
\end{eqnarray}
Therefore equation~(\ref{eq:JusteTheoFOnd}) reads:
 \begin{eqnarray} \label{eq:JusteTheoFOnd2}
F(\varphi+\psi) & =& F(\varphi)+
\int_0^1 dt 
 \int_M   \nabla F_{\varphi+t\psi}(x) \psi(x) dx.
 \end{eqnarray}
We show that Fubini's theorem can be applied to the function 
$\chi:(x,t) \mapsto  \nabla_{\varphi+t\psi} F (x) \psi(x)$.
By Prop.~\ref{localcompactsupport}, $F^{(1)}$ is locally compactly
supported, so that there is a convex neighborhood $V$ of $\varphi$
and a compact subset $K$ of $\Omega$
such that $F^{(1)}_{\varphi+\psi}$ is supported in $K$
for every $\varphi+\psi\in V$.
The function $\chi$ is defined on ${[}0,1{]}\times K$
and supported on $K$ for fixed $t\in {[}0,1{]}$.
Moreover, by imposing the additional assumption 
carried by item 2 in Theorem~\ref{TheoPrincipal},
namely that $\varphi \mapsto \nabla F_\varphi$ be Bastiani smooth
from $U$ to $\calD(M)$,
the support property of $F$ implies that 
the image of $\nabla F_{\varphi+t\psi}$ is actually in $\calD(K)$ and 
$\nabla F$ is smooth from $V$ to $\calD(K)$
because the topology induced on $\calD(K)$ by
$\calD(M)$ is the Fr\'echet topology of $\calD(K)$
determined by the seminorms $\pi_{m,K}$~[\onlinecite[p.~172]{Horvath}].
Since $\calD(K)$ injects continuously in $(C^0(K),\pi_{0,K})$,   
$\varphi \mapsto D_{\varphi} F$ is a continuous $(C^0(K),\pi_{0,K})$-valued 
map.  This implies that $(t,x) \mapsto \nabla F_{\varphi+t\psi}(x)$ 
is continuous as a $\bbK$-valued function on $[0,1] \times K$.
Hence so is the integrand of (\ref{eq:JusteTheoFOnd2}),
Fubini theorem holds and we obtain:
 
\begin{lem}\label{lem:Fubini}
Let $U$ be an open subset of $E=C^\infty(M)$
and $F:E\to \bbK$ be Bastiani smooth.
Assume that for every $\varphi\in U$, $\WF(F^{(1)}_\varphi)=\emptyset$  
and  $F^{(1)}:U\to \calD(M)$ is Bastiani smooth,
then, for every $\varphi\in U$, there is a convex neighborhood
$V$ of $\varphi$ such that, if $\varphi+\psi\in V$, then
\begin{eqnarray}\label{representationformula}
F(\varphi+\psi) &
= &F(\varphi) + \int_M dx \int_0^1 \nabla F_{\varphi+t\psi}(x) \psi(x) dt.
\end{eqnarray}
\end{lem}
From now on, we consider $\varphi\in U$ to be fixed.
Our candidate for $f(j_x^k\psi)$ is
\begin{eqnarray}
c_\psi(x) &=& \int_0^1 \nabla F_{\varphi+t\psi}(x) dt \,  \psi(x).~
\label{cpsi}
\end{eqnarray}
By definition and Lemma\ref{lem:Fubini}, for all
$\psi$ such that $\varphi+\psi\in V$,
$$F(\varphi+\psi)=F(\varphi)+\int_M c_\psi(x) dx.$$
To show that $c_\psi(x)$ is the right candidate
we first need
\begin{prop}\label{finiteJetdep}
The function $c_\psi$ depends only on a
finite jet of $\psi$.
More precisely, for every
$\varphi\in U$, there is a convex neighborhood $V$ of 
$\varphi$ and an integer $k\ge 0$ such that, 
for all $x\in M$,
for every $\psi_1$ and $\psi_2$ such that
$\varphi+\psi_1$ and $\varphi+\psi_2$ are in $V$ and
$j^k_x\psi_1=j^k_x\psi_2$, then
$c_{\psi_1}(x)=c_{\psi_2}(x)$.
\end{prop}

The beginning of the proof is inspired by
Ref.~\onlinecite{Brunetti-12}.
For fixed $\varphi\in U$,
by Proposition \ref{bounded-order}, 
there exists an integer $k$, a compact $K$ and
a convex neighborhood $V$ of $\varphi$ such that
the order of $F^{(2)}_{\varphi+\psi}$ is smaller than
$k$ and the support of $DF_{\varphi+\psi}$ is in $K$
if $\varphi+\psi\in V$.

Let us choose some point $x_0 \in M$.
Consider a pair $\psi_1,\psi_2$ of smooth functions such that 
$\psi_1(x_0)=\psi_2(x_0)$. Then,
\begin{eqnarray*}
 c_{\psi_1}(x_0)-c_{\psi_2}(x_0)&=& 
\int_0^1 dt 
\big(\nabla F_{\varphi+t\psi_2} (x_0) \, \psi_2(x_0) 
\\&&
-\nabla F_{\varphi+t\psi_1} (x_0) \, \psi_1(x_0)\big)
\\ &=& \psi_1(x_0) \int_0^1 dt \big(\nabla F_{\varphi+t\psi_2} (x_0)
\\&&
-\nabla F_{\varphi+t\psi_1}(x_0)\big).
\end{eqnarray*}

We use the fundamental theorem of analysis again
for $D F_\varphi (h)=\int_M dx\nabla F_{\varphi}(x)h(x)$ 
for an arbitrary $h\in C^\infty(M)$ to get
\begin{eqnarray*}
DF _{\varphi+t\psi_2}  (h)
-DF_{\varphi+t\psi_1} (h)  &=&
t\int_0^1 ds 
  \langle F^{(2)}_{\varphi+t\psi_1+st(\psi_2-\psi_1)},
\\&&
    (\psi_2-\psi_1)\otimes h \rangle.
\end{eqnarray*}
Now we take a sequence of smooth functions $(h_n)_{n\in \bbN}$ 
which  converges to $\delta_{x_0}$ in $\calD'(M)$ when $n$ goes 
to infinity and show that both the left and right hand side
have limits.
For the left hand side, the distribution
$DF_{\varphi+t\psi_i}$ being smooth, it defines
the  continuous form 
$u\mapsto DF_{\varphi+t\psi_i}(u)$ on $\calD'(M)$ by
duality pairing. By continuity,
$DF_{\varphi+t\psi_i}(h_n)\to DF_{\varphi+t\psi_i}(\delta_{x_0})$
and Eq.~(\ref{F1nabla}) yields
\begin{eqnarray*}
h &=& DF _{\varphi+t\psi_2}  (\delta_{x_0})
-DF_{\varphi+t\psi_1} (\delta_{x_0}) 
\\ &=& \int_M dx 
\big(\nabla F_{\varphi+t\psi_2}(x)-
\nabla 
  F^{(1)}_{\varphi+t\psi_1}(x)\big)\delta(x-x_0)
\\&=&
 \nabla F_{\varphi+t\psi_2}(x_0)-
\nabla F_{\varphi+t\psi_1}(x_0).
\end{eqnarray*}

For the right hand side, we know by Lemma~\ref{hardLemma} that for 
every $s\in[0,1]$, the wave front set of the distribution 
$F^{(2)}_{\varphi+t\psi_1+st(\psi_2-\psi_1)}$
is in the conormal $C_2$ and the sequence
$(\psi_2-\psi_1)\otimes h_n$ converges to 
$(\psi_2-\psi_1)\otimes \delta_{x_0}$ in
$\mathcal{D}^\prime_{N^*\left(M\times \{x_0\}\right)}$,
where $N^*\left(M\times \{x_0\}\right)$ is the conormal
of the submanifold $M\times\{x_0\}\subset M\times M$
in $T^*\left(M\times M\right)$.
Therefore, by transversality of the wave front sets
and hypocontinuity of the duality pairings~\cite{Viet-wf2},
the following limit exists
$$\lim_n\langle F^{(2)}_{\varphi+t\psi_1+st(\psi_2-\psi_1)},
    (\psi_2-\psi_1)\otimes h_n \rangle.$$

Moreover, still by Lemma~\ref{hardLemma}, we have
for $\varphi+\psi\in V$:
\begin{eqnarray*}
\langle F^{(2)}_{\varphi+\psi},g\otimes h\rangle &=& 
\sum_{|\alpha|\le k} (-1)^{|\alpha|}
  \int_M dx \theta^\alpha_\psi(x) g(x) \partial^\alpha h(x),
\end{eqnarray*}
for every $(g,h)\in C^\infty(M)^2$ and all $\theta^\alpha_\psi$ 
belong to $\calD(K)$.
An integration by parts yields
\begin{eqnarray*}
\langle F^{(2)}_{\varphi+\psi},g\otimes h\rangle &=& 
\sum_{|\alpha|\le k} 
  \int_M dx f^\alpha_\psi(x) h(x) \partial^\alpha g(x),
\end{eqnarray*}
where $f^\alpha_\psi=\sum_{\beta} \binom{\beta}{\alpha} 
\partial^{\beta-\alpha} \theta^\beta_\psi$
where the sum is over the multi-indices such that
$\beta\ge \alpha$ and $|\beta|\leqslant k$.

As a consequence, for $\varphi+\psi_1+\psi_2$ in the convex neighborhood $V$:
\begin{eqnarray*}
X&=&\langle F^{(2)}_{\varphi+t\psi_1+st(\psi_2-\psi_1)},
    (\psi_2-\psi_1)\otimes \delta_{x_0} \rangle \\&=& 
\sum_{|\alpha|\le k} 
  f_{t\psi_1+st(\psi_2-\psi_1)}^\alpha(x_0) 
   \partial^\alpha (\psi_2-\psi_1)(x_0).
\end{eqnarray*}
If, at the point $x_0$, $j^k_{x_0}\psi_1=j^k_{x_0}\psi_2$, then 
\begin{eqnarray*}
c_{\psi_1}(x_0)-c_{\psi_2}(x_0) &=& 
\psi_1(x_0) \sum_{|\alpha|\le k} 
\int_0^1 t dt 
\int_0^1 ds \, \\&&\hspace{-2cm}
f_{t\psi_1+st(\psi_2-\psi_1)}^\alpha(x_0)
\big(\partial^\alpha \psi_2(x_0)-\partial^\alpha \psi_1(x_0)\big)
=0.
\end{eqnarray*}
We showed that $c_{\psi}$ depends only on 
the $k$-jet of $\psi$ at $x_0$. Moreover, the number $k$ depends
only on $V$ and not on $x_0$, so that $c_\psi$ depends
on the $k$-jet for every $x\in M$.
In other words there is an integer $k$ and a function
$f$ such that $c_\psi(x)=f(x,\psi(x),\dots,\partial^\alpha\psi(x))$
for every $x\in M$, where $1\le |\alpha|\le k$.  

We want to show that $f$ is smooth in its arguments hence
we now investigate in which manner $c_\psi$ depends on $\psi$.
This suggests to study the regularity of the 
${\mathcal D}(K)$-valued function $\psi \mapsto 
\nabla F_{\varphi+t \psi} \psi $. 
More precisely, we need to show that the map
$\psi\mapsto \nabla F_{\varphi+t\psi}\psi$ 
is Bastiani smooth
from $U$ to $\calD(K)$. This is not completely trivial
because the map $x\mapsto \nabla F_{\varphi+t\psi}(x)$
is in $\calD(K)$ and $\psi$ in $C^\infty(M)$ and we must check that the
product of a function in $\calD(K)$ by a function
in $C^\infty(M)$ is continuous~[\onlinecite[p.~119]{Schwartz-66}].

\begin{lem}
\label{lemme-produit}
If $U$ is an open set in $C^\infty(M)$ and
$F:U\to\calD(K)$ is a compactly supported Bastiani-smooth map, 
then the function $G:U\to\calD(M)$ defined by
$G(\varphi)=F(\varphi)\varphi$ is compactly 
supported Bastiani-smooth with the same support as $F$.
\end{lem}
\begin{proof}
Yoann Dabrowski pointed out to us the following fact.
For any compact subset
$K$ of $\Omega$, both $\calD(\Omega)$
and $C^\infty(\Omega)$ induce on $\calD(K)$ the
usual Fr\'echet
topology of $\calD(K)$~[\onlinecite[p.~172]{Horvath}].
Thus to establish the smoothness
of $G$, it suffices to show that
the multiplication
$(u,v)\in \calD(K)\times C^\infty(M)\mapsto \calD(K)$
is continuous then it would be Bastiani smooth 
and by the chain rule
it follows that $\varphi\mapsto (F(\varphi),\varphi)\mapsto F(\varphi)\varphi$
is smooth.
Since both $\calD(K)$ and $C^\infty(M)$ are Fr\'echet,
the product $\calD(K)\times C^\infty(M)$
endowed with the product topology is 
metrizable and it is enough to prove that
the product is sequentially continuous.
Indeed, let $(u_n,v_n)\rightarrow (u,v) $ in $\calD(K)\times C^\infty(M)$,
we can find some cut--off function $\chi\in\mathcal{D}(M)$
such that $\chi=1$ on the support of all $u_n$, and for all $m$, 
by [\onlinecite[p.~1351]{Dabrouder-13}]
\begin{eqnarray*}
\pi_{m,K}(uv-u_nv_n) &\leqslant &
\pi_{m,K}\left((u-u_n)v\chi\right)
\\&& +\pi_{m,K}\left(u_n\chi(v-v_n)\right) 
\\
&\leqslant &
2^m\big(\pi_{m,K}(u-u_n)\pi_{m,K}(\chi v) \\&& +
\pi_{m,K}(u_n)\pi_{m,K}((v-v_n)\chi)\big)\rightarrow 0.
\end{eqnarray*}
Hence $G$ is smooth.
\end{proof}
This implies that
$\psi\mapsto c_\psi=\int_0^1 \nabla F_{\varphi+t\psi}\psi dt$
is smooth since the above Lemma shows the smoothness of 
$a(t,\psi)\mapsto \nabla F_{\varphi+t\psi}\psi$
and integration over $t$ conserves smoothness by Lemma~\ref{trivialLemm}.
\bcam At this point, Theorem~\ref{TheoPrincipal} follows directly from Proposition  \ref{prop:NeDependQueDesJets}.
\ecam

\subsection{Representation theory of local functionals.}

In this section, we discuss the issue of representation
of our local functionals and the relations between the
functionals $(c_\psi,\nabla F,f(j_x^k\psi))$ which are defined or constructed in the course of our proof of Theorem~\ref{TheoPrincipal}.
In the sequel, we assume that 
our manifold $M$ is connected, oriented without boundary, hence we can fix
a density $dx$ on $M$ which is also a differential form on $M$ of top degree.

In the sequel, 
we shall work out all explicit formulas in local charts which means
without loss of generality that we work on $\mathbb{R}^d$ and the reference
density $dx$ is chosen to be the standard Lebesgue measure. 
We will denote by $(x,u,u^\alpha)_{\vert\alpha\vert\leqslant k}$ where $\alpha$ are multi--indices, some local coordinates
on the jet bundle $J^k(\mathbb{R}^d)$.
Introduce the vertical Euler vector field $\rho=\sum u^{(\alpha)} \frac{\partial}{\partial u^{(\alpha)}} $
on the bundle $J^k(\mathbb{R}^d)$. In the manifold case if we work on $J^k(M)$, 
this vector field
is intrinsic since it generates scaling in the fibers of $J^k(M)$.
For all multiindex $(\alpha)=(\alpha_1\dots \alpha_p),\alpha_i\in \{1,\dots,d\}$, 
introduce the operators
$\partial^{(\alpha)}=\partial^{\alpha_1}\dots \partial^{\alpha_p}$
where $\partial^i=\frac{\partial}{\partial x^i}+\sum_{\alpha}u^{(\alpha i)}\frac{\partial}{\partial u^{(\alpha)}}$ and the Euler--Lagrange operator
$EL=u\frac{\partial}{\partial u}+\sum_{\alpha} (-1)^{\vert \alpha\vert}\partial^{(\alpha)}u^{(\alpha)}\frac{\partial }{\partial u^{(\alpha)}}$.
Let us discuss the nature of the objects involved, $\rho$ is a vertical vector field and acts on $C^\infty(J^kM)$ as a $C^\infty(M)$ linear map, for $\chi\in C^\infty(\mathbb{R}^d), f\in C^\infty(J^k\mathbb{R}^d)$, $\rho(\chi f)=\chi (\rho f)$. For every $i\in\{1,\dots,d\}$, $\partial^i$ is a vector field on $J^k\mathbb{R}^d$ but it has a horizontal component, therefore
it is not $C^\infty(\mathbb{R}^d)$ linear and 
the Euler--Lagrange operator is not $C^\infty(\mathbb{R}^d)$--linear either.

What follows is a definition--proposition where we give
an intrinsic and global definition of the Euler--Lagrange operator
in terms of the operator $\nabla F$ associated to a functional.
\begin{prop}[Euler-Lagrange operator is intrinsic]
Let $U$ be an open subset of $C^\infty(M)$ and $F:U\to \bbK$ a Bastiani smooth
\emph{local} functional.
For $\varphi\in U$, if there is an integer $k$, a neighborhood 
$V$ of $\varphi$, an open subset $\mathcal{V}$ of $J^kM$ and $f\in C^\infty(\mathcal{V})$
such that
$x\mapsto f(j^k_x\psi)$ is compactly supported and
\begin{eqnarray*}
F(\varphi+\psi) = F(\varphi)+\int_M f(j^k_x\psi) dx
\end{eqnarray*} 
whenever $\varphi+\psi\in V$ then
in every local chart 
\begin{eqnarray}
\nabla F_\varphi=EL(f)(j^k\psi)
\end{eqnarray}
where 
$EL(f)(\psi)=\sum_{\vert \alpha\vert\leqslant k}(-1)^{|\alpha|}\left(\partial^{(\alpha)}\left(
  \frac{\partial f}{\partial u^{(\alpha)}}\right)\right)(j^k_x\psi) $
is the Euler--Lagrange operator and $EL(f)(\psi)$ is uniquely
determined by $F$.
\end{prop}

The above proposition means that
$EL(f)$ does not depend on the choice of representative $f$ and is intrinsic
(i.e. it does not depend on the choice of a local chart).
\begin{proof}
Indeed, assume that we make a small perturbation $\varphi+\psi$ 
of the background field $\varphi$
by $\psi$ which is compactly supported in some open chart $U$ of $M$. Then a local calculation yields
\begin{eqnarray*}
DF_\varphi(\psi) &=& \sum_\alpha 
  \int_M 
   \frac{\partial f}{\partial u^{(\alpha)}(x)}
   \psi^{(\alpha)}(x) dx
\\&=&
 \sum_{|\alpha|\le k} (-1)^{|\alpha|}
  \int_M \psi(x)
\left(\partial^{(\alpha)}\left(
  \frac{\partial f}{\partial u^{(\alpha)}}\right)\right) dx,
\end{eqnarray*}
where we used an integration by parts to recover
the Euler-Lagrange operator and all boundary terms vanish since $f$ is compactly supported in $x$
and $\psi\in \mathcal{D}(U)$. We have just proved that
for all open chart $U\subset M$, $\nabla F_\varphi|_U=EL(f)|_U$.
But $\nabla F_\varphi$ is intrinsically defined on $M$ therefore
so is $EL(f)$ and we have the equality $\nabla F=EL(f)$.
The unique determination of $\nabla F_\varphi$ follows from Lemma \ref{nablaFunique}.
\end{proof}

\begin{thm}\label{functionaldensitiescohomology}[Global Poincar\'e]
Assume that $M$ is a smooth, connected, oriented manifold 
without boundary.
Let $U$ be an open subset of $C^\infty(M)$ and $F:U\to \bbK$ a Bastiani smooth
\emph{local} functional.
Then the following statements are equivalent:
\begin{itemize}
\item two functions $(f_1,f_2)\in C^\infty(\mathcal{V})$ 
for $\mathcal{V}$ an open subset of the jet space $J^kM$, 
are two representations of $F$ in a neighborhood $V$ of $\varphi\in U$~:
\begin{eqnarray*}
F(\varphi+\psi) &=& F(\varphi)+\int_M f_1(j^k_x\psi) dx
\\&=&  F(\varphi)+\int_M f_2(j^k_x\psi) dx
\end{eqnarray*} 
whenever $\varphi+\psi\in V$
\item for all $\psi\in V-\varphi$~:
\begin{equation}
f_1(j^k_x\psi)dx-f_2(j^k_x\psi)dx=d\beta(j^{2k}_x\psi),
\end{equation}
where $\beta(j^{2k}_x\psi)\in \Omega_c^{d-1}(M)$ is a differential form of degree $d-1$
whose value at a point $x$ depends only on the $2k$-jet of $\psi$ at $x$.
\end{itemize}
\end{thm}
Let us stress that we do not need to constraint the topology of $M$ 
in the above
Theorem only the compactness of the support of $f_i(j^k_x\psi)dx,i\in\{1,2\}$ really matters.
\begin{proof}
One sense of the equivalence is trivial since the integral of a compactly supported exact form on $M$ always vanishes.
By Proposition \ref{prop:NeDependQueDesJets}, we know that the map $\psi\in V-\varphi\mapsto j^k\psi$ has an open image in $J^k(M)$ denoted $\mathcal{V}$ and we only need $(f_1,f_2)$ to be defined on $\mathcal{V}$.
We denote by $(x,u,u^\alpha)_{\vert\alpha\vert\leqslant k}$ where $\alpha$ are multi--indices, some local coordinates
on the jet bundle $J^k(\mathbb{R}^d)$.
We use the vertical Euler vector field $\rho=\sum u^{(\alpha)} \frac{\partial}{\partial u^{(\alpha)}} $
on the bundle $J^k(\mathbb{R}^d)$.
For every multiindex $(\alpha)=(\alpha_1\dots \alpha_p),
\alpha_i\in \{1,\dots,d\}$, 
introduce the operators
$\partial^{(\alpha)}=\partial^{\alpha_1}\dots \partial^{\alpha_p}$
where $\partial^i=\frac{\partial}{\partial x^i}+\sum_{\alpha}u^{(\alpha i)}\frac{\partial}{\partial u^{(\alpha)}}$ and the Euler--Lagrange operator
reads
$EL=u\frac{\partial}{\partial u}+\sum_{\alpha} (-1)^{\vert \alpha\vert}\partial^{(\alpha)}u^{(\alpha)}\frac{\partial }{\partial u^{(\alpha)}}$.

We shall prove two related identities, in local chart
\begin{eqnarray}\label{firstidentitypoincare}
\left(\rho f\right)(j^{k}\psi)dx&=&(uEL(f))(j^{k}\psi)dx
\nonumber\\&& +d\left(\sum_{\mu=1}^d j_\mu(j^{2k}\psi) 
\frac{\partial}{\partial{x^\mu}}\lrcorner dx \right)\\
f(j^k(\psi_1+\psi_2))dx&=&f(j^{k}\psi_1)dx
\nonumber\\&&\hspace{-15mm} +\int_0^1 dt \psi_2
EL(f)(j^{k}(\psi_1+t\psi_2))dx\nonumber\\
&&\hspace{-15mm}
+d\left(\int_0^1\frac{dt}{t} j_\mu(j^{2k}(\psi_1+t\psi_2))\partial_{x^\mu}\lrcorner dx\right).
\label{secondidentitypoincare}
\end{eqnarray}

For all $(f,g)\in C^\infty(J^k\mathbb{R}^d)^2$ and all 
multiindices $\alpha$, the generalized Leibniz-like identity holds true:
\begin{eqnarray*}
(\partial^{\alpha_1}\dots \partial^{\alpha_p}f) g &=&
(-1)^pf(\partial^{\alpha_p}\dots \partial^{\alpha_1}g)
\\&&\hspace{-2cm}
+
\sum_{i=1}^p(-1)^{i+1}
\partial^{\alpha_i}\big((\partial^{\alpha_{i+1}}\dots \partial^{\alpha_p}f)
\partial^{\alpha_{i-1}} \dots \partial^{\alpha_1}g \big)
\end{eqnarray*}
where the second term is a sum of total derivatives.
Using this we derive the following
key identity which is valid on jet spaces. For all $fdx\in C^\infty(J^k\mathbb{R}^d)\otimes \Omega^d(\mathbb{R}^d)$~:
\begin{eqnarray*}
\left(\rho f\right)dx&=&\sum u^{(\alpha)} \frac{\partial f}{\partial u^{(\alpha)}}dx=u \frac{\partial f}{\partial u}+\sum_{\vert\alpha\vert\geqslant 1} u^{(\alpha)} \frac{\partial f}{\partial u^{(\alpha)}}dx\\
&=&uEL(f)dx+\sum_{\mu=1}^d\partial^\mu j_\mu(j^{2k}\psi) dx
\\&=&
uEL(f)dx+d\left(\sum_{\mu=1}^d j_\mu(j^{2k}\psi) \frac{\partial}{\partial_{x^\mu}}\lrcorner dx \right)
\end{eqnarray*}
where $j_\mu\in C^\infty(J^{2k}\mathbb{R}^d)$ is a local functional.

To prove the second identity, we shall use the fundamental Theorem of calculus
and the first identity~:
\begin{eqnarray*}
f(j^k(\psi_1+\psi_2))dx&=&f(j^{k}\psi_1)dx
\\&&+\int_0^1 \frac{dt}{t} \left(\rho f\right)(j^{k}(\psi_1+t\psi_2))dx\\
&=&f(j^{k}\psi_1)dx
\\&&+\int_0^1 dt \psi_2 EL(f)(j^{k}(\psi_1+t\psi_2))dx\\
&+&d\left(\int_0^1\frac{dt}{t} j_\mu(j^{2k}(\psi_1+t\psi_2))\partial_{x^\mu}\lrcorner dx\right)
\end{eqnarray*}

To prove the claim of the Lemma is equivalent to show 
the following statement: if
a local functional $F$ is locally constant i.e.
$F(\varphi+\psi)=F(\varphi)$
whenever $\varphi+\psi\in V$, then
$F(\varphi+\psi) = F(\varphi)+\int_M d\beta(j^{2k}_x\psi)$
and $\beta(j^{2k}_x\psi)\in \Omega^{n-1}_c(M)$ is a compactly supported $n-1$ form. 
For all $\psi$ in $V-\varphi$, 
\begin{eqnarray*}
F(\varphi+t\psi)=F(\varphi)\implies
\int_{\mathbb{R}^d} \int_0^1 \frac{dt}{t} (\rho f)(j_x^k(t\psi))dx=0, \\
\implies \int_{\mathbb{R}^d} \int_0^1 dt \left(\psi EL(f)(j_x^k(t\psi))dx  \right) =0.
\end{eqnarray*}
This means that $EL(f)=0$ therefore on any \textbf{open chart} $U$ ($U$
is contractible), Eq.~\eqref{secondidentitypoincare} yields 
$$f(j_x^p(\psi))=f(0)+d\int_0^1dt \left(\sum_\mu  j_\mu(j_x^{2p}(t\psi))\partial_{x^\mu}\lrcorner dx \right).$$ 

We want to prove that 
$EL(f)=0 \implies f(j_x^k(\psi))dx-f(0)dx|_{M_{p+1}}=d\beta(j^{2k}\psi)$ where $\beta\in C^\infty(J^{2p}(M))$
knowing that this holds true on any local chart, and that $EL(f)=0$ is equivalent to assuming that $F(\varphi+\psi):= \int_M f(j^k_x \psi) dx $ is locally constant.
We cover $M$ by some countable
union $\cup_{i\in \mathbb{N}}U_i$ of contractible 
open charts such that every element $x\in M$ belongs
to a \textbf{finite number} of charts $U_i$, set $M_p=\left(U_1\cup\dots\cup U_p\right)$
and we arrange the cover in such a way that
$M_p\cap U_{p+1}\neq \emptyset$ for all $p$ which is always possible.
Assume by induction on $p$
that $EL(f)=0$ and $\text{supp }(f)\subset M_p$
implies 
\begin{eqnarray*}
f(j_x^k(\psi))dx-f(0)dx|_{M_p}=d\beta(j^{2k}\psi),
\end{eqnarray*}
where $\beta\in C^\infty(J^{2p}(M))\otimes \Omega_c^{d-1}(M_p)$.

We want to prove that 
$EL(f)=0, \text{supp }(f)\subset M_{p+1}\implies f(j_x^k(\psi))dx-f(0)dx|_{M_{p+1}}=d\beta(j^{2k}\psi)$ where $\beta\in C^\infty(J^{2p}(M))\otimes \Omega_c^{d-1}(M_{p+1})$.
Choose a partition of unity
$(\chi,1-\chi)$ subordinated to
$M_p\cup U_{p+1}$, the key idea is to decompose
the variation $\psi$ of the background field $\varphi$ as the sum of two components
$\chi\psi+(1-\chi)\psi$ where $\chi\psi$ (resp $(1-\chi)\psi$) vanishes outside
$U_p$ (resp $U_{p+1}$) which yields~:
\begin{eqnarray*}
f(j^k\psi)&=&f(j^k(\chi\psi+(1-\chi)\psi))-
f(j^k((1-\chi)\psi))
\\&& +f(j^k((1-\chi)\psi))-f(0)+f(0).
\end{eqnarray*}
The second idea is to 
note that for every fixed $\psi$, the new functional
$$ \phi\mapsto \tilde{f}(j^k\phi)=f(j^k(\chi\phi+(1-\chi)\psi))-f(j^k((1-\chi)\psi))$$ has \textbf{trivial} Euler--Lagrange equation
$EL(\tilde{f})(j^k\phi)=EL(f)(j^k(\chi\phi+(1-\chi)\psi))=0 $ since $EL(f)=0$ and its support is contained in $M_p$.
Therefore~:
\begin{eqnarray*}
f(j^k\psi)&=&\tilde{f}(j^k\psi)+
f(j^k((1-\chi)\psi))-f(0)+f(0)
\\
&=& d\tilde{\beta}(j^{2k}\psi)+\underbrace{f(j^k((1-\chi)\psi))-f(0)}+f(0)
\end{eqnarray*}
by the inductive assumption. To treat the term under brace, define a new functional
$$\psi \mapsto g(j^k\psi)=f(j^k((1-\chi)\psi))-f(0)$$
whose support is contained in $U_{p+1}$ and whose Euler-Lagrange equation vanishes, $EL(g)=0$ again by the fact that $EL(f)=0$. Since $U_{p+1}$ is contractible
we know that $f(j^k((1-\chi)\psi))-f(0)=d\alpha(j^{2k}\psi)$ where $\alpha\in C^\infty(J^{2k}M)\otimes \Omega^{d-1}_c(U_{p+1})$
and therefore
we found that
$$f(j^k\psi)dx=d\beta(j^{2k}\psi)+f(0)dx$$
$\beta\in C^\infty(J^{2k}M)\otimes \Omega^{d-1}_c(M_{p+1}) $.
Therefore for all $\psi\in V-\varphi$,
$f(j_x^k(\psi))dx=d\beta(j^{2k}\psi)+f(0)dx$. Now we conclude by using the fact that $F$ is a constant functional thus $0=F(\varphi+\psi)-F(\varphi)=\int_M\left(f(0)dx+d\beta(j^{2k}\psi)\right)=\int_M f(0)dx$. But $f(0)dx$ is a top form in $\Omega^d_c(M)$ which does not depend on $\psi$ and whose integral over $M$ vanishes hence
$f(0)dx=dk$ for some $k\in \Omega_c^{d-1}(M)$ since $H_c^d(M,\mathbb{R})\simeq\mathbb{R}$ for the top de Rham cohomology with compact support when $M$ is \textbf{connected}
[\onlinecite[Theorem 17.30 p.~454]{LeeSmooth}].
\end{proof}
The next Theorem summarizes the above results~:
\begin{thm}
Let $U$ be an open subset of $C^\infty(M)$ and $F:U\to \bbK$ a Bastiani smooth
\emph{local} functional.
For $\varphi\in U$, if there is an integer $k$, a neighborhood 
$V$ of $\varphi$, an open subset $\mathcal{V}$ of $J^kM$ and $f\in C^\infty(\mathcal{V})$
such that
$x\mapsto f(j^k\psi_x)$ compactly supported and
\begin{eqnarray*}
F(\varphi+\psi) = F(\varphi)+\int_M f(j^k_x\psi) dx
\end{eqnarray*} 
whenever $\varphi+\psi\in V$ then
in every local chart 
\begin{eqnarray}
\nabla F_\varphi=EL(f)(j^k\psi)
\end{eqnarray}
where 
$EL(f)(\psi)=\sum_{\vert \alpha\vert\leqslant k}(-1)^{|\alpha|}\left(\partial^{(\alpha)}\left(
  \frac{\partial f}{\partial u^{(\alpha)}}\right)\right)(j^k_x\psi) $
is the Euler--Lagrange operator and $EL(f)(\psi)$ is uniquely
determined by $F$.

Furthermore, we find that~:
\begin{eqnarray*}
F(\varphi+\psi) &=& F(\varphi)+\int_M f(j^k_x\psi) dx
\nonumber\\&=& F(\varphi)+\int_M \left(\int_0^1 dt\psi EL(f)(t\psi)\psi\right) dx
\end{eqnarray*}
where
$f(j^k_x\psi)-\left(\int_0^1 dt \psi EL(f)(t\psi)\psi\right) dx=d\beta(j^{2k}_x\psi)$ and $\beta(j^{2k}_x\psi)\in \Omega^{d-1}_c(M)$ is a compactly supported $d-1$ form.   
\end{thm}

\subsubsection{Explicit forms}
In this section we derive the explicit expression 
of $\nabla F_\varphi$ and $F^{(\alpha)}(\varphi)$ in
terms of $f$ when $M=\bbR^d$. Since the general expression
is not very illuminating, 
let us start with the following simple example:
\begin{eqnarray*}
F(\varphi) &=& \int_M h(x) \varphi^4(x) + g^{\mu\nu}(x)
  \partial_\mu \varphi(x)\partial_\nu\varphi(x) dx,
\end{eqnarray*}
where $h$ and $g^{\mu\nu}$ are smooth and compactly supported
and $g^{\mu\nu}$ is symmetric.
We compute
\begin{eqnarray*}
DF_\varphi(u) &=& 2 \int_M dx 2h(x) \varphi^3(x)u(x) + g^{\mu\nu}(x)
  \partial_\mu \varphi(x)\partial_\nu u(x)
\\&=&
2 \int_M dx \Big(2h(x) \varphi^3(x)- \partial_\nu \big(g^{\mu\nu}(x)
  \partial_\mu \varphi(x)\big)\Big)u(x),
\end{eqnarray*}
where we used integration by parts.
Thus,
\begin{eqnarray*}
\nabla F_\varphi(x) &=& 4h(x) \varphi^3(x)
- 2\partial_\nu \big(g^{\mu\nu}(x)
  \partial_\mu \varphi(x)\big).
\end{eqnarray*}
Moreover,
\begin{eqnarray*}
D^2F_\varphi(u,v) &=&
2 \int_M dx u(x) \Big(6h(x) \varphi^2(x)v(x)
\\&& - \partial_\nu \big(g^{\mu\nu}(x)
  \partial_\mu v(x)\big)\Big).
\end{eqnarray*}
To write this as a distribution, we need to integrate
over two variables:
\begin{eqnarray*}
D^2F_\varphi(u,v) &=&
2 \int_{M^2} dx dy u(x)\delta(x-y)
\Big(6h(y) \varphi^2(y)v(y)
\\&& - \partial_\nu \big(g^{\mu\nu}(y)
  \partial_\mu v(y)\big)\Big).
\end{eqnarray*}
Now we can use integration by parts over $y$
to recover $v(y)$:
\begin{eqnarray*}
D^2F_\varphi(u,v) &=&
\sum_\alpha \int_{M^2} dx dy u(x)v(y) f^\alpha(\varphi)(y)
  \partial^\alpha_y \delta(x-y),
\end{eqnarray*}
where the non-zero $f^\alpha(\varphi)$ are
\begin{eqnarray*}
f^0(\varphi)(y) &=& 12 \varphi^2(y),\\
f^{\mu}(\varphi)(y) &=& - \partial_\nu g^{\mu\nu}(y),\\
f^{\mu\nu}(\varphi)(y) &=& - g^{\mu\nu}(y).
\end{eqnarray*}

More generally
\begin{prop}
If 
\begin{eqnarray*}
F(\varphi) &=& \int_M f(\varphi^{(\alpha)}(x)) dx,
\end{eqnarray*}
then
\begin{eqnarray*}
F^\alpha(\varphi)(x) &=& \sum_{\beta\le \gamma}(-1)^{|\beta|}
  \binom{\beta}{\gamma} \partial_{y^{\beta-\gamma}}
  \frac{\partial^2 f}
   {\partial \varphi^{(\alpha-\gamma)}(x)
   \partial \varphi^{(\beta)}(x)}.
\end{eqnarray*}
\end{prop}
\begin{proof}
The proof is a straightforward generalization of the example.
Indeed,
\begin{eqnarray*}
DF_\varphi(u) &=& \sum_\alpha 
  \int_M 
   \frac{\partial f}{\partial \varphi^{(\alpha)}(x)}
   u^{(\alpha)}(x) dx
\\&=&
 \sum_\alpha (-1)^{|\alpha|}
  \int_M u(x) \frac{\partial^{|\alpha|}}{\partial x^\alpha}
  \frac{\partial f}{\partial \varphi^{(\alpha)}(x)} dx,
\end{eqnarray*}
where we used an integration by parts to recover
the Euler-Lagrange operator.
The second derivative is
\begin{eqnarray*}
D^2F_\varphi(u,v) &=& 
 \sum_{\alpha\beta} (-1)^{|\alpha|}
  \int_M u(x) 
\\&& \frac{\partial^{|\alpha|}}{\partial x^\alpha}
  \Big(\frac{\partial f^2}
  {\partial \varphi^{(\alpha)}(x)\partial \varphi^{(\beta)}(x)} 
  v^{(\beta)}(x)\Big) dx.
\end{eqnarray*}
We write this as a double integral
\begin{eqnarray*}
D^2F_\varphi(u,v) &=& 
 \sum_{\alpha\beta} (-1)^{|\alpha|}
  \int_{M^2} u(x) \delta(x-y)
\\&&
   \partial_{y^\alpha}
  \Big(\frac{\partial f^2}
  {\partial \varphi^{(\alpha)}(y)\partial \varphi^{(\beta)}(y)} 
  v^{(\beta)}(y)\Big) dx dy.
\end{eqnarray*}
A first integration by parts gives us
\begin{eqnarray*}
D^2F_\varphi(u,v) &=& 
 \sum_{\alpha\beta}
  \int_{M^2} u(x) 
  \Big(\frac{\partial f^2}
  {\partial \varphi^{(\alpha)}(y)\partial \varphi^{(\beta)}(y)} 
  v^{(\beta)}(y)\Big)
\\&&
   \partial_{y^\alpha} \delta(x-y) dx dy.
\end{eqnarray*}
A second integration by parts isolates $v(y)$:
\begin{eqnarray*}
D^2F_\varphi(u,v) &=& 
 \sum_{\alpha\beta} (-1)^{|\beta|}\sum_{\gamma\le\beta}
  \binom{\beta}{\gamma}
  \int_{M^2}  dx dy u(x) v(y)
\\&&
  \Big(\partial_{y^{\beta-\gamma}} \frac{\partial f^2}
  {\partial \varphi^{(\alpha)}(y)\partial \varphi^{(\beta)}(y)} 
  \Big)
   \partial_{y^{\alpha+\gamma}} \delta(x-y).
\end{eqnarray*}
\end{proof}
If we calculate higher differentials
$D^kF_\varphi(u_1,\dots, u_k)$ we see that we always
obtain products of smooth functions by derivatives
of products of delta functions. This shows that
the wavefront set of $F^{(k)}_\varphi$ is in
the conormal $C_k$.

\section{Peetre theorem for local and multilocal functionals}
In this section, we propose an alternative
characterization of local
functionals
in terms of a nonlinear Peetre theorem.
We do not characterize the locality of the action $F$ but the locality
of the Lagrangian density, that we denoted
$\nabla F$ in the previous section.
We first state our theorems for local functionals, and then
we prove them for the case of multilocal functionals, which are 
a natural generalization of local functionals in quantum field
theory.
Our proof is inspired by recent works on the Peetre 
theorem~\cite{Navarro--Sancho,Pflaum--Brasselet}, 
however it is formulated in the
language
of Bastiani smoothness and uses
simpler assumptions than Slov\'ak's paper~\cite{Slovak-88}.

\subsection{Peetre theorem for local functionals}
Let $\Omega$ be some open set in a manifold
$M$. We first begin with an alternative definition
of a local map from $C^\infty(\Omega)$
to itself, that we call \emph{Peetre local}.
\begin{dfn}
\label{deflocalmap}
A map $F:C^\infty(\Omega)\to C^\infty(\Omega)$ is 
\emph{Peetre local} if
for every $x\in\Omega$, if
$\varphi_1=\varphi_2$ on some neighborhood of $x$
then $F(\varphi_1)(x)=F(\varphi_2)(x)$.
\end{dfn}
The relation with the additivity condition is given by
\begin{prop}
Let $F:C^\infty(\Omega)\to C^\infty(\Omega)$ be Peetre local.
For every $(\varphi_1,\varphi_2)\in C^\infty(\Omega)^2$
if $\text{supp }\varphi_1$ and  $\text{supp }\varphi_2$ do not meet
then for every $x\in \Omega$ and for all $\varphi$,
\begin{eqnarray}
F(\varphi_1+\varphi_2+\varphi)(x) &=& F(\varphi_1+\varphi)(x)
\nonumber\\&&+
   F(\varphi_2+\varphi)(x)-F(\varphi)(x).
\label{locequation}
\end{eqnarray}
\end{prop}
\begin{proof}
If $x\notin \left(\text{supp }\varphi_1\cup \text{supp
}\varphi_2\right)$
then $\varphi_1=\varphi_2=0$ in some neighborhood
of $x$, it follows that
$F(\varphi_1+\varphi_2+\varphi)(x)=F(0+0+\varphi)(x)=F(\varphi)(x)$
and
$F(\varphi_1+\varphi)(x)+F(\varphi_2+\varphi)(x)-F(\varphi)(x)
  =2F(\varphi)(x)-F(\varphi)(x)=F(\varphi)(x)$ 
hence Eq.~(\ref{locequation}) holds true.

If $x\in\text{supp }\varphi_1$ then necessarily
there is some neighborhood $U$ of $x$ on which $\varphi_2|_U=0$
hence $\varphi_1+\varphi_2+\varphi|_U=\varphi_1+\varphi|_U$ and 
$F(\varphi_1+\varphi_2+\varphi)(x)=F(\varphi_1+\varphi)(x)$.
Also 
$F(\varphi_1+\varphi)(x)+F(\varphi_2+\varphi)(x)-F(\varphi)(x)
=F(\varphi_1+\varphi)(x)+F(\varphi)(x)-F(\varphi)(x)=F(\varphi_1+\varphi)(x)$
hence again Eq.~(\ref{locequation}) holds true.
The case where $x\in \text{supp }\varphi_2$ can be treated
by similar methods which yields the final result.
\end{proof}

The Peetre theorem for local functionals is
\begin{thm}
Let $F:C^\infty(\Omega) \to C^\infty(\Omega)$ be a Bastiani
smooth Peetre local map.
Then, for every $\varphi\in C^\infty(\Omega)$ there is a neighborhood
$V$ of $\varphi$ in $C^\infty(\Omega)$ and an integer $k$ such that for
all $\psi$ such that $\varphi+\psi\in V$,
$F(\varphi+\psi)(x)=c(j^k\psi_x)$
for some smooth function $c$ on $J^k\Omega$. 
\end{thm}

In other words, if $F$ is a Bastiani smooth Peetre
local map, then for every $g\in \calD(M)$,
$\int_M F(\varphi)g$ is a Bastiani smooth local
map in the sense of the rest of the paper.
This relation between a priori different concepts
of locality strongly supports the idea that
our definition is a natural one.

If $F$ is only assumed to be a continuous local map, then 
a similar theorem
exists for which the function $c$ is not necessarily smooth.
These theorems are proved in the next section for the
more general case of multilocal functionals.

\subsection{Multilocal functionals and first Peetre theorem}

By generalizing Definition~\ref{deflocalmap} of local maps, 
we can define multilocal maps. These maps appear naturally
in quantum field theory as the product of several
Lagrangian densities $\calL(x_1)\dots\calL(x_k)$.
\begin{dfn}\label{klocaldef}
Let $k$ be an integer.
A map $F:C^\infty(\Omega)\to C^\infty(\Omega^k)$ is $k$-local if
for every $(x_1,\dots,x_k)\in\Omega^k$, if
$\varphi_1=\varphi_2$ on some neighborhood of $\{x_1,\dots,x_k\}\subset
\Omega$
then $F(\varphi_1)(x_1,\dots,x_k)=F(\varphi_2)(x_1,\dots,x_k)$.
\end{dfn}
The \emph{multilocal maps} are the maps that are $k$-local for
some $k$. 
We emphasize that \textbf{Peetre local} maps in the sense of definition \ref{deflocalmap} correspond with
\textbf{$1$-local maps} in the above sense.
For $M$ a smooth manifold, we denote
by $J^pM^{\boxtimes k}$ the bundle over $M^k$ whose fiber
over a $k$-tuple of points $(x_1,\dots,x_k)\in M^k$ is
$J^pM_{x_1}\times \dots\times J^pM_{x_k}$.

\begin{thm}
Let $F:C^\infty(\Omega) \to C^\infty(\Omega^k)$ be a continuous 
$k$-local map.
Then, for every $\varphi\in C^\infty(\Omega)$ there is a neighborhood
$V$ of $\varphi$ in $C^\infty(\Omega)$, $p\in\mathbb{N}$ such that for
all
$\psi$ such that $\varphi+\psi\in V$,
$$F(\varphi+\psi)(x_1,\dots,x_k)=c(j^p\psi_{x_1},\dots,j^p\psi_{x_k})$$
for some function
$c: J^pM^{\boxtimes k}|_{
(M^k\setminus D_k)}\to M^k$, where $M^k\setminus D_k$ denotes the configuration
space $M^k$ minus all diagonals.
\end{thm}
\begin{proof}
Fix a $k$--tuple of points $(x_1,\dots,x_k)\in\Omega^k$ and some
compact neighborhood $K$ of $(x_1,\dots,x_k)$ in $\Omega^k$.
Continuity of $F$
implies that for all $\varepsilon>0$, there exists $\eta>0$ and a
seminorm 
$\pi_{m,K^\prime}$ of $C^\infty(\Omega)$
such that
$\pi_{m,K^\prime}(\varphi_1-\varphi_2)\leqslant\eta$ implies
\begin{eqnarray*}
\sup_{(y_1,\dots,y_k)\in K} \vert
F(\varphi_1)(y_1,\dots,y_k)-F(\varphi_2)(y_1,\dots,y_k)\vert\leqslant
\varepsilon. 
\end{eqnarray*}

Assume that
$(\varphi_1,\varphi_2)$ 
have same ($m+1)$-jets at $\{x_1,\dots,x_k\}$. Let $(\chi_\lambda)_\lambda$ 
be the family of  compactly supported cut-off
functions equal to $1$ 
in some neighborhood of $X=\{x_1,\dots,x_k\}$ defined in lemma 
\ref{Malgrangestyletechnicallemma}. 
It follows that
$\varphi_{1,\lambda}=\varphi_1\chi_\lambda$ (resp.
$\varphi_{2,\lambda}=\varphi_2\chi_\lambda$)
coincides with $\varphi_1$ (resp. $\varphi_2$) near 
$\{x_1,\dots,x_k\}$. Hence, for all
$\lambda >0$,
$F(\varphi_{1,\lambda})(x_1,\dots,x_k)=F(\varphi_1)(x_1,\dots,x_k)$
and $F(\varphi_{2,\lambda})(x_1,\dots,x_k)=F(\varphi_2)(x_1,\dots,x_k)$.
Set $\varepsilon_n=\frac{1}{2^n}$ then there exists
$\eta_n$ such that
$\pi_{m,K^\prime}(\psi_1-\psi_2)\leqslant\eta_n$ implies
\begin{eqnarray*}
\sup_{(y_1,\dots,y_k)\in K} \vert
F(\psi_1)(y_1,\dots,y_k)-F(\psi_2)(y_1,\dots,y_k)\vert\leqslant
\frac{1}{2^n}.
\end{eqnarray*}
Therefore it suffices to find some
sequence $\lambda_n\rightarrow 0$
such that
$\pi_{m,K^\prime}(\varphi_{1,\lambda_n}-\varphi_{2,\lambda_n})\leqslant
\eta_n$.
Since $\varphi_1-\varphi_2$ vanishes at order
$m+1$ on the set $X=\{x_1,\dots,x_k\}$, 
Lemma~\ref{Malgrangestyletechnicallemma} yields the estimate
\begin{eqnarray*}
\vert\pi_{m,K^\prime}(\varphi_{1,\lambda}-\varphi_{2,\lambda})\vert
\leqslant \tilde{C}\lambda \pi_{m+1,K}\left(\varphi_1-\varphi_2\right),
\end{eqnarray*}
which implies that
\begin{eqnarray*}
\underset{\lambda\rightarrow
0}{\lim}\pi_{m,K^\prime}(\varphi_{1,\lambda}-\varphi_{2,\lambda})= 
\underset{\lambda\rightarrow
0}{\lim}\pi_{m,K^\prime}((\varphi_1-\varphi_2)\chi_\lambda)=0.
\end{eqnarray*}

Finally, we obtain that if $\varphi_1,\varphi_2$ have same $(m+1)$-jet at
$X=\{x_1,\dots,x_k\}$ then for all $n>0$:
\begin{eqnarray*}
\vert F(\varphi_1)(x_1,\dots,x_k)-F(\varphi_2)(x_1,\dots,x_k)\vert &=&
\\&&\hspace{-6cm} \vert
F(\varphi_{1,\lambda_n})(x_1,\dots,x_k)- 
F(\varphi_{2,\lambda_n})(x_1,\dots,x_k)\vert\leqslant \frac{1}{2^n}
\end{eqnarray*}
which implies
$F(\varphi_1)(x_1,\dots,x_k)=F(\varphi_2)(x_1,\dots,x_k)$.
\end{proof}

\begin{lem}\label{Malgrangestyletechnicallemma}
Let $X$ be any closed subset of $\mathbb{R}^{d}$.
Let
$\mathcal{I}^{m+1}(X,\mathbb{R}^d)$ denote
the closed ideal
of functions of regularity $C^{m+1}$ 
which vanish at order $m+1$ on $X$.
Then there is a function $\chi_\lambda\in C^\infty(\mathbb{R}^{d})$
parametrized by $\lambda\in(0,1]$
s.t. $\chi_\lambda=1$ (resp $\chi_\lambda=0$) when
$d(x,X)\leqslant\frac{\lambda}{8}$
(resp $d(x,X)\geqslant\lambda$) such that for all compact subset 
$K\subset \mathbb{R}^{d}$,
there is a constant $\tilde{C}$ such that, for every $\lambda\in(0,1]$
and every $\varphi \in \mathcal{I}^{m+1}(X,\mathbb{R}^{d})$
\begin{eqnarray}
\pi_{m,K}\left(\chi_\lambda \varphi\right)\leqslant \tilde{C}
\lambda  \pi_{m+1,K\cap\{d(x,X)\leqslant \lambda\}}
\left(\varphi\right). 
\end{eqnarray} 
\end{lem}
\begin{proof}
Choose $\phi\geqslant 0$ s.t. $\int_{\mathbb{R}^d}\phi(x)d^dx=1$ and 
$\phi=0$ if $\vert x\vert\geqslant \frac{3}{8}$.
Then set $\phi_\lambda=\lambda^{-d}\phi(\lambda^{-1}.)$ 
and set
$\alpha_\lambda$ to be the characteristic function
of the set $\{x \text{ s.t. }d(x,X)\leqslant \frac{\lambda}{2} \}$ 
then the convolution
product 
$\chi_\lambda=\phi_\lambda*\alpha_\lambda$ 
satisfies $\chi_\lambda(x)=1$
if $d(x,X)\leqslant \frac{\lambda}{8}$
and $\chi_\lambda(x)=0$
if $d(x,X)\geqslant \lambda$.
Since by Leibniz rule
one has $$\partial^\alpha (\chi_\lambda\varphi)(x)=\underset{\vert
k\vert\leqslant \vert\alpha\vert}\sum\left(
\begin{array}{c}
\alpha\\
k
\end{array}
\right) \partial^k\chi_\lambda\partial^{\alpha-k}\varphi(x), $$
it suffices to estimate
each term
$\partial^k\chi_\lambda\partial^{\alpha-k}\varphi(x)$
of the above sum.
For every multi-index $k$, there is some constant
$C_k$ such that
$\forall x\in \mathbb{R}^{d}\setminus X, \vert \partial_x^k 
\chi_\lambda \vert\leqslant
\frac{C_k}{\lambda^{\vert k\vert}}$
and $\text{supp }\partial_x^k 
\chi_\lambda\subset\{d(x,X)\leqslant \lambda \}$. 
Therefore for all $\varphi\in\mathcal{I}^{m+1}
(X,\mathbb{R}^{d})$, 
for all
$x\in \text{supp }\partial_x^k \chi_\lambda\partial^{\alpha-k}\varphi$,
for $y\in X$ such that
$d(x,X)=\vert x-y\vert$, we find that $\partial^{\alpha-k}\varphi$
vanishes at $y$ at order $|k|+1$. Indeed $\varphi$ vanishes at order $m+1$ hence $\partial^{\alpha-k}\varphi$
vanishes at order $m+1-\vert\alpha\vert + k\geqslant k+1$ since $\vert\alpha\vert\leqslant m$. Therefore:
$$\partial_x^{\alpha-k}\varphi(x) =\sum_{\vert\beta\vert=|k|+1}
(x-y)^\beta R_\beta(x),$$
where the right hand side
is just the integral remainder
in Taylor's expansion of $\partial^{\alpha-k}\varphi$ around $y$.
Hence:
$$ \vert\partial^k\chi_\lambda\partial^{\alpha-k}\varphi(x)   \vert
\leqslant \frac{C_k}{\lambda^{\vert k\vert}}\sum_{\vert\beta\vert=\vert
k\vert+1}
\vert (x-y)^\beta R_\beta(x)\vert  .$$
 It is easy to see 
that $R_\beta$ only
depends on the jets 
of $\varphi$ of order $\leqslant m+1$.
Hence
$$\vert\partial^k\chi_\lambda\partial^{\alpha-k}\varphi(x)   \vert
\leqslant 
C_{k} \lambda \sup_{x\in K, d(x,X)\leqslant
\lambda}\sum_{\vert\beta\vert=\vert k\vert+1}
\vert R_\beta(x)\vert$$
and the conclusion follows easily.
\end{proof}

\subsection{The second Peetre Theorem}

\begin{thm}
Let $F:C^\infty(\Omega) \to C^\infty(\Omega^k)$ be a
Bastiani smooth $k$-local map.
Then, for every $\varphi\in C^\infty(\Omega)$ there is a neighborhood
$V$ of $\varphi$ in $C^\infty(\Omega)$, $p\in\mathbb{N}$ such that for
all
$\psi$ such that $\varphi+\psi\in V$,
$$F(\varphi+\psi)(x_1,\dots,x_k)=c(j^p\psi_{x_1},\dots,j^p\psi_{x_k})$$
for some smooth function $c$ on $J^pM^{\boxtimes k}|_{
(M^k\setminus D_k)}$ where $M^k\setminus D_k$ denotes the configuration space
$M^k$ minus all diagonals.
\end{thm}
\begin{proof}
Without loss of generality, we may assume that
$M=\mathbb{R}^d$ and to go back to arbitrary manifolds, we use partitions of unity as in 
the proof of Lemma \ref{lem:existsSection}.
The coordinates on the jet space $J^p(\mathbb{R}^d)$ are denoted by $(x,p^\alpha)_{\vert\alpha\vert\leqslant p}$.
Let $(U_1,\dots,U_k)$ be two by two disjoint open subsets of $\mathbb{R}^d$, then
$U_1\times\dots\times U_k$ is an open subset of $(\mathbb{R}^d)^k\setminus
D_k$. 
We define
the smooth map:
$\Phi:(x_1,\dots,x_k;p_1,\dots,p_k)\in J^p(\mathbb{R}^d)^{\boxtimes
k}|_{U_1\times\dots\times U_k}\mapsto 
(\sum_{1\leqslant i\leqslant k}\frac{p_{i,\alpha}}{\alpha !}
(.-x_i)^\alpha \chi_i(.-x_i)) 
\in C^\infty(\mathbb{R}^d)$ 
where the functions $\chi_i\in C^\infty_c(\mathbb{R}^d)$ are cut--off functions equal to $1$ near $0$ and such that for all $(x_1,\dots,x_k)\in U_1\times \dots\times U_k$, the support of the functions $\chi_i(.-x_i)$ are disjoint on $\mathbb{R}^d$.
Then the map sending
$(x_1,\dots,x_k;p_1,\dots,p_k),(y_1,\dots,y_k)$
to $F(\varphi+\Phi(x_1,\dots,x_k;p_1,\dots,p_k))(y_1,\dots,y_k)$
is smooth by smoothness of $F$ and $\Phi$. Hence, its
pull--back 
on the diagonal $x_1=y_1,\dots,x_k=y_k$ is also smooth and reads
\begin{eqnarray*}
F(\varphi+\Phi(x_1,\dots,x_k;p_1,\dots,p_k))(x_1,\dots,x_k) &=&
\\&&\hspace{-3cm} c(x_1,\dots,x_k;p_1,\dots,p_k)
\end{eqnarray*}
as the composition 
of smooth functions and it follows
that $c$ is smooth on $J^pM^{\boxtimes k}|_{U_1\times\dots\times U_k}$.
\end{proof}

\section{Multi-vector fields and graded functionals}

In the quantum theory of gauge fields, especially
in the Batalin-Vilkovisky approach, it is necessary to
deal, not only with functionals as discussed above, but
also  with 
multi-vector fields
on the configuration space $E$ (assumed to be the space of sections
of some vector bundle $B$)~\cite{Fredenhagen-11}. Such multi-vector fields 
 can be seen as functionals on the graded space $T^*[1]E
 \doteq E\oplus E^*[1]$, where $E^*\doteq \Gamma(M,B^*)$ 
is the space of smooth sections. To make this notion precise, we use the  ideas presented in~\cite{Rejzner-11} and characterize the ``odd'' space $E^*[1]$ through the space of functions on it, understood as multilinear smooth, totally antisymmetric, functionals.
Then we shall make a conjectural claim on the meaning of 
locality in that context.

\subsubsection{Locality of functionals on graded space.}

We consider a graded space  $E_0\oplus E_1[1]$, where 
$E_0=\Gamma(M,B_0)$ and $E_1=\Gamma(M,B_1)$
are spaces of smooth sections of finite rank vector bundles $B_0$ and $B_1$ 
over $M$ 
respectively.
Before giving formal definitions, let us explain the
idea of our construction. 
We will first define the space $\Ocal(E_0\oplus E_1[1])$
to be space of maps from $E_0$ to $\calA$,
where
\[
\Acal\doteq\prod_{k=0}^\infty \Acal^k\doteq\prod_{k=0}^\infty 
  \Gamma'_a(M^k,B_1^{\boxtimes k})\,,
\]
satisfying an appropriate smoothness condition. Let us clarify the notation $\Gamma'_a$. We first define the
iterated wedge product of $k$ elements $u_1$,\dots,$u_k$
of the space of distributional
sections $\Gamma'(M,B_1)$ by
\begin{eqnarray*}
\langle u_1\wedge \dots \wedge u_k, h_1\otimes \dots \otimes h_k\rangle
 &=& \sum_\sigma (-1)^\sigma \langle u_1, h_{\sigma(1)}\rangle
\\&&
  \dots \langle u_k, h_{\sigma(k)}\rangle,
\end{eqnarray*}
where $h_1$,\dots,$h_k$ are sections in $\Gamma(M,B_1)$ and $\sigma$
runs over the permutations of $\{1,\dots,k\}$.
Then, the $k$-th exterior power $\Lambda^k \Gamma'(M,B_1)$ 
is the vector space of finite sums of such iterated wedge
products and $\Gamma'_a(M^k,B_1^{\boxtimes k})$ is the completion
of $\Lambda^k \Gamma'(M,B_1)$ with respect to the topology of 
$\Gamma'(M,B_1)^{\hat{\otimes}_\pi k}\cong\Gamma'(M^k,B_1^{\boxtimes k})$
where all the duals are strong.
The subscript ``$a$'' stands for antisymmetry.	

In the case of multilinear symmetric functions, we can identify
a $k$-linear map $f(h_1,\dots, h_k)$ of $k$ variables
with a polynomial map of one variable $f(h,\dots,h)$ by using
the polarization identity. There is no polarization
identity in the antisymmetric case and we must
consider a function $F: E_0\to \calA^k$ as a function
of one variable $\varphi_0$ in $E_0$ and $k$ variables 
$(h_1,\dots,h_k)$ in $E_1$ 
(or a variable in $H\in E_1^{\hat\otimes_\pi k}$).
Then, we can identify a function $F: E_0\to \calA^k$ and the
function $\tilde{F}:E_0\times E_1^{\hat\otimes_\pi k}\to\bbK$
defined by
\begin{eqnarray*}
\tilde{F}(\varphi_0;h_1\otimes\dots\otimes h_k) &=& 
F(\varphi_0)(h_1\otimes\dots\otimes h_k).
\end{eqnarray*}
This motivates the following
\begin{dfn}
Let $M$ be a smooth manifold, $(B_0,B_1)$ are smooth vector bundles
on $M$ and $E_0=\Gamma(M,B_0)$, $E_1=\Gamma(M,B_1)$ are spaces of smooth sections of the respective bundles.
We say that a function $F$ from $E_0$ to $ \calA^k$ is an element of $\Ocal^k(E_0\oplus E_1[1])$ if there exists
a Bastiani smooth map $\tilde{F}:E_0\times E_1^{\hat\otimes_\pi k}\to\bbK$
which is linear in $E_1^{\hat\otimes_\pi k}$ and \textbf{antisymmetric} 
w.r.t. the natural action of permutations on $E_1^{\hat\otimes_\pi k}$
such that~: 
\begin{eqnarray}
\label{def:withTilde}
\tilde{F}(\varphi_0;h_1\otimes\dots\otimes h_k) &=& 
F(\varphi_0)(h_1\otimes\dots\otimes h_k).
\end{eqnarray}
 We denote by $\Ocal(E_0\oplus E_1[1])$ the direct product of all $\Ocal^k(E_0\oplus E_1[1])$, over $k\in\bbN_0$ and set $\Ocal^0(E_0\oplus E_1[1])\equiv\bbK$.
\end{dfn}

Let us now discuss the notion of derivative for the type of functionals 
introduced above. Clearly, if $F$ belongs to
$\Ocal(E_0\oplus E_1[1])$, there are two natural ways to differentiate it. 
In the first instance we can differentiate  $\tilde F$ in the sense of Bastiani 
in the first variable ($\varphi\in E_0$) 
and we denote this derivative as
 \[
 D_0F_{(\varphi;u)}(g)\doteq D\tilde{F}_{(\varphi,u)}(g,0)\,,
 \]
 where $u\in E_1^{\hat\otimes_\pi k}\to\bbK$, $g\in E_0$
or 
$\frac{\delta}{\delta \varphi_0}F$.

\subsubsection{The contraction operation.}

Let us now consider contraction of the
graded part with some $h\in E_1$, sometimes referred to as derivations with respect to odd variables.
This concept is needed in order to define 
the Koszul complex and the Chevalley-Eilenberg complex in the Batalin--Vilkovisky
formalism in infinite dimension.
The definition is 
spelled out below.
\begin{dfn}\label{GradedDerivaitve}
Let $F\in\Ocal^k(E_0\oplus E_1[1])$, $h\in E_1$. 
The contraction of $F$ by $h$  is defined,
 for every integer $k > 0$ and  $u\in E_1^{\otimes k-1}$, by
\begin{eqnarray*}
 \left<\iota_h F,u\right>
    &=& \tilde{F}(h\otimes u),\\
 \text{ and }  \iota_hF&=& 0\quad\text{ if } F\in \mathcal{A}^0\,.
\end{eqnarray*}
In particular, 
$\iota_h F=\left< \tilde{F},h\right>$ if $F\in \Acal^1$.
We extend this definition to $\Acal$ by linearity.
\end{dfn}

In view of (\ref{def:withTilde}) and the definition of $\Ocal^k(E_0\oplus E_1[1])$, it is clear that $ \iota_h F \in \Ocal^{k-1}(E_0\oplus E_1[1])$ 
for all $F \in \Ocal^{k}(E_0\oplus E_1[1])$. Equation (\ref{def:withTilde}) allows also to make sense of a second 
 important operation on $\mathcal{O}(E_0\oplus E_1[1])$:
 
\begin{dfn}\label{wedge}
The wedge product $\wedge : \mathcal{O}^k(E_0\oplus E_1[1])\times \mathcal{O}^{k^\prime}(E_0\oplus E_1[1]) \to \mathcal{O}^{k+k^\prime} (E_0\oplus E_1[1])$ 
is defined by 
\begin{eqnarray*}
 \left(\widetilde{F\wedge G}\right)(u_1,\dots,u_{k+k^\prime}) &=&
\sum_\sigma\operatorname{sgn}(\sigma) \tilde{F} (u_{\sigma(1)},\dots,u_{\sigma(k)})
\\&&  \tilde{G} (u_{\sigma(k+1)}, \dots,u_{\sigma(k + k^\prime)} 
\end{eqnarray*}
(where the sum runs over $k-k^\prime$ shuffles) 
and extended by linearity on $\mathcal{O}(E_0\oplus E_1[1])\times\mathcal{O}(E_0\oplus E_1[1])$.
\end{dfn}

Again, in view of (\ref{def:withTilde}) and the definition of $\Ocal(E_0\oplus E_1[1])$, it is clear
that the wedge product of an element in  $ \mathcal{O}^{k}(E_0\oplus E_1[1]) $ with
an element in  $ \mathcal{O}^{k^\prime}(E_0\oplus E_1[1]) $ is an element in $ \mathcal{O}^{k+k^\prime} (E_0\oplus E_1[1] )$.
The contraction and wedge product satisfy the following 
relation on $\mathcal{O}(E_0\oplus E_1[1])$:
\begin{lem}\label{grLeibniz}
The contraction satisfies
the graded Leibniz rule: if $F\in \mathcal{O}^k(E_0\oplus E_1[1])$,
$G\in \mathcal{O}(E_0\oplus E_1[1])$ and $h\in E_1$, 
then
\begin{eqnarray*}
\iota_h(F\wedge G)  &=& 
(\iota_h F)  \wedge G +
(-1)^k F\wedge \iota_hG.
\end{eqnarray*}
\end{lem}

Let us now discuss the notion of support which is the
appropriate generalization of the
notion of support for graded functionals,
generalizing the definitions in Section~\ref{SuppFunct}.
\begin{dfn}\label{gradedfirstsupportdefi}
Let $F\in \mathcal{O}^k(U\oplus E_1[1])$ be a graded functional, 
with $U$ an open subset of $E_0$. 
The support of 
$F$ is defined by
$\supp F = \overline{A \cup B}$, where
\begin{eqnarray*}
A &=& \bigcup_{(h_1,\dots,h_k)\in E_1^k}\supp\left(\varphi\mapsto\left(\iota_{h_1}
 \dots\iota_{h_k}  F \right)(\varphi) \right)\\
B &=& \bigcup_{\varphi\in U,(h_1,\dots,h_{k-1})\in 
E_1^{k-1}} \supp \left( h \mapsto \left( \iota_{h_1} \dots\iota_{h_{k-1}}  F(\varphi,h )\right) \right).
\end{eqnarray*}
\end{dfn}

\subsubsection{Some conjectures on local graded functionals.}

Let $F\in\Ocal^k(E_0\oplus E_1[1])$ be such that the WF set of both
$\left(\iota_{h_1} \dots\iota_{h_k} F \right)_\varphi^{(1)}$ and  $ \iota_{h_1} \dots\iota_{h_{k-1}} F(\varphi,.)$
is empty for all $\varphi\in U$ and $(h_1,\dots,h_k)\in E_1^k$. 
We conjecture that some version of Lemmas~\ref{nablaFunique} and \ref{hardLemma} should hold in the graded
case.
The ``standard'' characterization of locality for a functional 
$F\in\Ocal^k(E_0\oplus E_1[1])$
is the requirement that $F$  is compactly supported and for each 
$(\varphi;u_1,\dots,u_k)\in E_0\times E_1^k$ there exists 
$i_0,\dots,i_k\in\mathbb{N}$ such that
\begin{equation}
F(\varphi;u_1,\dots,u_k)=\int_{M} \alpha(j^{i_0}_x(\varphi),
j^{i_1}_x(u_1),\dots,j^{i_k}_x(u_k))\,,
\end{equation}
where $\alpha$ is a density-valued function on the jet bundle. 
To conclude, we \textbf{conjecture} some graded analogue of Theorem~\ref{TheoPrincipal} whose 
formulation would be as follows~:
\\
	Let $U$ be an open subset of $E_0$ and 
	$F\in \mathcal{O}^k(U\oplus E_1[1])$ be a graded functional.
	Assume that
\begin{enumerate}
\item $F$ is additive in some suitable sense, still to be written with care 
(conceivably this would be additivity of $\tilde{F}$ as a function of several variables).
\item $\left(\iota_{h_1} \dots\iota_{h_k} F \right)_\varphi^{(1)}$ and $ \iota_{h_1} \dots\iota_{h_{k-1}} F(\varphi,.)$
have empty wave front set for all $\varphi\in U$ and $(h_1,\dots,h_k)\in E_1^k$  
and the maps 
$(\varphi,u)\mapsto \left(\iota_{h_1} \dots\iota_{h_k} F \right)_\varphi^{(1)}, \iota_{h_1} \dots\iota_{h_{k-1}} F(\varphi,.)$ are Bastiani smooth from $U\times \bigoplus_{k\in\bbN} E_1^{\hat{\otimes}_\pi k}$  to $\Gamma_c(M,B^*_0)$ and  $\Gamma_c(M,B^*_1)$, respectively. Here $B_0^*$ and $B^*_1$ denote dual bundles.
\end{enumerate}
Then, for every  $\varphi\in U$, $u\in\bigoplus_{k\in\bbN} E_1^{\hat{\otimes}_\pi k}$, there is a 
neighborhood $V$ of the origin in $E_0$, an integer $N$ and a smooth 
$\bbK$-valued function $f$ on the $N$-jet bundle such that 
\begin{equation}
F(\varphi+\psi;v_1\otimes\dots\otimes v_k)=\int_{M} \alpha(j^{i_0}_x(\psi),j^{i_1}_x(v_1),\dots,j^{i_k}_x(v_k))\,,
\end{equation}
for every $\psi\in V$ and some $i_0,\dots,i_k<N$.

\section{Acknowledgements}
We are very grateful to the  Institut Poincar\'e for making
this work possible through the grant and facilities it offered
us within the framework of its ``Research in Paris'' program.
We thank the CNRS Groupement de Recherche Renormalisation for
its financial support.
We are extremely grateful to Yoann Dabrowski 
for his very generous help during all the phases of this work.
We thank Romeo Brunetti and Pedro Ribeiro for
very useful discussions.


%
\end{document}